\documentclass[11pt]{article}
\usepackage[left=1in, right=1in, top=1in, bottom=1in, margin=1in]{geometry}

\usepackage{algorithm}
\usepackage{tcolorbox}
\usepackage{amsmath,amssymb,xspace,graphicx,relsize,bm,xcolor,breqn,algpseudocode,multirow}

\usepackage{tikz}
\usetikzlibrary{quantikz2}
\usepackage{adjustbox}
\usepackage{booktabs,tabularx,array}

\usepackage{tcolorbox}
\usepackage{ninecolors}

\usepackage[margin=1in]{geometry}

\newcommand{\M}{\mathcal{M}}

\usepackage{amsthm}
\usepackage{dsfont}
\usepackage{array}
\usepackage{makecell}

\newcommand{\poly}{\ensuremath{\mathsf{poly}}}
\newcommand{\polylog}{\ensuremath{\mathsf{polylog}}}

\usepackage{upgreek}
\usepackage{enumerate}

\usepackage[usestackEOL]{stackengine}
\usepackage{thm-restate,mathrsfs}
\usepackage{enumerate}
\usepackage{array}
\usepackage{parskip}
\def\01{\{0,1\}}
\usepackage[margin=1in]{geometry}

\definecolor{citegreen}{HTML}{208054}
\definecolor{citeblue}{HTML}{0055cc}

\usepackage[hypertexnames=false]{hyperref}
\NineColors{saturation=high}
\hypersetup{
    breaklinks=true,   % splits links across lines
    colorlinks=true, % false: boxed links; true: colored links
    linkcolor=blue3, % color of internal links
    citecolor=cyan, % color of links to bibliography
    urlcolor=blue3, % color of external links
}
\usepackage[noabbrev,capitalise]{cleveref}

\newcommand{\be}{\begin{equation}}
\newcommand{\ee}{\end{equation}}
\newcommand{\ba}{\begin{array}}
\newcommand{\ea}{\end{array}}
\newcommand{\bea}{\begin{eqnarray}}
\newcommand{\eea}{\end{eqnarray}}

\usepackage{mathtools}
\DeclarePairedDelimiter\ceil{\lceil}{\rceil}
\DeclarePairedDelimiter\floor{\lfloor}{\rfloor}

\newcommand{\ra}{\rangle}

\newcommand{\expnorm}{\textsf{C}}
\newcommand{\ODE}{\textsf{ODE}}
\newcommand{\VIDE}{\textsf{VIDE}}
\newcommand{\SOE}{\textsf{SOE}}
\newcommand{\QLSA}{\textsf{QLSA}}

\newcommand{\norm}[1]{\left\lVert#1\right\rVert}

\newcommand{\calA}{{\cal A }}

\newcommand{\calP}{{\cal P }}
\newcommand{\calM}{{\cal M }}
\newcommand{\calZ}{{\cal Z }}

\newcommand{\soe}{\mathrm{soe}}
\newcommand{\lyap}{\mathrm{lp}}
\newcommand{\lp}{\mathrm{lp}}
\newcommand{\diag}{\mathrm{diag}}

\newtheorem{theorem}{Theorem}[section]
\newtheorem{definition}[theorem]{Definition}

\newtheorem{lemma}[theorem]{Lemma}
\newtheorem{remark}{Remark}
\newtheorem{corollary}[theorem]{Corollary}

\newtheorem{fact}[theorem]{Fact}
\newtheorem{claim}[theorem]{Claim}

\newtheorem{problem}[theorem]{{Problem}}
\newtheorem{result}[theorem]{{Result}}

\usepackage{thm-restate,mathrsfs} %To repeat theorems, lemmas etc....

\definecolor{olivegreen}{rgb}{0.42, 0.56, 0.14}
\definecolor{darkolivegreen}{rgb}{0.25, 0.4, 0.15}

\makeatletter
\def\widebreve{\mathpalette\wide@breve}
\def\wide@breve#1#2{\sbox\z@{$#1#2$}%
     \mathop{\vbox{\m@th\ialign{##\crcr
\kern0.08em\brevefill#1{0.8\wd\z@}\crcr\noalign{\nointerlineskip}%
                    $\hss#1#2\hss$\crcr}}}\nolimits}
\def\brevefill#1#2{$\m@th\sbox\tw@{$#1($}%
  \hss\resizebox{#2}{\wd\tw@}{\rotatebox[origin=c]{90}{\upshape(}}\hss$}
\makeatletter

\usepackage{authblk}

\title{Quantum simulation of non-Markovian dynamical systems}
\author[1*]{Abtin Ameri}
\author[2,\dag]{Arkopal Dutt}
\author[2,\ddag]{Hari Krovi}
\affil[1]{Massachusetts Institute of Technology, Cambridge, MA\vspace{2pt}}
\affil[2]{IBM Research, Cambridge, MA\vspace{2pt}}
\affil[ ]{
$^{*}$aameri@mit.edu
 \quad
$^{\dag}$ arkopal@ibm.com  \quad
$^{\ddag}$hari.krovi@ibm.com \quad
}
\date{\today}

\begin{document}

\date{\today}

\maketitle

\begin{abstract}
Existing quantum algorithms for simulating dynamical systems -- from Hamiltonian simulation to linear and nonlinear differential equations solvers -- simulate Markovian dynamics, in which the system's future evolution depends solely on its current state. We turn our attention to developing quantum algorithms for non-Markovian dynamical systems where the system's future evolution depends on its past history and thus has memory. Specifically, we develop efficient algorithms for linear Volterra integro-differential equations (VIDEs) with a convolution memory kernel that output a quantum state encoding the state description over a time interval or at a particular time. Given efficient circuits for the problem inputs, our algorithms achieve an exponential speedup in system size over existing classical algorithms. 

We develop an algorithm for general kernels assuming that $\mathcal{M} < 1$, where $\mathcal{M}$ characterizes the strength of the memory term relative to the dissipation of the Markovian part of the dynamics. We complement this with lower bounds for general-kernel VIDEs when $\mathcal{M} \geq 1$, showing that the problem becomes intractable for a family of systems. However, by specializing to structured kernels which admit concise decompositions over exponentials, we develop efficient quantum algorithms even when $\mathcal M \geq 1$ by converting the VIDE into a larger set of ODEs, a procedure which we call Markovianization. As an application of the overall framework, we discuss the Mori-Zwanzig formalism used in open quantum systems and fluid dynamics. Overall, our results expand the range of dynamical systems that quantum computers can simulate efficiently.
\end{abstract}

\newpage
\setcounter{tocdepth}{2}
{\tableofcontents}

\newpage

\section{Introduction} \label{sec:intro}
Differential equations are an extremely powerful tool for modeling and studying many dynamical systems in nature. Developing quantum algorithms to simulate differential equations, with the goal of attaining an exponential speedup in system size over classical algorithms, has been an active area of research for the past few decades. Simulating the Schrödinger equation -- better known as Hamiltonian simulation -- was one of the first problems tackled, as it was seen as a natural starting point of quantum algorithms for dynamical systems~\cite{feynman1982simulating,lloyd1996universal,low2017optimal}. Developing efficient quantum algorithms that went beyond Hamiltonian simulation came later, first for dissipative linear ordinary differential equations ($\ODE$s) ~\cite{berry2017quantum,krovi2023improved}, time-marching~\cite{fang2023time}, linear combination of Hamiltonian simulations~\cite{an2023linear,an2023quantum,low2025optimal}, quantum eigenvalue processing~\cite{low2024quantum}, and using Lindbladian dynamics~\cite{shang2025designing}.

The main limitation of the above algorithms is that they are applicable to \textit{Markovian dynamics} only. Markovian dynamical systems have the general form:
\begin{equation}
    \frac{d\mathbf{u}}{dt} = \mathbf{f}(\mathbf{u}(t),t), \qquad \mathbf{u}(0) = \mathbf{u}_0,
\end{equation}
with the right-hand-side solely depending on the solution $\mathbf u$ at the current time $t$. However, there is a plethora of dynamical systems that are inherently \textit{non-Markovian}~\cite{amari1977dynamics,maccamy1977integro,yanik1988finite,sforza1991parabolic,lakshmikantham1995theory,burton2005volterra}. Specific examples of such systems include population dynamics~\cite{kuang1993delay}, epidemiology models~\cite{kiss2015generalization}, viscoelasticity in solids and fluids~\cite{caputo1971linear}, and open-system dynamics and anomalous transport~\cite{zwanzig1961memory}. Such systems take the form
\begin{equation}
    \frac{d\mathbf{u}}{dt} = \boldsymbol {\mathcal F}\left[t, \{\mathbf{u}(\tau)\}_{0\leq \tau \leq t} \right], \qquad \mathbf{u}(0) = \mathbf{u}_0,
\end{equation}
where $\boldsymbol{\mathcal F}$ is a functional that incorporates the solution from all previous times. In the cases we consider, the functional has a natural decomposition into a Markovian part and a non-Markovian part. The latter is usually referred to as the memory term (also known as delay or lag term in the literature~\cite{burton2005volterra}). The presence of memory can modify the dynamics in non-trivial ways. For instance, as shown in Figure~\ref{fig:stability}, a suitably-chosen memory term can stabilize an otherwise unstable system.
\begin{figure}[h!]
    \centering
    \includegraphics[width=0.5\linewidth]{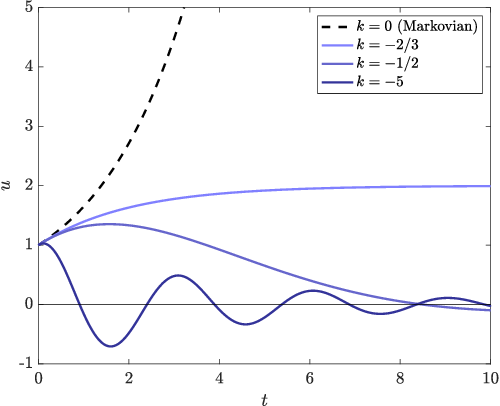}
    \caption{Numerical simulation of $\frac{du}{dt} = au + \int_0^t ke^{-\gamma (t-\tau)}u(\tau)d\tau$ with $u(0) = 1$, $a = 1/2$, and $\gamma = 1$. While the Markovian part of the dynamics is unstable (as shown by the dashed curve), adding a memory term can stabilize the overall dynamics and result in nontrivial behavior.}
    \label{fig:stability}
\end{figure}

Similar to Markovian dynamical systems, non-Markovian dynamical systems are also typically very high-dimensional; state-of-the-art classical numerical methods for such systems have complexity that scales polynomially with the dimension~\cite{linz1969linear,schadle2006fast}. Furthermore, a direct time discretization of non-Markovian equations typically couples each time step to all previous ones, leading to additional temporal overhead in the number of steps even before accounting for the system dimension. More concretely, the presence of the memory term makes these systems stiffer, requiring finer time stepping to accurately simulate the dynamics. For high-dimensional systems and long simulation times, classical simulations can become prohibitively expensive. Thus, a natural question to ask is whether quantum computers can accelerate the simulation of non-Markovian dynamics. Unfortunately, many of the existing techniques for quantum simulation of linear $\ODE$s are not applicable to simulating non-Markovian systems. This is because for linear $\ODE$s, the analytical solution is known and takes the form
\begin{equation}
    \frac{d\mathbf{u}}{dt} = A(t)\mathbf{u} + \mathbf{b}(t), \qquad \mathbf{u}(0) =\mathbf{u}_0,
\end{equation}
the exact solution is
\begin{equation}
    \mathbf{u}(t) = \mathcal{T}e^{\int_0^t A(\tau)  d\tau} \mathbf{u}_0 + \int_0^t \mathcal{T}e^{\int_{\tau}^t A(\tau')  d\tau'} \mathbf{b}(\tau) \ d\tau,
\end{equation}
where $\mathcal T$ is the time-ordering operator. There are numerous ways to discretize and approximate such an analytical solution, which has resulted in the variety of algorithms we noted earlier. On the other hand, an analytical solution to non-Markovian dynamics is not known, even for linear ones. Thus, many of the techniques developed for simulating Markovian dynamics are not applicable. This poses a major obstacle that we overcome in this work.

\subsection{Problem setup}
We tackle the problem of simulating linear Volterra integro-differential equations ($\VIDE$s):
\begin{equation} \label{eq:conv-LVIDE}
    \frac{d\mathbf u}{dt} = A \mathbf u + \int_0^t K(t-\tau) \mathbf u(\tau) \ d\tau + \mathbf b, \qquad \mathbf u(0) = \mathbf u_0.
\end{equation}
To tackle the problem of simulating Eq.~\eqref{eq:conv-LVIDE}, we make use of block-encodings. Block-encodings allow us to embed any matrix $\mathcal A$ into a larger unitary $U_{\mathcal A}$ such that
\begin{equation}
    U_{\mathcal A} = \begin{bmatrix}
        \mathcal A/\alpha & . \\
        . & .
    \end{bmatrix},
\end{equation}
with $\alpha \geq \|\mathcal A\|$. 

We consider two simulation problems. The first is history state preparation, where the output of the quantum algorithm is a (normalized) state containing the solution to the $\VIDE$ over $m$ time steps $t_j$ and takes the form
\begin{equation}
    |\Psi_{\mathrm{hist}}\rangle \coloneq \frac{1}{\mathcal N} \sum_{j=0}^m \|\mathbf u(t_j)\| |j\rangle |\mathbf{u}(t_j)\rangle.
\end{equation}
The second problem is final state preparation, where the output is the (normalized) solution state at the final simulation time $T$ and takes the form
\begin{equation}
    |\mathbf u(T)\rangle \coloneq \frac{\mathbf u(T)}{\|\mathbf u(T)\|}.
\end{equation}
First, let us define the history state preparation problem.
\begin{problem} \label{prob:hist-state}
Consider the linear $\VIDE$ in Eq.~\eqref{eq:conv-LVIDE}, with $K$ being differentiable entry-wise. Let $T$ be the final simulation time, which is divided into $m$ intervals of length $h$ (such that $T=mh$). Suppose we have access to $\|\mathbf u_0\|$ and $\|\mathbf b\|$, as well as oracles $O_{\mathbf u}$ and $O_{\mathbf b}$, which prepare the initial condition state $\ket{\bf{u}_0}$ and forcing state $\ket{\bf{b}}$, respectively.

Furthermore, we have access to an $(\alpha,a,\delta)$ block-encoding $U_A$ of $A$ and $(\beta,a,\delta)$ block-encoding $U_{K_{j}}$ of $K(jh)$ for each $j\in[m]_0$. Let the history state be
\begin{equation*}
    |\Psi_{\mathrm{hist}}\rangle \coloneq \frac{1}{\mathcal N} \sum_{j=0}^m \|\mathbf u(jh)\| |j\rangle |\mathbf{u}(jh)\rangle,
\end{equation*}
with $\mathbf u(jh)$ being the solution to Eq.~\eqref{eq:conv-LVIDE} at time $jh$ with $|\mathbf{u}(t)\rangle := \mathbf u(t)/\|\mathbf u(t)\|$ and $\mathcal N := \left ( \sum_{j=0}^m \|\mathbf u(jh)\|^2 \right )^{1/2}$. The goal is to then prepare the quantum state $|\Psi\rangle$ such that 
$$
\left \| |\Psi\rangle - |\Psi_{\mathrm{hist}}\rangle \right \| \leq \varepsilon.
$$
\end{problem}
\begin{remark} \label{rem:hist-state-prob}
    It is also common to formulate the history state preparation problem such that the output satisfies $\|\langle j|\psi'\rangle - \langle j|\psi\rangle\| \leq \varepsilon$ $\forall j\in[m]_0$. The complexity of this formulation and the one in Problem~\ref{prob:hist-state} should be equivalent up to polynomial factors in $T$ and $\varepsilon$.
\end{remark}
We now define the final state preparation problem below.
\begin{problem} \label{prob:final-state}
Consider the context of Problem~\ref{prob:hist-state}. The goal is to output a quantum state $|\psi\rangle$ such that 
\begin{equation*}
    \| |\psi\rangle - |\mathbf u(T)\rangle \| \leq \varepsilon,
\end{equation*}
where $\mathbf u(T)$ is the solution to Eq.~\eqref{eq:conv-LVIDE} at time $T$.
\end{problem}
\begin{remark}
    Note that, in Problems~\ref{prob:hist-state} and~\ref{prob:final-state}, we have access to the block-encoding of $K(x_i)$ at each desired $x_i$. This is equivalent to having access to the block-encoding of $\sum_{i}|i\rangle \langle i| \otimes K(x_i)$, which is used in time-dependent Hamiltonian simulation~\textnormal{\cite{berry2020time,fang2025time}} and differential equation solvers~\textnormal{\cite{fang2023time,an2023quantum}}. We shall formulate the complexity of the algorithms in terms of the total queries made to each $K(x_i)$, which we shall refer to as queries made to $K$.
\end{remark}

\subsection{Main results}
The main contributions of this paper can be divided into five parts. 

\paragraph{Memory strength.} We define a parameter $\mathcal M$, which characterizes the ratio of the strength of the memory term to the dissipation of the Markovian part of the dynamics. We denote dynamics that satisfies $\mathcal M < 1$ as short-term memory dynamics, where essentially the system quickly forgets about its past history. This condition is somewhat similar to the $\mathrm{R}<1$ condition on Carleman linearization of nonlinear $\ODE$s~\cite{liu2021efficient}, which characterizes the ratio of the nonlinearity to the linear dissipation in the system. We show that the short-term memory condition is crucial in the stability of $\VIDE$s, which will be needed when developing our algorithms.

\paragraph{Algorithms for general kernels.} We show that, under certain constraints, one can develop efficient quantum algorithms for linear $\VIDE$s with an exponential speedup with respect to the system dimension over existing classical algorithms. The algorithm relies on directly discretizing the equations, recasting them as a linear system, and using $\QLSA$s to solve the linear system. The output of the linear system is a history state, though one can also post-select on the solution at the final time using amplitude amplification. For our algorithm to be efficient, we require the short-term memory condition to be satisfied. We informally state these results below.
\begin{result} [Informal version of~\cref{thm:general-dynamics-algorithm-hist-state}]
There exists a quantum algorithm that solves Problem~\ref{prob:hist-state} with time complexity $\poly(T,\ \|A\|,\ 1/\varepsilon,\ \mathscr K, \ \mathcal K, \, \mathfrak g)$ where
\begin{align*}
\mathscr K \coloneq \sup_{x\geq 0}\|K(x)\|, \quad \mathcal K \coloneq \sup_{x\geq 0}\|K'(x)\|, \quad \text{and } \mathfrak g \coloneq \sqrt{\frac{1}{T} \int_0^T \|\mathbf u(\tau)\|^2 \ d\tau}.
\end{align*}
\end{result}
\begin{result} [Informal version of~\cref{thm:general-dynamics-algorithm-final-state}]
There exists a quantum algorithm that solves Problem~\ref{prob:final-state} with time complexity $\poly(T,\ \|A\|,\ 1/\varepsilon,\ \mathscr K, \ \mathcal K, \ q, \, g)$ where
\begin{align*}
\mathscr K \coloneq \sup_{x\geq 0}\|K(x)\|,\,\, \mathcal K \coloneq \sup_{x\geq 0}\|K'(x)\|, \,\, q \coloneq \|\mathbf u(T)\|, \text{ and } \,\, g \coloneq \frac{\max_{t\in[0,T]}\|\mathbf u(t)\|}{\|\mathbf u(T)\|}.
\end{align*}
\end{result}

\paragraph{Lower bound.} The third part of our contributions is a lower bound for simulating general linear $\VIDE$s with $\mathcal M \geq 1$. We show that such $\VIDE$s can accelerate state discrimination, which, in turn, imposes a limitation on the number of copies of the initial condition needed. We state this result informally below.

\begin{result}[Informal version of~\cref{thm:lb_approx}]
    For any $\mathcal{M}\geq 1$, there is an instance of the $\VIDE$ of Eq.~\eqref{eq:conv-LVIDE} such that any quantum algorithm producing a quantum state approximating $|\mathbf{u}(T)\rangle $ must require exponentially many (in $T$) copies of the initial condition. 
\end{result}

\paragraph{Algorithms for structured kernels.} 
The lower bound above rules out efficient algorithms for general kernels violating the short-term memory condition. However for structured kernels $K(x) = -k(x) I$ where $k(x)$ admits a concise decomposition over exponentials, we show that these limitation can be circumvented and one can simulate non-Markovian dynamics with $\mathcal M \geq 1$. We develop algorithms based on converting the linear $\VIDE$ into a larger system of linear $\ODE$s, which we refer to as \textit{Markovianization}. This technique is also similar to the linear embedding technique used for nonlinear $\ODE$s, which involves introducing auxiliary variables. Once we have a set of linear $\ODE$s, we can use existing quantum $\ODE$ solvers. In addition to circumventing the short-term memory constraint, the algorithms developed via Markovianization can have better complexity since existing quantum ODE solvers are nearly optimal. However, we require the memory kernel to have a specific form; specifically, we demand that the kernel be decomposable as a sum of exponentials. An example is the power-law kernel, for which we have the following results for Lyapunov stable $A$.
\begin{result} [Informal version of~\cref{thm:power-law-alg-hist-state}]
There exists a quantum algorithm that solves Problem~\ref{prob:hist-state} with $K(x) = -x^{-\beta}I$ and $0<\beta<1$. For a fixed $\beta$, the algorithm has complexity that is polynomial in $T,\ \|A\|,\ 1/\varepsilon$, $\norm{\mathbf{u}_0}$ and $1/(1-\beta)$.
\end{result}

\begin{result}[Informal version of ~\cref{thm:power-law-alg-final-state}]
There exists a quantum algorithm that solves Problem~\ref{prob:final-state} with $K(x) = -x^{-\beta}I$ and $0<\beta<1$. For a fixed $\beta$, the algorithm has complexity that is polynomial in $T,\ \|A\|,\ 1/\varepsilon, \ g_0$, and $1/(1-\beta)$, where
\begin{equation}
    g_0 \coloneq \frac{\|\mathbf u_0\|}{\|\mathbf u(T)\|}.
\end{equation}
\end{result}

\paragraph{Application to the Mori-Zwanzig formalism.} As an application, we consider the Mori-Zwanzig formalism (also known as reduced-order modeling or coarse-graining in the literature). This formalism converts a linear $\ODE$ in a larger state space to a $\VIDE$ in a smaller state space. Such a technique is useful when we are only interested in a subset of the variables of the original $\ODE$, but have no desire to track the time evolution of the undesired variables. We show conditions under which our algorithms can be used for the Mori-Zwanzig formalism.

The algorithms we develop in this work attain an exponential speedup in the dimension of the dynamical system provided access to efficient constructions of $O_{\mathbf u}$ and $O_{\mathbf b}$, as well as efficient block-encoding circuits for $A$ and $K$. Such efficient circuits can be constructed if, for instance, the matrices are row- and column-sparse~\cite{gilyen2019quantum} or have special structure that can be exploited.

\subsection{Technical summary}\label{sec:techsum}
Here we will provide a brief overview of the techniques used to prove the main results of the paper. 

\paragraph{General kernels.} The core subroutine we use is a quantum linear system algorithm ($\textsf{QLSA}$)~\cite{harrow2009quantum,childs2017quantum,costa2022optimal}. We now describe the linear system constructed from the $\VIDE$. We discretize the time interval
\([0,T]\) with $m$ intervals of step size \(h=T/m\). 
We use forward Euler for time-discretization of $\mathbf u$ and a left Riemann sum for the memory integral, giving the following recurrence
\begin{equation}
    \mathbf u_{j+1} = (I+Ah)\mathbf u_j + h^2\sum_{k=0}^{j-1} K((j-k)h)\mathbf u_k + h\mathbf b.
\end{equation}
Since each update depends on all previous steps, the evolution is naturally encoded as a block lower-triangular linear system $L\mathbf y=\mathbf c$, whose solution \(\mathbf y\) contains the whole discretized trajectory. The matrix \(L\) can then be block-encoded band by band using standard sparse-access and linear-combination-of-unitaries techniques~\cite{gilyen2019quantum}, while the vector \(\mathbf c\) can be prepared from the input oracles for \(\mathbf u_0\) and \(\mathbf b\). Using $\textsf{QLSA}$s to solve this system yields the history state. For final-state preparation, the construction of $L$ and $\mathbf c$ is padded so that the last time slice can be obtained by postselection and boosted by amplitude amplification~\cite{Brassard_2002}.

The main contribution to the complexity of the algorithm is the condition number of the linear system $L$ which we show is bounded under the short-term memory condition. Specifically, we show that when \(\mu(A)<0\), \(\calM<1\), and the time step size is sufficiently small, the condition number of the linear system grows
only linearly with the number of time steps.

\paragraph{Lower bound.} 
We observe that $\VIDE$s with $\mathcal M \geq 1$ can quickly amplify the distinguishability of nearly overlapping initial states. We construct a simple two-dimensional $\VIDE$ and two initial states whose overlap is \(1-\epsilon\). We then observe that when \(\M\ge 1\), the distance between the normalized evolved states becomes constant after time \(T=O(\log(1/\epsilon))\).
If such dynamics could be simulated efficiently in general, one could distinguish the two initial
states using too few copies, contradicting known state-discrimination lower bounds~\cite{abrams1998nonlinear,childs2016optimal,liu2021efficient,an2022theory}.
This establishes an exponential-in-\(T\) hardness result for generic kernels beyond the short-term memory regime.

\paragraph{Structured kernels.} As mentioned earlier, we can still give efficient quantum algorithms for structured kernels that do not satisfy the short-term memory condition. Particularly for kernels which can be approximated by a sum of exponentials ($\SOE$), we convert the linear $\VIDE$ into a small coupled system of $\ODE$s by introducing auxiliary variables that absorb the memory term. For example, for \(K(x) = e^{-\gamma x} I\), we define $\mathbf y_1(t)=\int_0^t e^{-\gamma(t-\tau)}\mathbf u(\tau)\,d\tau$, which satisfies its own first-order $\ODE$. The original $\VIDE$ is then exactly equivalent to a finite-dimensional $\ODE$ over \((\mathbf u,\mathbf y_1)\). When $K(x)$ is a sum of $s$ many exponentials, one auxiliary variable is introduced per exponential term resulting in an $\ODE$ over $s+1$ variables. At this point, a quantum $\ODE$ solver can be used directly for history-state preparation~\cite{krovi2023improved}. As an application, we illustrate this for the weakly-singular power-law kernel, which can appear in Langevin equations~\cite{kubo1966fluctuation}.

\subsection{Outlook} \label{sec:conclusion}
There are several avenues that can be investigated in the area of quantum algorithms for non-Markovian dynamics. We briefly discuss each one below.

\paragraph{Higher-order methods.} The direct discretization algorithm for general kernels can be improved upon by considering linear multistep methods. These methods employ higher-order methods to discretize the time derivative as well as the integral. This would result in polynomially improved complexity with respect to $T$ and $\varepsilon$. Related to this, developing lower bounds in terms of $T$ and $\varepsilon$ for $\VIDE$s with $\mathcal M<1$ would be interesting. In particular, can quantum algorithms achieve $O(T)$ complexity, or does the presence of non-Markovian behavior not allow this? Furthermore, is it possible to achieve $O(\log(1/\epsilon))$ scaling for general kernels? 

To improve the complexity of this class of algorithms beyond linear multistep methods might require additional constraints on the dynamics. For our Markovianization approach, we highly restricted the kernels. Is it possible to remove some constraints and still end up with improved complexity? Are there constraints we can place on the spectrum of $K$ that can help us find better algorithms? For instance, some forms of $K$ can stabilize a dynamics that would be otherwise unstable. Would it be possible to develop efficient algorithms for such systems? Lastly, would it be possible to develop efficient algorithms for dynamical systems that have no Markovian term?

\paragraph{Lower bounds in the short-term memory regime.} While we have efficient quantum algorithms in the $\mathcal M<1$ regime, it is far from optimal. However, what lower bounds can one prove in this regime? For instance, is it possible to have algorithms with $\poly(\log(1/\epsilon))$ complexity? Developing lower bounds in the short-term memory can help guide us in developing more efficient algorithms.

\paragraph{Improving the Markovianization approach}
Our results for the Markovianization approach, while powerful, are quite restricted; there are many constraints placed on the kernel in order to extract an efficient quantum algorithm. The main direction that should be explored in this area is extending the $\SOE$ analysis to complex-valued decompositions. The $\SOE$ decomposition currently only works when each of the exponentials are real-valued. This is sufficient to handle the power-law kernel, but many kernels are not decomposable in this way. One example is the Gaussian kernel, which is very well-used in non-Markovian dynamics. Extending our analysis to complex-valued exponentials would greatly expand the set of kernels that can be handled by the $\SOE$ approach.

\paragraph{Beyond linear dynamics.} One can also consider non-Markovian systems with a bilinear kernel term
\begin{equation}
    \frac{d\mathbf{u}}{dt} = A\mathbf{u} + \int_0^t K(t-s)(\mathbf{u}(t)\otimes \mathbf{u}(s) )\ ds + \mathbf{b}.
\end{equation}
It would be interesting to see whether simulating such a system is any easier than simulating quadratically nonlinear $\ODE$s of the form
\begin{equation}
    \frac{d\mathbf{u}}{dt} = A\mathbf{u} + F \mathbf{u}^{\otimes 2}+ \mathbf{b},
\end{equation}
which has been explored in previous work~\cite{liu2021efficient,krovi2023improved,wu2025quantum,jennings2025quantum}.

\subsection{Organization of paper}
The rest of this paper is organized as follows. Section~\ref{sec:preleminaries} lays out the preliminaries that will be needed for later. In Section~\ref{sec:short-term}, we define the short-term memory condition and its impact on the dynamics of $\VIDE$s. Section~\ref{sec:general_kernels} discusses the algorithm for simulating general kernels with $\mathcal M < 1$. Section~\ref{sec:lb} shows that simulating general linear $\VIDE$s with $\mathcal M \geq 1$ is intractable on quantum computers. Then, Section~\ref{sec:structured_kernels} focuses on algorithms for structured kernels via Markovianization. In Section~\ref{sec:applications} we explore the Mori-Zwanzig formalism as an application of our algorithms. 

\paragraph{Note added.} While finalizing this paper, \cite{setoyama2026} appeared on arXiv, which focuses on simulating linear $\VIDE$s with convolution kernels of the \emph{phase type}. There are several differences between their results and ours. First, our paper considers linear $\VIDE$s with general convolution kernels and develops algorithms for this general case under the \emph{short-term memory condition} (which we introduce later). We also provide lower bounds for when the short-term memory condition is not satisfied. The general case and lower bounds for quantum simulation is not explored in~\cite{setoyama2026}. We also consider linear $\VIDE$s with structured kernels that can be approximated by a sum of exponentials and introduce an approach, which we call Markovianization. \cite{setoyama2026} considers kernels of the phase type, which are a subset of the structured kernels that we consider, and uses the linear chain trick, which also introduces auxiliary variables but differently from that of Markovianization. We rigorously show, however, that the Markovianization approach we propose here can be extended beyond exponential kernels by considering kernels that are decomposable as a sum of exponentials. 

\paragraph{Acknowledgements.} AA thanks Diogo Cruz and Thibault Fredon for early discussions on simulating integro-differential equations, specifically on the Markovianization approach. AD thanks Lin Lin for discussions regarding simulating classical dynamical systems with memory during the IPAM's (Institute for Pure and Applied Mathematics) long program ``Mathematical and Computational Challenges in Quantum Computing" in Fall 2023. AA and HK acknowledge support from the US Department of Energy (grant DE-SC0020264). A part of this work was done while AA was an intern at IBM Quantum in the summer of 2025. The authors dedicate this paper to the memory of Nuno F.~G.~Loureiro.

\newpage

\section{Preliminaries} \label{sec:preleminaries}
\subsection{Notation and definitions}
Let $\mathbb N_0$, $\mathbb{N}$, $\mathbb Z$, $\mathbb R$, $\mathbb R_+$, and $\mathbb C$ denote the set of nonnegative integers, positive integers, integers, real numbers, nonnegative real numbers, and complex numbers, respectively. For $n \in \mathbb N$, we define $[n]_0 \coloneq \{0,1,\dots,n\}$ and $[n]\coloneq \{1,\dots,n\}$. We let $\| \cdot \|$ denote the 2-norm, unless otherwise specified by an appropriate subscript. We denote the identity matrix by $I$. Finally, we denote vectors in bold.

For any vector $\mathbf u \in \mathbb C^n$ with components $u_j$, $j\in[n-1]_0$, we denote its quantum encoding as 
\begin{equation}
    |\mathbf u\rangle \coloneq \frac{\mathbf u}{\|\mathbf u\|} = \frac{1}{\|\mathbf u\|} \sum_{j=0}^{n-1}u_j|j\rangle.
\end{equation}

Let $A\in \mathbb C^{n\times n}$ be a matrix. We define the spectrum of $A$, $\sigma(A)$, and the log-norm of $A$, $\mu(A)$, as
\begin{align}
    \sigma(A) &\coloneq \{\lambda | \lambda \ \mathrm{is \ an \ eigenvalue \ of} \ A \} \\
    \mu(A) &\coloneq \max \{\lambda | \lambda \in \sigma ((A+A^\dagger)/2) \}.
\end{align}

It is easy to see that
\begin{equation}
    \mu(A) \leq \|A\|.
\end{equation}
Let $T>0$. We define the exponential norm of $A$ on the interval $[0,T]$ as
\begin{equation}
    \expnorm (A) \coloneq \sup_{t\in[0,T]}\left \|e^{At} \right \|.
\end{equation}

The following result is well-known.
\begin{fact} [\cite{strom1975logarithmic}, Lemma 1c] \label{lem:log-norm}
    For any matrix $A$ with log-norm $\mu(A)$, we have
    \begin{equation}
        \|e^{At}\| \leq e^{\mu(A) t} \leq e^{\|A\|t}
    \end{equation}
\end{fact}
For more information about these quantities in the context of quantum $\ODE$ solvers, we refer the reader to~\cite{krovi2023improved}.

\subsection{Quantum routines}
Since we deal with quantum states that are normalized, we will need the following commonly-used fact.
\begin{fact} \label{lem:normalized_state_error}
Let $|\mathbf u\rangle = \mathbf u/\|\mathbf u\|$ and $|\mathbf v \rangle = \mathbf v/\|\mathbf v\|$. Then, 
$$
\left \||\mathbf u\rangle - |\mathbf v \rangle \right \| \leq \frac{2\|\mathbf u - \mathbf v\|}{\|\mathbf u\|}.
$$
\end{fact}
\begin{proof}
Note that $\|\mathbf v \| - \|\mathbf u \| \leq \|\mathbf u - \mathbf v\|$. Using the triangle inequality, we obtain
\begin{align*}
    \left \| \frac{\mathbf u}{\|\mathbf u\|} - \frac{\mathbf v}{\|\mathbf v\|} \right \| \leq \left \| \frac{\mathbf u}{\|\mathbf u\|} - \frac{\mathbf v}{\|\mathbf u\|} \right \| + \left \| \frac{\mathbf v}{\|\mathbf u\|} - \frac{\mathbf v}{\|\mathbf v\|} \right \|
    \leq \frac{\|\mathbf u - \mathbf v \|}{\|\mathbf u\|} +  \frac{\|\mathbf v\| - \|\mathbf u\|}{\|\mathbf u\|}
    \leq \frac{2\|\mathbf u - \mathbf v\|}{\|\mathbf u\|},
\end{align*}
where we added and subtracted $\mathbf{v}/\|\mathbf{u}\|$ in the second inequality. This gives us the desired result.
\end{proof}
\subsubsection{Block-encodings}
Block-encodings are a powerful tool that allow for matrix arithmetic by embedding an arbitrary matrix $A$ in the top-left block of a larger unitary.
\begin{definition}[Block-encodings]
\label{def:BE}
Suppose $A$ is an $n$-qubit matrix operator, $\alpha,\varepsilon \in \mathbb R_+$ and $b\in \mathbb N$. Then, an $n+m$-qubit unitary $U_A$ is said to be an $(\alpha,m,\varepsilon)$ block-encoding of $A$ if
$$
\|A - \alpha (\langle 0|^m \otimes I )U_A(|0\rangle^m \otimes I) \| \leq \varepsilon.
$$
\end{definition}

We also need to define the state-preparation unitary, which we will use for matrix arithmetic on block-encodings.
\begin{definition}[\cite{chakraborty2023quantum}, Definition 16]
Let $m \in \mathbb N$ and $s=\ceil{\log m}$. Let $\boldsymbol{\eta} \in \mathbb R^m_+$. Then we call an $s$-qubit unitary $P$ an $\boldsymbol{\eta}$ state-preparation unitary if
\begin{equation*}
    P|0\rangle = \frac{1}{\sum_{j=0}^{m-1} \eta_j} \sum_{j=0}^{m-1} \sqrt{\eta_j} |j\rangle.
\end{equation*}
\end{definition}

\begin{lemma}[\cite{gilyen2019quantum}, Lemma 48] \label{lem:sparse-block-encoding}
Let $A \in \mathbb{C}^{2^w \times 2^w}$ be a matrix that is $s_r$-row-sparse and $s_c$-column-sparse, and each element of $A$ has absolute value at most $1$. Suppose that we have access to the following sparse-access oracles acting on two $(w+1)$-qubit registers:
\begin{align*}
    O_r:\quad &\ket{i}\ket{k} \;\to\; \ket{i}\ket{r_{ik}}, 
    \quad \forall i \in [2^w]-1,\; k \in [s_r], \\
    O_c:\quad &\ket{\ell}\ket{j} \;\to\; \ket{c_{\ell j}}\ket{j}, 
    \quad \forall \ell \in [s_c],\; j \in [2^w]-1.
\end{align*}

Here, $r_{ij}$ is the index for the $j$-th non-zero entry of the $i$-th row of $A$, and $c_{ij}$ is the index for the $i$-th non-zero entry of the $j$-th column of $A$. Additionally, assume that we have access to an oracle $O_A$ that returns the entries of $A$ in a binary description:
\begin{equation*}
    O_A:\quad \ket{i}\ket{j}\ket{0}^{\otimes b} \;\to\; \ket{i}\ket{j}\ket{a_{ij}}, 
\qquad \forall i,j \in [2^w]-1,
\end{equation*}
where $a_{ij}$ is a $b$-bit binary description of the $ij$-th matrix element of $A$. Then we can implement a $(\sqrt{s_r s_c},\, w+3,\, \varepsilon)$-block-encoding of $A$ with a single use of $O_r$, $O_c$, two uses of $O_A$, and additionally using $O\!\left(w + \log^{2.5}\!\left(\frac{s_r s_c}{\varepsilon}\right)\right)$ one- and two-qubit gates while using $O\!\left(b + \log^{2.5}\!\left(\frac{s_r s_c}{\varepsilon}\right)\right)$ ancilla qubits.
\end{lemma}

Next we state some results on how to perform matrix arithmetic on block-encoded matrices that we will need later. The first result is on linear combination of block-encodings.

\begin{lemma} [\cite{chakraborty2023quantum}, Lemma 17] \label{lem:lin-comb-BE}
For each $j \in [m-1]_0$, let $A_j$ be an $s$-qubit operator, and $y_j \in \mathbb{R}_+$. Let $U_j$ be a $(\alpha_j, a_j, \varepsilon_j)$-block-encoding of $A_j$, implemented in time $T_j$. Define the matrix $A = \sum_j y_j A_j$, and the vector $\boldsymbol \eta \in \mathbb{R}^m$ s.t. $\eta_j = y_j \alpha_j$. Let $U_{\boldsymbol\eta}$ be an $\boldsymbol\eta$ state-preparation unitary, implemented in time $T_{\boldsymbol\eta}$. Then we can implement a

\begin{equation*}
\left(
\sum_{j\in[m-1]_0} y_j \alpha_j,\ \max_{j\in[m-1]_0} (a_j) + s,\ \sum_{j\in[m-1]_0} y_j \varepsilon_j
\right)
\end{equation*}
block-encoding of $A$ at a cost of $O\left(\sum_{j=0}^{m-1} T_j + T_{\boldsymbol \eta}\right)$.
\end{lemma}
The second result is on tensor product of block-encodings.
\begin{lemma} [\cite{Takahira:2021xsr}, Proposition 4] \label{lem:tens-prod-BE}
    If $U_A$ is a $(\alpha, a, \varepsilon)$ block-encoding of an $n$ qubit matrix $A$ with gate complexity $T_A$ and $U_B$ is an $(\beta, b, \delta)$ block-encoding of an $m$ qubit matrix $B$ with gate complexity $T_B$, then $(I_b\otimes U_A\otimes I_m)(I_a\otimes I_n\otimes U_B)$ is a $(\alpha\beta, a+b,\beta\varepsilon + \alpha\delta)$ block-encoding of $A\otimes B$ with gate complexity $O(T_A + T_B)$.
\end{lemma}

\subsubsection{Quantum linear system algorithms}
Quantum linear system algorithms ($\QLSA$s) are a powerful set of algorithms for solving linear systems of equations of the form $A\mathbf x = \mathbf b$~\cite{harrow2009quantum,childs2017quantum,costa2022optimal,an2022quantum}. We will need to use $\QLSA$s as a main subroutine in Section~\ref{sec:general_kernels}. We will make use of the following result. 
\begin{theorem}[\cite{costa2022optimal}, Theorem 19] \label{thm:QLSA}
Let A be such that $\|A\| = 1$ and $\|A^{-1}\| = \kappa$. Given a block-encoding of A and an oracle for implementing $|b\rangle $, there exists a quantum algorithm which produces the normalized state $A^{-1}|b\rangle$ to within an error $\varepsilon$ using
\begin{equation}
    O(\kappa \log 1/\varepsilon )
\end{equation}
calls to the oracles.
\end{theorem}
Note that, while the algorithm requires $\|A\|=1$, we can relax this constraint to $\|A\|=O(1)$ without changing the asymptotic complexity of the algorithm. 

\subsubsection{Quantum $\ODE$ solvers} 
In Section~\ref{sec:structured_kernels}, we will also require a quantum algorithm for solving linear $\ODE$s. We will use the algorithm of \cite{krovi2023improved} which has the following guarantee when given access to an approximate block-encoding of the linear differential operator (as explicitly shown in \cite{dutt2026rlc}).

\begin{theorem}[\cite{dutt2026rlc}, Theorem 4.25]\label{thm:ODE_solver_Krovi_history_state} 
Let $\varepsilon \in (0,1)$. Consider the linear ordinary differential equation
\begin{equation}
    \frac{d\mathbf{x}}{dt} = A\mathbf{x} + \mathbf{b}, \qquad \mathbf{x}(0) = \mathbf{x}_0,
\end{equation}

for $A \in \mathbb{C}^{N \times N}$ and $\mathbf{b}, \mathbf{x}_0 \in \mathbb{C}^N$.
We are given an $(\alpha_{A}, a_{A}, \varepsilon_{A})$ block-encoding $U_{A}$ for $A$ and unitaries $O_x, O_b$ that prepare states proportional to $\mathbf{x}_0$ and $\mathbf{b}$, respectively. Let the history state encoding of the time-evolution of $\mathbf{x}(t)$ across the time-interval $[0,T]$ be
\begin{equation*}
\ket{\Psi} = \frac{1}{\mathcal{N}} \sum_{j=0}^{m} \norm{\mathbf x (j h)} \ket{j, \mathbf x(jh)},
\end{equation*}
where $m = \lceil T/h \rceil$ and $h = O(1/\|A\|)$. Then, there exists a quantum algorithm that outputs $|\widehat{\Psi}\rangle$ such that $\||\widehat{\Psi}\rangle - \ket{\Psi}\| \leq \varepsilon$. Define
\begin{equation}
    \kappa \coloneq T\|A\|\expnorm(A), \qquad \mu^2 \coloneq \frac{1}{m}\sum_{j=1}^m \|\mathbf x(jh)\|^2, \qquad \mathcal C_b \coloneq \log\left(1 + \frac{T e^2 \|\mathbf{b}\|}{\mu}\right).
\end{equation}
The complexity is as follows:
\begin{itemize}
    \item Uses of $U_{A}$: $\tilde O\Big(\kappa^2 \cdot\poly(\mathcal C_b,\log(1/\varepsilon),\log(\kappa))  \Big)$
    \item Uses of $O_x, O_b$: $\tilde O\Big( \kappa \mathcal C_b \log(1/\varepsilon)\Big)$
    \item Additional gates: $\tilde O\Big(\kappa^2 \cdot \poly(\mathcal C_b,\log(1/\varepsilon),\log(\kappa))\Big)$
\end{itemize}
The algorithm is successful provided the block-encoding precision is 
\begin{equation}
    \varepsilon_A = o\left(\varepsilon \kappa^{-3}\cdot \poly (\mathcal C_b,\log(\kappa),\log(1/\varepsilon) )^{-1} \right).
\end{equation}
\end{theorem}

The algorithm of Theorem~\ref{thm:ODE_solver_Krovi_history_state} can also be used to prepare the normalized state at a particular time. This is formally described stated below.
\begin{theorem}[\cite{dutt2026rlc}, Theorem 4.26]\label{thm:ODE_solver_Krovi_final_state}
Let $\varepsilon \in (0,1)$. Consider the context of Theorem~\ref{thm:ODE_solver_Krovi_history_state}. Then, there exists a quantum algorithm that outputs $\ket{\phi}$ such that $\left\| \ket{\phi} - \ket{\mathbf x(T)} \right\| \leq \varepsilon$.
Define
\begin{equation}
    \kappa \coloneq T\|A\|\expnorm(A), \qquad g\coloneq \frac{\max_{t\in[0,T]}\|\mathbf x(t)\|}{\|\mathbf x(T)\|}, \qquad \mathcal C_b \coloneq \log\left(1 + \frac{T e^2 \|\mathbf{b}\|}{\|\mathbf x(T)\|}\right).
\end{equation}
The complexity of the algorithm is as follows:
\begin{itemize}
    \item Uses of $U_{A}$: $\tilde O\Big(g \kappa \cdot \poly(\mathcal C_b,\log(1/\varepsilon),\log(\kappa)) \Big)$
    \item Uses of $O_x, O_b$: $\tilde O\Big( g\kappa \cdot\log(1/\varepsilon)\Big)$,
    \item Additional gates: $\tilde O \Big ( g\kappa \cdot \poly(\mathcal C_b,\log(1/\varepsilon),\log(\kappa)) \Big )$.
\end{itemize}
The algorithm is successful provided the block-encoding precision is 
\begin{equation}
    \varepsilon_A = o\left(\varepsilon \kappa^{-1}\cdot \poly(\mathcal C_b,\log(1/\varepsilon),\log(\kappa))^{-1}  \right).
\end{equation}
\end{theorem}
We will often use the following fact~\cite{bhatia2013matrix} to bound the exponential norm of a matrix.
\begin{fact}\label{fact:expnorm_bound}
For any matrix $B$, the norm of the exponential of $B$ is upper bounded by the norm of the exponential of its Hermitian part i.e., 
$$
\norm{\exp(B)} \leq \norm{\exp((B+B^\dagger)/2)}.
$$
\end{fact}
We will also use the following Gr\"onwall's inequality~\cite{evans2022partial}.
\begin{lemma}\label{lem:gronwalls_inequality}
Let $u:[t_0,T]\to [0,\infty)$ be continuous, and let $a:[t_0,T]\to[0,\infty)$ be integrable. If
$$
u(t) \leq C + \int_{t_0}^t a(s) u(s) ds,
$$
for all $t \in [t_0,T]$ where $C \geq 0$ then
$$
u(t) \leq C \exp \Big(\int_{t_0}^t a(s) ds \Big).
$$
\end{lemma}
Lastly, we will require the following lemma to comment on the success probability when extracting particular components from a prepared history state.
\begin{lemma}\label{lem:extract_from_hist_state}
Let $\varepsilon \in (0,1)$. Let $\Pi$ be an orthogonal projector and let $\ket{\psi},\ket{\phi}$ be $n$-qubit quantum states such that $\norm{\ket{\psi} - \ket{\phi}} \leq \varepsilon$. Define $p:=\|\Pi \ket{\psi}\|^2 > 0$. If $\varepsilon < \sqrt{p}$, then
$$
\norm{\frac{\Pi \ket{\psi}}{\norm{\Pi \ket{\psi}}} - \frac{\Pi \ket{\phi}}{\norm{\Pi \ket{\phi}}}} \leq \frac{2 \varepsilon}{\sqrt{p}},
$$
where the probability of obtaining $\Pi \ket{\phi}$ is at least $(\sqrt{p} - \varepsilon)^2$.
\end{lemma}
\begin{proof}
Using submultiplicativity of norms, we have $\norm{\Pi \ket{\psi} - \Pi \ket{\phi}} \leq \norm{\Pi} \cdot \norm{\ket{\psi} - \ket{\phi}} \leq \varepsilon$. Using $\norm{\Pi \ket{\psi}} = \sqrt{p}$ combined with Fact~\ref{lem:normalized_state_error} gives us the desired result. The probability of success follows from noting that $\norm{\Pi \ket{\phi}}^2 \geq (\norm{\Pi \ket{\psi}} - \varepsilon)^2 \geq (\sqrt{p} - \varepsilon)^2$. This completes the proof.
\end{proof}

\subsection{Riemann sums}
Riemann sums are a simple method of numerically evaluating integrals. In Section~\ref{sec:general_kernels}, we will make use of the left Riemann sum, whose error bound is provided by the following result.
\begin{lemma} \label{lem:riemann}
Let $t>0$ and $\mathbf{f}(x): \mathbb R\to \mathbb C^N$ be differentiable on $x\in[0,t]$. Divide the interval $[0,t]$ into $n$ sub-intervals of length $h$, with $x_j \coloneq jh$ for $j\in[n]_0$, and $x_n= nh= t$. Let $M \coloneq \sup_{x\in[0,t]} \|\mathbf{f}'(x)\|$. Then, 
\begin{equation*}
\left \| \int_0^t \mathbf f(x) \ dx - \sum_{j=0}^{n-1}\mathbf f(x_j)h \right \| \leq \frac{Mt^2}{2n}.
\end{equation*}
\end{lemma}
\begin{proof}
Noting that
\begin{equation*}
    \sum_{j=0}^{n-1}\mathbf{f}(x_j)h = \sum_{j=0}^{n-1}\int_{x_{j}}^{x_{j+1}} \mathbf{f}(x_j)  \ dx,
\end{equation*}
we can then expand 
\begin{align*}
\left \| \int_0^t \mathbf f(x) \ dx - \sum_{j=0}^{n-1}\mathbf f(x_j)h \right \| &= \left \| \int_0^{x_n} \mathbf{f}(x) \ dx - \sum_{j=0}^{n-1}\int_{x_{j}}^{x_{j+1}} \mathbf{f}(x_j) \ dx \right \| \\
&= \left \| \sum_{j=0}^{n-1} \int_{x_j}^{x_{j+1}} \mathbf{f}(x) - \mathbf{f}(x_j) \ dx \right \| \\
&\leq \sum_{j=0}^{n-1} \int_{x_j}^{x_{j+1}} \left \| \mathbf{f}(x) - \mathbf{f}(x_j) \right \| \ dx.
\end{align*}
Noting that $\|\mathbf f(x) - \mathbf f(x_j)\| \leq M\| x - x_j\|$ and changing the integration interval to $[0,h]$ in the above equation, we get
\begin{align*}
\left \| \int_0^t \mathbf f(x) \ dx - \sum_{j=0}^{n-1}\mathbf f(x_j)h \right \| \leq \sum_{j=0}^{n-1} \int_0^h M x \ dx
= \frac{Mnh^2}{2}
= \frac{Mt^2}{2n},
\end{align*}
which gives us the desired result.
\end{proof}

\begin{remark}
    Note that, for~\cref{lem:riemann} to hold, it is enough for $\mathbf f$ to be locally Lipschitz with Lipschitz constant $M$. However, we will be dealing with differentiable $\mathbf f$ in the rest of the paper, so we restrict $\mathbf f$ to be differentiable here as well.
\end{remark}

\subsection{Non-Markovian dynamics}
As briefly discussed in Section~\ref{sec:intro}, we consider dynamics that can be split into a Markovian and a non-Markovian part:
\begin{equation}
    \frac{d\mathbf u}{dt} = \mathbf f(t,\mathbf u(t)) + \int_0^t \mathbf k(t,\tau,\mathbf u(t),\mathbf u(\tau)) \ d\tau.
\end{equation}
Specifically, we consider linear dynamical systems with memory. These are commonly referred to as linear Volterra integro-differential equations ($\VIDE$s), and take the following form
\begin{equation} \label{eq:LVIDE}
    \frac{d\mathbf{u}}{dt} = A\mathbf{u} + \int_0^t K(t,\tau) \mathbf u(\tau) \ d\tau + \mathbf{b}.
\end{equation}
The memory term on the right-hand-side incorporates information from all previous times into the present dynamics of the solution and involves a memory kernel $K(t,\tau)$. We will further narrow down our focus to convolution linear $\VIDE$s:
\begin{equation*}
    \frac{d\mathbf u}{dt} = A \mathbf u + \int_0^t K(t-\tau) \mathbf u(\tau) \ d\tau + \mathbf b, \ \mathbf u(0) = \mathbf u_0,
\end{equation*}
with $\mathbf{u},\mathbf b \in \mathbb C^{N}$, $A \in \mathbb C^{N\times N}$, and $K\colon\mathbb R_+ \to \mathbb C^{N\times N}$ being the convolution memory kernel. The convolution kernels we consider here satisfy the following property:
\begin{equation}
    \lim_{x\rightarrow \infty} \|K(x)\| = 0,
\end{equation}
meaning that more recent times tend to have a larger effect on the dynamics compared to later times. Examples include the exponentially decaying kernel:
\begin{equation}
    K(x) = e^{-\gamma x}B,
\end{equation}
and the power-law kernel
\begin{equation}
    K(x) = \frac{1}{x^{\beta}} B, \ 0 < \beta < 1,
\end{equation}
with $B \in \mathbb C^{N\times N}$, for which we develop quantum algorithms in Section~\ref{sec:structured_kernels} for the special case of $B = -I$.

\section{Short-term memory condition}\label{sec:short-term}
Recall from Section~\ref{sec:techsum} that for the algorithm for general kernels, which we later describe in Section~\ref{sec:general_kernels}, we discretize the $\VIDE$ of Eq.~\eqref{eq:conv-LVIDE} and recast the discretized equations as a linear system. In this section we define the memory strength and the short-term memory condition, which is crucial for characterizing the complexity of the algorithms in Section~\ref{sec:general_kernels}.

\begin{definition}\label{def:memory-strength}
Consider the linear $\VIDE$ in Eq.~(\ref{eq:conv-LVIDE}). Let $\mu$ be the log-norm of $A$. The \textnormal{memory strength} is then defined as
$$
\mathcal M \coloneq \frac{1}{|\mu|}\int_0^\infty \|K(x)\| dx.
$$
We say that the system has \textit{short-term memory} if $\mathcal M < 1$.
\end{definition}

Using this condition, we will derive bounds on the growth of $\mathbf u$ as well as quantities needed for analysis of the algorithms in Section~\ref{sec:general_kernels}.

\subsection{Properties of $\VIDE$s with short-term memory}
\paragraph{Stability.}
We now prove a key result on the stability of convolution linear $\VIDE$s with short-term memory, which we need later.
\begin{lemma} \label{lem:stability}
Consider the convolution linear $\VIDE$ of Eq.~(\ref{eq:conv-LVIDE}). Suppose $\mu \coloneq \mu(A) < 0$ and $\mathcal{M}<1$. Define
$$
u_{\mathrm{max}} \coloneq \max \left \{ \|\mathbf{u}_0\|, \frac{\|\mathbf{b}\|}{|\mu|(1-\mathcal M)} \right \}.
$$
Then, $\|\mathbf{u}(t)\| \leq u_{\mathrm{max}}, \,\, \forall t \geq 0$.
\end{lemma}
\begin{proof}

Treating the memory term as an inhomogeneity, we have
\begin{equation}
    \mathbf u(t)=e^{At}\mathbf u_0+\int_0^t e^{A(t-\tau)}\left[\int_0^\tau K(\tau-s)\mathbf u(s)\,ds+\mathbf b\right]d\tau ,
\end{equation}
Fix $t\ge0$ and define 
\begin{equation}
     U(t):=\max_{0\le \tau\le t}\|\mathbf u(\tau)\|,
\end{equation}
For every $\tau\in[0,t]$,
\begin{align}
\left \|\int_0^\tau K(\tau-s)\mathbf u(s)\ ds\right\|
&\leq  U(t)\int_0^\infty\|K(x)\| \ dx\\
&=\mathcal M|\mu|  U(t).
\end{align}
By~\cref{lem:log-norm}, $\|e^{At}\|\leq  e^{-|\mu|t}$, so, for some $t^* \leq t$,
\begin{align}
    \|\mathbf u(t^*)\| &\leq e^{-|\mu|t^*} \|\mathbf u_0\| + \int_0^{t^*} e^{-|\mu|(t^*-\tau)}  \left \|\int_0^\tau K(\tau-s)\mathbf u(s)\ ds + \mathbf b\right\| \ d\tau \\
    &\leq e^{-|\mu|t^*} \|\mathbf u_0\| + (\mathcal M |\mu|U(t) + \|\mathbf b\| )\int_0^t e^{-|\mu|(t^*-\tau)} \ d\tau \\
    &= e^{-|\mu|t^*} \|\mathbf u_0\| + \left (1 - e^{-|\mu|t^*} \right ) \left(\mathcal M U(t) + \frac{\|\mathbf b\|}{|\mu|} \right ) \\
    &\leq \max \left \{ \|\mathbf u_0\|, \mathcal M U(t) + \frac{\|\mathbf b\|}{|\mu|} \right \}.
\end{align}
The last step is due to the fact that the penultimate expression is a linear interpolation between $\|\mathbf u_0\|$ and $\mathcal M U(t)+\|\mathbf b\|/|\mu|$, as $e^{-|\mu|t^*}\in(0,1]$). Maximizing the left-hand
side over $t^*\in[0,t]$ gives
\begin{equation}
    U(t)\leq\max\Big\{\|\mathbf u_0\|,\ M\,U(t)+\frac{\|\mathbf b\|}{|\mu|}\Big\}.
\end{equation}
If the maximum is attained by the first argument, then 
\begin{equation}
    U(t)\leq\|\mathbf u_0\|\leq u_{\max}.
\end{equation}
Otherwise, $U(t)\leq \mathcal M U(t)+\|\mathbf b\|/|\mu|$, and since $M<1$ we may rearrange to get
\begin{equation}
    U(t)\leq \frac{\|\mathbf b\|}{|\mu|(1-M)}\leq u_{\max}.
\end{equation} 
In both cases, $U(t)\leq u_{\max}$, and since
$t\geq0$ is arbitrary, $\|u(t)\|\leq u_{\max}$ for all $t\ge0$.

\end{proof}
This stability result allows us to prove two key Lemmas, which we make use of in the complexity analysis in Section~\ref{sec:general_kernels}.
\begin{lemma} \label{lem:bounded-derivative}
Consider the $\VIDE$ in Eq.~\eqref{eq:conv-LVIDE}. Let $\mathcal M <1$ and $\mu < 0$. Let 
\begin{equation*}
    \Xi \coloneq \sup_{t\geq 0, \tau \in[0,t]} \left \| \frac{d}{d\tau} K(t-\tau)\mathbf u(\tau)\right \|.
\end{equation*}
Suppose $\mathscr{K} \coloneq \sup_{x\geq 0} \|K(x)\| < \infty$ and $\mathcal K \coloneq \sup_{x\geq 0} \|K'(x)\| < \infty$ are finite. Then, we have
$$
\Xi \leq 3(\mathcal K + \mathscr K \|A\| )u_{\mathrm{max}},
$$
where $u_{\mathrm{max}}$ is defined as in~\cref{lem:stability}.
\end{lemma}
\begin{proof}
Using the $\VIDE$ to replace the derivative, we get
\begin{align}
    \left \|\frac{d}{d\tau}K(t-\tau)\mathbf{u}(\tau) \right\| &= \left \|\frac{dK(t-\tau)}{d\tau}\mathbf{u}(\tau) + K(t-\tau) \frac{d\mathbf{u}(\tau)}{d\tau} \right \| \\
    &\leq \|K'(t-\tau)\| \|\mathbf{u}(\tau)\| + \|K(t-\tau)\| \left \| A\mathbf{u}(\tau) + \int_0^\tau K(\tau-s)\mathbf{u}(s) \ ds + \mathbf{b} \right \|\\
    &\leq \|K'(t-\tau)\| \|\mathbf{u}(\tau)\| \\
    &+ \|K(t-\tau)\| \left [ \|A\| \|\mathbf u(\tau) \| + \int_0^\tau \|K(\tau - s)\| \| \mathbf{u}(s)\| \ ds + \|\mathbf{b}\| \right ].
\end{align}
Using definitions of $\mathscr{K}$ and $\mathcal{K}$ as well as \cref{lem:stability}, we get
\begin{align}
    \left \|\frac{d}{d\tau}K(t-\tau)\mathbf{u}(\tau) \right\| &\leq \mathcal K u_{\mathrm{max}} + \mathscr K\|A\|u_{\mathrm{max}} + \mathscr K u_{\mathrm{max}}\int_0^\infty \|K(x)\| \ dx  + \mathscr K\|\mathbf{b}\|\\
    &\leq \mathcal K u_{\mathrm{max}} + \mathscr K \|A\| u_{\mathrm{max}} + \mathscr K |\mu| \mathcal M u_{\mathrm{max}} + \mathscr K \|\mathbf{b}\|.
\end{align}
Using the fact that $\mathcal M |\mu| \leq \|A\|$ and $\|\mathbf{b}\| \leq |\mu| (1-\mathcal M) u_{\mathrm{max}} \leq \|A\|u_{\mathrm{max}}$, we get
\begin{align}
    \left \|\frac{d}{d\tau}K(t-\tau)\mathbf{u}(\tau) \right\| &\leq (\mathcal K + 3\mathscr K \|A\| )u_{\mathrm{max}} \leq 3(\mathcal K + \mathscr K \|A\| )u_{\mathrm{max}}.
\end{align}
\end{proof}

\begin{lemma} \label{lem:bounded-double-derivative}
Consider the $\VIDE$ in Eq.~\eqref{eq:conv-LVIDE}. Let $\mathcal M <1$ and $\mu < 0$. Let
\begin{equation}
    \Lambda \coloneq \sup_{t\geq 0}  \left \| \frac{d^2\mathbf{u}}{dt^2}\right \|.
\end{equation}
Then, we have
\begin{equation}
    \Lambda \leq  6(\|A\|^2 + \mathscr K ) u_{\mathrm{max}},
\end{equation}
where $u_{\mathrm{max}}$ and $\mathscr K$ are defined as in Lemmas~\ref{lem:stability} and~\ref{lem:bounded-derivative}, respectively.
\end{lemma}
\begin{proof}
    First, we upper bound the first derivative:
    \begin{align}
        \left \| \frac{d\mathbf u}{dt} \right \| &= \left \|A\mathbf u + \int_0^t K(t-\tau) \mathbf u(\tau) \ d\tau + \mathbf b \right \| \\
        &\leq \|A\| \|\mathbf{u}\| + \int_0^t \|K(t-\tau)\| \|\mathbf{u}(\tau)\| \ d\tau + \|\mathbf{b}\|\\
        &\leq (\|A\|+ \mathcal M |\mu|) u_{\mathrm{max}} + \|\mathbf{b}\|.
    \end{align}
    The equation for the second derivative is:
    \begin{align}
        \frac{d^2\mathbf u}{dt^2} = A\frac{d\mathbf{u}}{dt} + K(0)\mathbf{u}(t) + \int_0^t\partial_t K(t-\tau) \mathbf{u}(\tau) d\tau,
    \end{align}
    where we have used the Leibniz integral rule. Observing that $\partial_t K(t-\tau) = -\partial_\tau K(t-\tau)$ and using integration by parts, we get
    \begin{equation}
        \frac{d^2\mathbf u}{dt^2} = A\frac{d\mathbf{u}}{dt} +K(t)\mathbf{u}(0) + \int_0^t K(t-\tau) \frac{d\mathbf u}{d\tau} \ d\tau
    \end{equation}
    So,
    \begin{align}
        \left \| \frac{d^2\mathbf u}{dt^2} \right \| &= \left \| A\frac{d\mathbf{u}}{dt} +K(t)\mathbf{u}(0) + \int_0^t K(t-\tau) \frac{d\mathbf u}{d\tau} \ d\tau \right \| \\
        &\leq (\|A\| + \mathcal M |\mu|) \left ( (\|A\| + \mathcal M |\mu|)u_{\mathrm{max}} + \|\mathbf{b}\|  \right ) + \mathscr K u_{\mathrm{max}} \\
        &= \left [(\|A\| + \mathcal M |\mu|)^2 + \mathscr K \right] u_{\mathrm{max}} + (\|A\| + \mathcal M |\mu|)\|\mathbf{b}\|.
    \end{align}
    Using the fact that $\mathcal M |\mu| \leq \|A\|$ and $\|\mathbf{b}\| \leq |\mu| (1-\mathcal M) u_{\mathrm{max}} \leq \|A\|u_{\mathrm{max}}$, we get
    \begin{equation}
        \left \| \frac{d^2\mathbf u}{dt^2} \right \| \leq (6\|A\|^2 + \mathscr K ) u_{\mathrm{max}} \leq 6(\|A\|^2 + \mathscr K ) u_{\mathrm{max}}.
    \end{equation}
    This completes the proof.
\end{proof}

We will use the above result to prove an important Lemma that constrains the number of time steps required for the algorithm in Section~\ref{sec:general_kernels}. 

\begin{lemma} \label{lem:norm_stability_bound}
Consider the linear $\VIDE$ in Eq.~\eqref{eq:conv-LVIDE}. Suppose $\mu\coloneq \mu(A) < 0$ and $\mathcal{M}<1$. Let $T>0$ be the simulation time. Furthermore, let $\mathcal K$ be defined as in \cref{lem:bounded-derivative}. Suppose the interval $[0,T]$ is divided into $m$ sub-intervals of length $h$ such that $T=mh$. Then, for
\begin{equation}
    h \leq \frac{2|\mu|(1-\mathcal M)}{\mathcal K T + \|A\|^2},
\end{equation}
we have
\begin{equation} \label{eq:lem-time step-bound}
    \|I+Ah\| + h^2 \sum_{k=0}^{m-1}\|K((m-k)h)\|\leq 1.
\end{equation}
\end{lemma}
\begin{proof}
First, let us bound the first term on the left-hand-side:
\begin{align*}
\|I+Ah\|^2 = \|(I+Ah)(I+A^\dagger h)\| &= \|I+(A+A^\dagger)h + AA^\dagger h^2 \| \\
&\leq \|I+(A+A^\dagger)h\| + \|AA^\dagger\|h^2 \\
&\leq 1-2|\mu|h + \|A\|^2 h^2.
\end{align*}
So
\begin{equation}
    \|I+Ah\| \leq \sqrt{1-2|\mu|h + \|A\|^2 h^2}.
\end{equation}
Now let us bound the second term on the left-hand-side of Eq.~(\ref{eq:lem-time step-bound}) using the memory strength $\mathcal M$ and \cref{lem:riemann}:
\begin{align}
    h^2 \sum_{k=0}^{m-1} \|(K((m-k)h)\| &= h^2 \sum_{k=0}^{m-1} \|K((m-k)h)\| - h\int_0^{T}\|K(t-\tau)\|d\tau +h\int_0^{T}\|K(t-\tau)\|d\tau  \\
    &\leq h \left | \int_0^{T}\|K(t-\tau)\| \ d\tau - \sum_{k=0}^{m-1}h\|K((m-k)h)\| \right | + h\int_0^\infty \|K(t-\tau)\| d\tau \\
    &\leq \frac{1}{2}M_1mh^3 + \mathcal{M}|\mu|h
\end{align}
where $M_1 \coloneq \sup_{x\in[0,T]}\left|\frac{d \|K(x)\|}{dx} \right|$ can be bounded as
\begin{align}
M_1 = \sup_{x\in[0,T]}|\|K(x)\|'| \leq \sup_{x\geq 0}\|K'(x)\| = \mathcal K,
\end{align}
allowing us to get the following bound
\begin{equation}
    h^2 \sum_{k=0}^{m-1} \|(K((m-k)h)\| \leq \frac{1}{2}\mathcal K mh^3 + \mathcal M |\mu| h.
\end{equation}
Thus, combining things, we get
\begin{align}
    \|I+Ah\| + h^2 \sum_{k=0}^{m-1}\|K(kh)\|&\leq \sqrt{1-2|\mu|h + \|A\|^2 h^2} + \mathcal M |\mu| h + \frac{1}{2}\mathcal K mh^3 \\
    &\leq 1-|\mu|h + \frac{1}{2}\|A\|^2 h^2 + \mathcal{M}|\mu|h + \frac{1}{2} \mathcal KTh^2 \\
    &= 1 - (1-\mathcal{M})|\mu|h + \frac{1}{2}(\mathcal KT+\|A\|^2)h^2
\end{align}
If we pick
\begin{equation}
    h \leq \frac{2|\mu|(1-\mathcal M)}{\mathcal K T + \|A\|^2},
\end{equation}
then
\begin{align}
    1 - (1-\mathcal{M})|\mu|h + \frac{1}{2}(\mathcal KT+\|A\|^2)h^2 \leq 1,
\end{align}
thus completing the proof.
\end{proof}

\section{Simulating short-term memory dynamics with general kernels} \label{sec:general_kernels}
In this section, we propose quantum algorithms for solving Problems~\ref{prob:hist-state} and~\ref{prob:final-state} for $\VIDE$s satisfying $\mu<0$, as well as the short-term memory condition. Before we state the results, let us recall all the parameters that the algorithm complexity will depend on.
\begin{align}
    u_{\mathrm{max}}&= \max \left \{ \|\mathbf u_0 \|, \frac{\|\mathbf b\|}{|\mu|(1-\mathcal M)}\right \} \\
    \mathscr K &= \sup_{x\geq 0}\|K(x)\| \\
    \mathcal K &= \sup_{x\geq 0} \|K'(x)\| \\
    \Xi &= O((\mathcal K + \mathscr K \|A\|)u_{\mathrm{max}}) \\
    \Lambda &= O((\|A\|^2 + \mathscr K)u_{\mathrm{max}}) .
\end{align}
For history state preparation, the algorithm will have additional dependence on the following parameter:
\begin{equation}
     \mathfrak g(T) \coloneq \frac{1}{\sqrt{T}} \|\mathbf u\|_{L^2([0,T])} =  \left [\frac{1}{T}\int_0^T \|\mathbf u(\tau)\|^2 \ d\tau \right ]^{1/2}.
\end{equation}
Finally, the final state preparation algorithm has additional dependence on the following parameters:
\begin{equation}
    q(T)\coloneq \|\mathbf u(T)\|, \qquad g(T) \coloneq \frac{\max_{t\in[0,T]}\|\mathbf u(t)\|}{\|\mathbf u(T)\|}.
\end{equation}

We will prove the following result for history state preparation.
\begin{restatable}{theorem}{generalhistory} \label{thm:general-dynamics-algorithm-hist-state} 
Let $\varepsilon \in (0,\varepsilon_{\mathrm{max}})$, where
\begin{equation}
    \varepsilon_{\mathrm{max}}^2 =\frac{16T^3(\Lambda + \Xi T)^2}{\mathfrak g^2}\cdot \min \left \{\frac{\mathfrak g^2}{6\|A\|u_{\mathrm{max}}^2},\frac{2|\mu|(1-\mathcal M)}{\|A\|^2 + \mathcal KT}, \frac{1}{\alpha + \beta T} \right \}.
\end{equation}
There exists a quantum algorithm that solves Problem~\ref{prob:hist-state}

using the following complexity
\begin{align*}
\text{Queries to } U_A: & \,\, O\left (\frac{(\Lambda + \Xi T)^2 T^4}{\mathfrak g^2} \frac{\log(1/\varepsilon)}{\varepsilon^2} \right ) \\
\text{Queries to } \{U_{K_j}\}_j: &\,\, O\left (\frac{(\Lambda + \Xi T)^4 T^8}{\mathfrak g^4} \frac{\log(1/\varepsilon)}{\varepsilon^4} \right) \\
\text{Queries to } O_{\mathbf u}, O_{\mathbf b}: &\,\, O \left (\frac{(\Lambda + \Xi T)^2 T^4}{\mathfrak g^2} \frac{\log(1/\varepsilon)}{\varepsilon^2}  \right ),
\end{align*}
with an additional gate complexity of a factor of 
\begin{equation*}
    O\left ( \poly(\log T, \log(\Lambda + \Xi T),\log(1/\mathfrak g),\log(1/\varepsilon),\log\log(1/\varepsilon))\right )
\end{equation*}
times the queries to $U_A$. The algorithm is successful provided the block-encoding precision is
\begin{equation*}
    \delta = o \left (\frac{\mathfrak g^2 \varepsilon^3}{T^4(\Lambda+\Xi T)^2\log(1/\varepsilon)}  \right ).
\end{equation*}
\end{restatable}

We will also prove the following result for final state preparation
\begin{restatable}{theorem}{generalfinal} \label{thm:general-dynamics-algorithm-final-state}
Let $\varepsilon \in (0,\varepsilon_{\mathrm{max}})$, where
\begin{equation*}
    \varepsilon_{\mathrm{max}} = \frac{2(\Lambda T + \Xi T^2)}{q}\cdot\min \left \{  \frac{2|\mu|(1-\mathcal M)}{\|A\|^2+\mathcal KT}, \frac{1}{\alpha + \beta T}\right \}.
\end{equation*}
There exists a quantum algorithm that solves Problem~\ref{prob:final-state} when $\mu<0$ and $\mathcal M < 1$ using the following complexity:
\begin{align*}
\text{Queries to } U_A: & \,\, O\left ( g\frac{(\Lambda + \Xi T) T^2 }{q \varepsilon} \log(g/\varepsilon)\right ) \\
\text{Queries to } \{U_{K_j}\}_j: &\,\, O \left (g\frac{ (\Lambda + \Xi T)^2 T^4}{q^2 \varepsilon^2} \log(g/\varepsilon) \right ) \\
\text{Queries to } O_{\mathbf u}, O_{\mathbf b}: &\,\, O \left ( g\frac{(\Lambda  + \Xi T)T^2}{q\varepsilon} \log(g/\varepsilon)  \right ),
\end{align*}
with an additional gate complexity of a factor of
\begin{equation*}
    O\left (g\cdot \poly(\log T,\log(\Lambda + \Xi T),\log(1/q),\log(g/\varepsilon),\log\log(1/\varepsilon))\right )
\end{equation*}
times the queries to $U_A$. The algorithm is successful provided the block-encoding precision is
\begin{equation*}
\delta = o \left (\frac{q \varepsilon^2}{gT^2(\Lambda+\Xi T)\log(g/\varepsilon)}  \right ).
\end{equation*}
\end{restatable}
We will discuss the algorithms behind these theorems and perform the corresponding complexity analysis below.

\subsection{Approach}
At a high-level, the algorithms for both history state and final state preparation work by discretizing the time interval $[0,T]$ into $m$ intervals of length $h$. Then, the derivative and integral are discretized using forward Euler and a left Riemann sum, respectively. This results in a time marching scheme, where the solution at the subsequent time step depends on the solution at all previous time steps. We recast the time marching scheme into a linear system whose solution is a history state containing the numerical solution at the desired time steps. If the goal is to prepare the solution at the final time, we incur an additional cost from post-selection.

To show the efficiency of the quantum algorithm, we show that the condition number of the linear system is well-behaved, the error from the numerical discretization is controllable, and the linear system matrix and the right-hand-side vector can be prepared efficiently.
\subsection{Algorithm}
We discretize the time interval $[0,T]$ into $m$ intervals of equal length $h$, such that $T = mh$. Let $t_j \coloneq jh$ and $\mathbf{u}^{j} \coloneq \mathbf{u}(jh) $ for $j\in[m]_0$. The simplest approach is to discretize the time derivative via forward Euler:
\begin{equation} \label{eq:euler-disc}
    \frac{d\mathbf{u}}{dt} \Bigg |_{t = t_{j}} \approx \frac{\mathbf{u}^{j+1} - \mathbf{u}^{j}}{h},
\end{equation}
and the integral via a left Riemann sum:
\begin{equation} \label{eq:riemann-disc}
    \int_0^{t_{j}} K(t-\tau) \mathbf{u}(s) \ d\tau \approx h \sum_{k = 0}^{j-1} K_{j-k}\mathbf{u}^{k},
\end{equation}
where
\begin{equation} 
    K_l \coloneq K(lh).
\end{equation}
Substituting Eqs.~\eqref{eq:euler-disc} and~\eqref{eq:riemann-disc} into the $\VIDE$ results in the following time-marching scheme
\begin{equation}\label{eq:discretization}
    \mathbf{u}^{j+1} = (I+Ah)\mathbf{u}^{j} + h^2\sum_{k=0}^{j-1} K_{j-k}\mathbf{u}^k + h\mathbf{b},
\end{equation}
which we can recast as a linear system 
\begin{equation} \label{eq:linear_system}
    L\mathbf{y} = \mathbf{c},
\end{equation}
where $L\in \mathbb C ^{(m+p+1)N\times (m+p+1)N}$ can be written as
\begin{equation} \label{eq:L-decomp}
    L = L_0-L_1 - L_2,
\end{equation}
with
\begin{align}
    L_0 &\coloneq \sum_{j=0}^{m+p}|j\rangle \langle j | \otimes I \\
    L_1 &\coloneq \sum_{j=1}^{m}|j\rangle \langle j-1| \otimes (I+Ah) +  \sum_{l=2}^{m}\sum_{j=l}^{m}|j\rangle \langle j-l| \otimes K_{l-1}h^2 \\
    L_2 &\coloneq \sum_{j=m+1}^{m+p} |j\rangle \langle j-1| \otimes I.
\end{align}
In matrix form, $L$, $\mathbf y$, and $\mathbf c$ will look like
\begin{gather}
    L = \begin{bmatrix}
        I \\
        -(I+Ah) & I \\
        -K_{1}h^2 & -(I + Ah) &  I \\
        \vdots & \ddots &  \ddots & \ddots \\
        -K_{m-1}h^2 & \dots & -K_{1}h^2 &-(I+Ah) & I \\
        &&&&-I&I \\
        &&&&&\ddots & \ddots \\
        &&&&&&-I&I
    \end{bmatrix}, \ 
    \mathbf{y} = 
    \begin{bmatrix}
        \mathbf{y}_0 \\
        \mathbf{y}_1 \\
        \mathbf{y}_2 \\
        \vdots \\
        \mathbf{y}_m \\
        \mathbf{y}_{m+1} \\
        \vdots \\
        \mathbf{y}_{m+p}
    \end{bmatrix}, \ 
    \mathbf{c} = 
    \begin{bmatrix}
        \mathbf{u}_0 \\
        h\mathbf b \\
        h\mathbf b \\
        \vdots \\
        h \mathbf b \\
        0 \\
        \vdots \\
        0
    \end{bmatrix}.
\end{gather} 
One can observe that the solution to Eq.~\eqref{eq:linear_system} is 
\begin{equation}
    \mathbf{y}^{j}=\begin{cases}
        \mathbf{u}^j, &  0\leq j \leq m \\
        \mathbf{u}^m, & m+1 \leq  j \leq m+p
    \end{cases}.
\end{equation}

The algorithm is then divided into the following steps.
\begin{enumerate}
    \item Use $O_{\mathbf u}$ and $O_{\mathbf b}$ to prepare $\mathbf c$ as a quantum state $|\mathbf c\rangle$.
    \item Use the block-encodings of $A$ and $K$ to prepare the block-encoding of $L$. 
    \item Solve for $|\mathbf y \rangle$ via the $\QLSA$ algorithm in \cref{thm:QLSA}.
    \item If preparing the final state, post-select on the solution at time $T$. Use amplitude amplification~\cite{Brassard_2002} to boost the success probability.
\end{enumerate}

\subsubsection{Preparing $\mathbf{c}$}
The right-hand-side vector $\mathbf c$ of the linear system in Eq.~\eqref{eq:linear_system} can be prepared efficiently, as we will show below. The proof is inspired by techniques used in similar papers~\cite{berry2014high,berry2017quantum,krovi2023improved}.
\begin{lemma} \label{lem:state-prep}
There exists a quantum circuit that prepares a normalized version of $\mathbf{c}$ using one call to $O_{\mathbf u}$ and one call to $O_{\mathbf b}$ and gate complexity larger by a factor of $O(\polylog(m))$.
\end{lemma}
\begin{proof}
The normalized version of $\mathbf c$ is 
\begin{equation}
    \ket{\mathbf{c}} = \frac{1}{\sqrt{\|\mathbf{u}_0\|^2 + mh^2\|\mathbf{b}\|^2}}(\|\mathbf{u}_0\|\ket{0,\mathbf{u}_0} + \sum_{j=1}^m h\|\mathbf{b}\|\ket{j,\mathbf{b}}).
\end{equation}
To prepare this state, we first consider the initial state $\ket{0,0,0}$ and apply a rotation to the last qubit to obtain
\begin{equation}
    \frac{\|\mathbf{u}_0\|}{\sqrt{\|\mathbf{u}_0\|^2 + mh^2\|\mathbf{b}\|^2}}\ket{0,0,0} + \frac{\sqrt{m}h\|\mathbf{b}\|}{\sqrt{\|\mathbf{u}_0\|^2 + mh^2\|\mathbf{b}\|^2}}\ket{0,0,1}.
\end{equation}
Then, conditioned on the last qubit being in $0$, we implement $O_{\mathbf u}$ in the second register, and conditioned on it being in $1$, we implement $O_{\mathbf b}$ to get
\begin{equation}
    \frac{\|\mathbf{u}_0\|}{\sqrt{\|\mathbf{u}_0\|^2 + mh^2\|\mathbf{b}\|^2}}\ket{0,\mathbf{u}_0,0} + \frac{\sqrt{m}h\|\mathbf{b}\|}{\sqrt{\|\mathbf{u}_0\|^2 + mh^2\|\mathbf{b}\|^2}}\ket{0,\mathbf{b},1}.
\end{equation}
Then, conditioned on the last register being in $1$, we perform $O(\log m)$ gates over the first register to get the superposition
\begin{equation}
    \frac{\|\mathbf{u}_0\|}{\sqrt{\|\mathbf{u}_0\|^2 + mh^2\|\mathbf{b}\|^2}}\ket{0,\mathbf{u}_0,0} + \frac{\sqrt{m}h\|\mathbf{b}\|}{\sqrt{\|\mathbf{u}_0\|^2 + mh^2\|\mathbf{b}\|^2}}\sum_{j=1}^m\frac{1}{\sqrt{m}}\ket{j,\mathbf{b},1}.
\end{equation}
Finally, conditioned on the first register being in $j\neq 0$, we flip the last register from $1$ to $0$ and return that register to get
\begin{equation}
    \frac{\|\mathbf{u}_0\|}{\sqrt{\|\mathbf{u}_0\|^2 + mh^2\|\mathbf{b}\|^2}}\ket{0,\mathbf{u}_0} + \frac{h\|\mathbf{b}\|}{\sqrt{\|\mathbf{u}_0\|^2 + mh^2\|\mathbf{b}\|^2}}\sum_{j=1}^m\ket{j,\mathbf{b}}.
\end{equation}
This procedure requires one call to $O_{\mathbf u}$, one call to $O_{\mathbf b}$, and gate complexity larger by a factor of $O(\polylog(m))$.
\end{proof}
\subsubsection{Block-encoding $L$}
The block-encoding of $L$ can be constructed by first preparing the block-encoding of each diagonal band and then adding them together. 

\begin{lemma}\label{lem:block-encoding}
Suppose that $A$ has a $(\alpha, a, \delta)$ block-encoding 

and $K_i$ has a $(\beta,a,\delta)$ block-encoding 

for $i\in[m-1]$, where $m=T/h$ and
$$
h \leq \frac{1}{\alpha + \beta T}.
$$
We can construct a $(\tilde \alpha, \tilde a, \tilde \delta) $ block-encoding of $L$, where
$$
\tilde \alpha \leq 3, \quad \tilde a = a+\log (N(m+p)) + 3, \quad \delta = O(\delta),
$$
with $O(1)$ calls to the block-encoding of $A$, $O(m)$ calls to the block-encoding of $K$, and an additional $O(\poly(\log (m+p),\log(1/\delta)))$ gates.
\end{lemma}
\begin{proof}
First, note that, from \cref{lem:sparse-block-encoding}, matrices of the form
\begin{equation}
    \sum_{j=l}^r |j\rangle \langle j -l|,
\end{equation}
with $l\in[m]$ and $r\in\{m,m+p\}$ are 1-row sparse and 1-column sparse and can be block-encoded with a $(1,\log r + 3,\delta)$ triple and $O(\poly\left (\log r,\log(1/\delta) \right ))$ gates. We will now construct each band of the block-encoding and combine the resultant block-encodings.

Let $B_0$ be the diagonal band, which is easy to prepare as it is a trivial $(1,0,0)$ block-encoding. Let $B_1$ be the first lower off-diagonal band. Then
\begin{equation}
    B_1 = B_1^{(1)} + hB_1^{(2)},
\end{equation}
where
\begin{align}
B_1^{(1)} =- \sum_{j=1}^{m+p}|j\rangle \langle j -1|\otimes I, \quad \text{ and } B_1^{(2)} = -\sum_{j=1}^{m}|j\rangle \langle j -1|\otimes A.
\end{align}
Using Lemmas~\ref{lem:sparse-block-encoding} and~\ref{lem:tens-prod-BE}, $B_1^{(1)}$ can be prepared as a $(1,\log (m+p) + 3,\delta)$ block-encoding and $O(\poly(\log (m+ p),\log(1/\delta)))$ gates. Furthermore, $B_1^{(2)}$ can be prepared as a $(\alpha,a+\log m + 3,(\alpha + 1)\delta)$ block-encoding with one call to the block-encoding of $A$ and $O(\poly(\log m,\log(1/\delta)))$ gates.

All other off-diagonal bands are of the form $h^2B_{l}$, where
\begin{equation}
    B_l = -\sum_{j=l}^m|j\rangle \langle j - l| \otimes K_{l-1}, \ 2\leq l \leq  m.
\end{equation}
Using Lemmas~\ref{lem:sparse-block-encoding} and~\ref{lem:tens-prod-BE}, each $B_l$ can be encoded as a $(\beta,a+\log m + 3,(\beta + 1)\delta)$ block-encoding using $1$ call to the block-encoding of $K$ and gate complexity larger by a factor of $O(\poly(\log m, \log(1/\delta)))$.

Now we need to add the block-encodings correctly to obtain a block-encoding of $L$. We make use of \cref{lem:lin-comb-BE}. Let $\tilde B_j$ be the matrix at the top left block of the block-encoding of $B_j$. We want to implement the following linear combination of block-encodings:
\begin{equation}
    \tilde B_0 + \tilde B_1^{(1)} + \alpha h \tilde B_1^{(2)} + \sum_{l=2}^m\beta h^2 \tilde B_l
\end{equation}
To do this, we construct a state-preparation-pair unitary $P$, where
\begin{equation}
    P|0\rangle = \frac{1}{2+\alpha h + (m-1)\beta h^2}\left (|0\rangle + |1\rangle + \sqrt{\alpha h} |2\rangle + \sum_{l=2}^m \sqrt{\beta h^2} |l\rangle \right ).
\end{equation}
This requires $O(\polylog (m))$ gates. We then use \cref{lem:lin-comb-BE} to attain a $(\tilde \alpha, \tilde a, \tilde \delta)$ block-encoding of $L$, where
\begin{align}
    \tilde \alpha &= (2+\alpha h + (m-1)\beta h^2\\
    \tilde a &= a+\log (N(m+p)) + 3 \\
    \tilde \delta &= \left [1+(\alpha+1)h + (m-1)(\beta+1)h^2 \right ]\delta,
\end{align}
with $O(1)$ calls to the block-encoding of $A$, $O(m)$ calls to the block-encoding of $K$, and gate complexity larger by a factor of $O(\poly(\log(m+p),\log(1/\delta))$. Note that
\begin{align}
\tilde \alpha &=  2 + \alpha h + (m-1)\beta h^2 \leq 2 + \alpha h + \beta h T,
\end{align}
and, similarly
\begin{align}
    \tilde \delta = O \left (\alpha h + \beta hT \right ).
\end{align}
Setting $h$ according to the Lemma statement, we can see that $\tilde \alpha \leq 3$ and $\tilde \delta = O(\delta)$, thereby completing the proof.
\end{proof}

\subsubsection{Solving the linear system}
As $L$ and $\mathbf c$ can be block-encoded and prepared efficiently, the third step in the algorithm is to use $\QLSA$s to solve the linear system. The following shows the complexity of this step.
\begin{theorem} \label{thm:QLSA-complexity}
The linear system in Eq.~\eqref{eq:linear_system} can be solved to accuracy $\varepsilon' \in (0,1)$ with the following query complexity:
\begin{itemize}
    \item queries to $A$: $\left(\kappa(L) \log(1/\varepsilon')\right)$
    \item queries to $K$: $O \left (m\kappa(L) \log(1/\varepsilon') \right)$
    \item queries to $O_{\bf{u}}, O_{\bf{b}}$: $O \left (\kappa(L)\log(1/\varepsilon') \right )$ 
\end{itemize}
and gate complexity greater by a factor of at most
\begin{equation*}
    O\left ( \poly(\log (m+p),\log \kappa(L),\log(1/\varepsilon'),\log\log(1/\varepsilon'))\right ),
\end{equation*}
where $\kappa(L)$ is the condition number of $L$. This procedure requires the block-encoding error to be
\begin{equation*}
    \delta = o \left ( \frac{\varepsilon'}{\kappa(L)\log(1/\varepsilon')} \right ).
\end{equation*}
\end{theorem}
\begin{proof}
    We use the algorithm in \cref{thm:QLSA}, which makes
    \begin{equation}
        O(\kappa(L) \log(1/\varepsilon'))
    \end{equation}
    queries to the block-encoding of $L$ and the circuit for preparing $\mathbf c$. Using \cref{lem:block-encoding}, we get the queries to $A$ and $K$. The queries to $O_{\mathbf u}$ and $O_{\mathbf b}$ follow easily. We now choose the block-encoding error $\delta$ such that the solution to the linear system has error $\varepsilon'$. Since $O(\kappa(L)\log(1/\varepsilon'))$ calls to the block-encoding of $L$ and the block-encoding of $L$ has error $O(\delta)$ we choose
    \begin{align}
        \delta &= o\left (\frac{\varepsilon'}{\kappa(L)\log(1/\varepsilon')} \right ) 
    \end{align}
    The gate complexity follows from this choice of $\delta$. This completes the proof.
\end{proof}

\subsubsection{Success probability}
For the last step of the final-state-preparation algorithm, we need to post-select on the solution at the final time. In this section we will lower bound the success probability of post-selection.
\begin{theorem} \label{thm:post-selection}
    Let $\mathbf u(T)$ be the exact solution to the $\VIDE$ in Eq.~\eqref{eq:conv-LVIDE} at time $T$. Furthermore, let $\mathbf u^m$ be the numerical solution obtained from Eq.~\eqref{eq:discretization} at time step $m$ corresponding to time $T$. Suppose the global error from the numerical discretization is bounded by
    \begin{equation}
        \max_{j\in [m]}\|\mathbf{u}^j - \mathbf{u}(jh) \| \leq \frac{1}{2} \|\mathbf{u}(T)\|.
    \end{equation}
    Define
    \begin{equation}
        g\coloneq \frac{\max_{t\in[0,T]}\|\mathbf{u}(t)\|}{\|\mathbf{u}(T)\|}.
    \end{equation}
    Consider the solution $\mathbf y$ to the linear system of Eq.~\eqref{eq:linear_system}. Then the success probability of measuring a state $|\mathbf y_k\rangle$ for $k=[m+p]_0 \setminus [m-1]_0$ satisfies
    \begin{equation}
        P_{\mathrm{measure}} \geq \frac{p+1}{9(m+p+1)g^2}.
    \end{equation}
\end{theorem}
\begin{proof}
We decompose the (unnormalized) $\QLSA$ solution as
\begin{equation}
    \mathbf{y} = \|\mathbf y_{\mathrm{good}}\||\mathbf y_{\mathrm{good}}\rangle + \|\mathbf y_{\mathrm{bad}}\||\mathbf y_{\mathrm{bad}}\rangle,
\end{equation}
where
\begin{align}
|\mathbf y_{\mathrm{good}}\rangle \propto \sum_{j=m}^{m+p} |j\rangle |\mathbf y_m\rangle, \text{ and }
|\mathbf y_{\mathrm{bad}}\rangle \propto \sum_{j=0}^{m-1} |j\rangle |\mathbf y_j\rangle,
\end{align}
are normalized states. Then, the success probability is
\begin{equation}
    P_{\mathrm{measure}} = \frac{\|\mathbf y_{\mathrm{good}}\|^2}{\|\mathbf{y}\|^2}.
\end{equation}
First we upper bound the denominator:
\begin{align}
\|\mathbf y\|^2 = \sum_{j=0}^{m-1} \|\mathbf y_j\|^2 + (p+1)\|\mathbf y_m \|^2  \leq (m+p+1) \max_{i\in[m]_0} \|\mathbf y_i \|^2.
\end{align}
Suppose $i'$ is the index that maximizes $\|\mathbf y_i \|^2$. Then,
\begin{align}
\|\mathbf y_{i'} \| \leq \|\mathbf y_{i'} - \mathbf{u}(i'h) \| + \|\mathbf{u}(i'h)\| 
&\leq \|\mathbf y_m - \mathbf{u}(T)\| + \max_{t\in[0,T]} \|\mathbf{u}(t)\| \\
&\leq \frac{1}{2}\|\mathbf{u}(T)\| + \max_{t\in[0,T]} \|\mathbf{u}(t)\| \\
&\leq \frac{3}{2} \max_{t\in[0,T]} \|\mathbf{u}(t)\|.
\end{align}
So,
\begin{equation} \label{eq:numer}
\|\mathbf{y}\|^2 \leq \frac{9}{4}(m+p+1) g^2 \|\mathbf{u}(T)\|^2.
\end{equation}
We can now lower bound the numerator as
\begin{align} \label{eq:denom}
\|\mathbf{y}_{\mathrm{good}}\|^2 =(p+1) \|\mathbf{y}_m\|^2 \geq (p+1)(\|\mathbf{y}_m - \mathbf{u}(T)\| - \|\mathbf{u}(T)\|)^2 \geq \frac{1}{4}(p+1) \|\mathbf{u}(T)\|^2.
\end{align}
Putting Eqs.~\eqref{eq:numer} and~\eqref{eq:denom} together, we get
\begin{equation}
    P_{\mathrm{measure}} \geq \frac{(p+1)}{9(m+p+1)g^2}.
\end{equation}
This compeletes the proof.
\end{proof}

\subsection{Analysis}
We are now ready to analyze our algorithm. First, we will show that the algorithm is correct in the sense that, as the number of time steps increases, the discretization error decreases. Then, since the backbone of the algorithm is a $\QLSA$, we will analyze the complexity of our algorithm by bounding the condition number.
\subsubsection{Correctness}
In this section we will bound the error from the discretization of the $\VIDE$ in Eq.~(\ref{eq:conv-LVIDE}). We analyze the error by looking at the local truncation error and then using that to find the global truncation error. We can prove the following.
\begin{lemma}\label{lem:global-truncation-error}
    Suppose $\mathbf{u}(jh)$ is the exact solution to the linear $\VIDE$ of Eq.~\eqref{eq:conv-LVIDE} at time step $j\in[m]_0$, with $T=mh$ being the final simulation time and
    \begin{equation}
        h \leq \frac{2|\mu|(1-\mathcal M)}{\mathcal K T + \|A\|^2}.
    \end{equation}
    Let $\hat{\mathbf{u}}^{j}$ be the numerical solution obtained via Eq.~\eqref{eq:discretization} at time step $j$. We define the global truncation error at time step $j$ as
    \begin{equation}
        e_g^{j} \coloneq \|\mathbf{u}(jh) - \hat{\mathbf{u}}^{j}\|.
    \end{equation}

    Then, we have
    \begin{equation}
        e_g^m \leq \frac{\Lambda T + \Xi T^2}{2}h.
    \end{equation}
    with $\Xi$ and $\Lambda$ being defined in Lemmas~\ref{lem:bounded-derivative} and~\ref{lem:bounded-double-derivative}, respectively.
\end{lemma}
\begin{proof}
Let $\mathbf{u}(kh)$ be the exact solution at time step $k$ for $k\in [j]_0$. Then, the \textit{local approximation} at time step $j+1$ is given by
\begin{equation}
    \mathbf{v}^{j+1} \coloneq (I+Ah)\mathbf{u}(jh) + h^2\sum_{k=0}^{j-1}K((j-k)h) \mathbf{u}(kh) + h\mathbf{b}.
\end{equation}
Define the local truncation error at time step $j+1$ as
\begin{align}
    e_{l}^{j+1} \coloneq \|\mathbf{u}((j+1)h) - \mathbf{v}^{j+1} \|.
\end{align}
Using the triangle inequality, we bound the global truncation error at time step $j+1$ as:
\begin{align}
e_g^{j+1} = \|\mathbf{u}((j+1)h) - \hat{\mathbf{u}}^{j+1}\|
&\leq  \|\mathbf{u}((j+1)h) - \mathbf{v}^{j+1} \| + \| \mathbf{v}^{j+1} -\hat{\mathbf{u}}^{j+1}\| \\
&= e_l^{j+1} + \| \mathbf{v}^{j+1} -\hat{\mathbf{u}}^{j+1}\|. \label{eq:error-decomp}
\end{align}
First, let us bound the local truncation error. We use a Taylor expansion with a Lagrange remainder for $\mathbf{u}((j+1)h)$ on the interval $[jh,(j+1)h]$:
\begin{align}
    \mathbf{u}((j+1)h) &= \mathbf{u}(jh) + \mathbf{u}'(jh)h + \frac{\mathbf{u}^{''}(\xi)h^2}{2} \\
    &= \mathbf{u}(jh) + h\left [A\mathbf{u}(jh) + \int_0^{jh} K(jh-\tau) \mathbf{u}(\tau) \ d\tau + \mathbf{b}\right ] + \frac{\mathbf{u}^{''}(\xi)h^2}{2} \\
    &= (I+Ah)\mathbf{u}(jh) + h\int_0^{jh} K(jh-\tau) \mathbf{u}(\tau) \ d\tau + h\mathbf{b} + \frac{\mathbf{u}^{''}(\xi)h^2}{2},
\end{align}
for some $\xi\in [jh,(j+1)h]$. So, the local truncation error is\begin{align}
    e_l^{j+1} &= \left \|\mathbf{u}((j+1)h) - \mathbf{v}^{j+1} \right \| \\
    &= \left \| h\int_0^{jh} K(jh-\tau) \mathbf{u}(\tau) \ d\tau - h^2 \sum_{k=0}^{j-1}K((j-k)h)\mathbf{u}(kh) + \frac{\mathbf{u}''(\xi)h^2}{2}   \right \| \\
    &\leq \|\mathbf{u}''\| \frac{h^2}{2} + h \left \| \int_0^{jh} K(jh-\tau) \mathbf{u}(\tau) \ d\tau - h\sum_{k=0}^{j-1}K((j-k)h)\mathbf{u}(kh)\right \| \\
    &\leq \frac{\Lambda h^2}{2} + \frac{\Xi jh^3}{2} \\
    &= (\Lambda+\Xi jh)\frac{h^2}{2},
\end{align}
where we have used Lemmas~\ref{lem:riemann},~\ref{lem:bounded-derivative}, and~\ref{lem:bounded-double-derivative}.

The second term in Eq.~\eqref{eq:error-decomp} can be bounded as follows:
\begin{align}
    \|\mathbf{v}^{j+1} - \hat{\mathbf{u}}^{j+1} \| &= \left \| (I+Ah)(\mathbf{u}(jh) - \hat{\mathbf{u}}^{j}) + h^2\sum_{k=0}^{j-1} K((j-k)h)(\mathbf{u}(kh)-\hat{\mathbf{u}}^k)\right \| \\
    &\leq \|I+Ah\|e_g^{j} + h^2\sum_{k=0}^{j-1}\|K((j-k)h)\|e_g^{k} \\
    &\leq \left [\|I+Ah\| + h^2 \sum_{k=0}^{j-1}\|K((j-k)h)\| \right] \max_{k\in[j]_0} e_g^k \\
    &\leq \max_{k\in[j]_0}e_g^k,
\end{align}
where we have used \cref{lem:norm_stability_bound}.

So, putting things together, we get the following expression for the global truncation error at time step $j+1$:
\begin{equation} \label{eq:recurrence}
    e_g^{j+1}\leq e_l^{j+1}+\max_{k\in[j]_0}e_g^k.
\end{equation}
We will now show that
\begin{equation}
    \max_{k\in[j]_0}e_g^k \leq \sum_{k=1}^j e_l^k,
\end{equation}
which we prove inductively. We start with the base level. Note that $e_l^0 = e_g^0 = 0$ by definition. Then
\begin{align}
e_g^1 \leq e_l^1 + e_g^0 = e_l^1,
\end{align}
and so 
\begin{equation}
    \max_{k\in\{0,1\}}e_g^k = e_l^1.
\end{equation}
Now, for the inductive step, suppose that 
\begin{equation}
    \max_{k\in[n]_0}e_g^k \leq  \sum_{k=1}^n e_l^n,
\end{equation}
Then,
\begin{align}
    \max_{k\in[n+1]_0}e_g^k &= \max \left \{e_g^{n+1}, \max_{k\in[n]_0}e_g^k\right \} \\
    &\leq \max\left \{e_l^{n+1}+ \max_{k\in[n]_0}e_g^k, \max_{k\in[n]_0}e_g^k \right \} \\
    &\leq \max\left \{e_l^{n+1}+ \sum_{k=1}^n e_l^k, \sum_{k=1}^n e_l^k \right \} \\
    &= \sum_{k=1}^{n+1} e_l^k,
\end{align}
completing the inductive proof. So, the global truncation error at time step $j\in[m]$ is then
\begin{align}
e_g^m \leq  \sum_{k=1}^m e_l^k \leq \sum_{k=1}^m (\Lambda + \Xi kh) \frac{h^2}{2} \leq \frac{\Lambda T + \Xi T^2}{2}h.
\end{align}
The expression for $e_g^m$ follows from above. This completes the proof.
\end{proof}

\begin{remark}
    Note that the presence of the non-Markovian term in the $\VIDE$ results in a quadratic scaling of the error with $T$, whereas the Markovian component's contribution is only linear. This means that the presence of the memory term makes $\VIDE$s ``stiffer" in the sense that a more refined time stepping is required to achieve smaller errors. 
\end{remark}

We use the result above to bound the error for history state and final state preparation.

\begin{theorem} \label{thm:hist-state-error}
Let $\mathbf u(jh)$ be the exact solution to the $\VIDE$ of Eq.~\eqref{eq:conv-LVIDE} and $\hat{\mathbf u}^j$ be the numerical solution at time step $j\in [m]$. Let $\ket{\Psi}$ be the exact history state and $\ket{\hat{\Psi}}$ be the corresponding numerical solution, defined as
\begin{align*}
    |\Psi\rangle \coloneq \frac{1}{\mathcal N} \sum_{j=1}^m \|\mathbf u(jh)\||j\rangle |\mathbf u(jh)\rangle, \quad \text{ and }
    |\hat \Psi\rangle \coloneq \frac{1}{\hat{\mathcal N}} \sum_{j=1}^m \|\hat{\mathbf u}^j\||j\rangle |\hat{\mathbf u}^j\rangle.
\end{align*}
Suppose
\begin{equation*}
h \leq \frac{\mathfrak g^2}{6\|A\|u_{\mathrm{max}}^2}, \quad \text{ where }
\mathfrak g(T) := \left [\frac{1}{T}\int_0^T \|\mathbf u(\tau)\|^2 \ d\tau \right ]^{1/2}.
\end{equation*}
Then,
\begin{equation*}
    \| |\Psi\rangle - |\hat \Psi \rangle \| \leq 2\frac{\Lambda T^{3/2} + \Xi T^{5/2}}{\mathfrak g} \sqrt{h}.
\end{equation*}
\end{theorem}
\begin{proof}
From~\cref{lem:normalized_state_error}, we get
\begin{align}
\||\Psi\rangle - |\hat{\Psi}\rangle \| &\leq \frac{2}{\mathcal N} \left \|\sum_{j=1}^m |j\rangle (\|\mathbf u(jh)\||\mathbf u(jh)\rangle - \|\hat{\mathbf u}^j\||\hat{\mathbf u}^j\rangle) \right \| \leq \frac{2}{\mathcal N} \sum_{j=1}^m \|\mathbf u(jh) - \hat{\mathbf u}^j\| \leq \frac{1}{\mathcal N} (\Lambda T^2 + \Xi T^3). \label{eq:hist-error}
\end{align}
We will now lower bound the normalization factor:
\begin{align}
\mathcal N^2 = \sum_{j=1}^m \|\mathbf u(jh)\|^2 &= \left [\sum_{j=1}^m \|\mathbf u(jh)\|^2 h \right] \frac{m}{T}  \\
    &\geq \left [\int_0^T \|\mathbf u(\tau)\|^2 \ d\tau - \left |\int_0^T \|\mathbf u(\tau)\|^2 \ d\tau -  \sum_{j=1}^m \|\mathbf u(jh)\|^2 h  \right |  \right] \frac{m}{T}.
\end{align}
Let us upper bound the second term on the right-hand-side. Using \cref{lem:riemann}, we get
\begin{align}
    \left |\int_0^T \|\mathbf u(\tau)\|^2 \ d\tau -  \sum_{j=1}^m \|\mathbf u(jh)\|^2 h  \right | \leq \max_{t\in[0,T]}\left | \frac{d\|\mathbf u\|^2}{dt} \right | \frac{Th}{2}.
\end{align}
Let us bound the time derivative of the norm. A similar calculation as in \cref{lem:stability} yields
\begin{align}
    \left | \frac{d\|\mathbf u\|^2}{dt} \right | &\leq 2\|A\|\|\mathbf u\|^2 + 2\|\mathbf u\| u_{\mathrm{max}} |\mu| \mathcal M + 2\|\mathbf u\|\|\mathbf b\| \\
    &\leq 2\|A\| u_{\mathrm{max}}^2 + 2\mathcal M \|A\|u_{\mathrm{max}}^2 + 2\|A\|u_{\mathrm{max}}^2 \\
    &\leq 6\|A\| u_{\mathrm{max}}^2,
\end{align}
where we have also used \cref{lem:stability}.
Using this and setting $h$ according to the theorem statement yields
\begin{equation}
    \left | \int_0^T \|\mathbf u(\tau)\|^2 \ d\tau - \sum_{j=1}^m \|\mathbf u(jh)\|^2 h\right | \leq \frac{1}{2} \int_0^T \|\mathbf u(\tau)\|^2 \ d\tau.
\end{equation}
which implies
\begin{align}
\mathcal N^2 \geq \frac{1}{2} \frac{m}{T}\int_0^T \|\mathbf u(\tau)\|^2 \ d\tau = \frac{m}{2} \mathfrak g^2 = \frac{T}{2h}\mathfrak g^2.
\end{align}
Substituting this back into Eq.~\eqref{eq:hist-error} gives us
\begin{equation}
    \||\Psi\rangle - |\hat{\Psi}\rangle \| \leq  \sqrt{2}\frac{\Lambda T^{3/2} + \Xi T^{5/2}}{\mathfrak g} \sqrt{h},
\end{equation}
which is the desired result. This completes the proof.
\end{proof}

Note that $\mathfrak g$ for history state preparation is similar to $g$ in the context of final state preparation. While $g$ measures the decay of the solution, $\mathfrak g$ measures the time average of the history state.

\begin{remark}
    In the worst-case, when $\mathbf b = 0$ and the solution norm decays exponentially, $\mathfrak g(T) = \Omega(T^{-1/2})$. However, if the solution norm satisfies $\|\mathbf u(T)\| = \Theta(1)$, which typically occurs when $\mathbf b \neq 0$, we then have $\mathfrak g(T) = \Omega(1)$. 
\end{remark}

\begin{theorem}
    Let $|\mathbf u(T)\rangle = \mathbf u(T)/\|\mathbf u(T)\|$ be the encoding of the exact solution to the $\VIDE$ of Eq.~\eqref{eq:conv-LVIDE}, and $|\hat{\mathbf u}^m\rangle = \hat{\mathbf u}^m/\|\hat{\mathbf u}^m\|$ be the encoding of the numerical solution at time step $m$, with $T=mh$. Then,
    \begin{equation}
        \| |\mathbf u(T)\rangle - |\hat{\mathbf u}^m\rangle \| \leq \frac{(\Lambda T + \Xi T^2)h}{\|\mathbf u(T)\|}.
    \end{equation}
\end{theorem}
\begin{proof}
        Follows from~\cref{lem:normalized_state_error} and~\cref{lem:global-truncation-error}.
\end{proof}
\subsubsection{Complexity}
Finally, we will provide query and gate complexity analyses in terms of calls to the block-encodings of $A$ and $K$, as well as queries to $O_{\mathbf u}$ and $O_{\mathbf b}$. As we use a $\QLSA$ as a subroutine in our algorithm and its complexity scales with the condition number, we must bound the condition number of $L$ in order to characterize the complexity of our quantum algorithm.

\paragraph{Condition number.}
For the condition number, we have the following theorem.
\begin{theorem}
    Consider the linear system of Eq.~\eqref{eq:linear_system}. Let $T>0$ be the simulation time and $m\geq 1$ be the number of time steps with spacing $h$ such that $T=mh$. Suppose $\mu\coloneq \mu(A) <0$ and the $\VIDE$ satisfies $\mathcal{M}<1$. Furthermore, let $\mathcal K$ be defined as in \cref{lem:bounded-derivative}. Then, for
    \begin{equation}
        h \leq \frac{2|\mu|(1-\mathcal M)}{\mathcal K T + \|A\|^2},
    \end{equation}
    we have
    \begin{equation}
        \kappa(L) \leq 3(m+p+1).
    \end{equation} 
\end{theorem}
\begin{proof}
We know that $\kappa(L) = \|L\| \|L^{-1}\|$. Let us first bound $\|L\|$. Note that, from Eq.~\eqref{eq:L-decomp},
\begin{equation}
    \|L\| \leq \|L_0\| + \|L_1\| + \|L_2\|.
\end{equation}
Clearly, $\|L_0\| = \|L_2\| = 1$. As for $L_1$, we can see that
\begin{equation}
    \|L_1\| \leq \|I+Ah\| + \sum_{j=1}^{m-1} \|K(jh)\|h^2.
\end{equation}
Using \cref{lem:norm_stability_bound}, we obtain $\|L_1\| \leq 1$, so
\begin{equation}
    \|L\| \leq 3.
\end{equation}
To bound $\|L^{-1}\|$, we consider its action on a sub-normalized vector $|c\rangle $, i.e., we use the fact that
\begin{equation}
    \|L^{-1}\| = \sup_{\||{c}\rangle\| \leq 1} \| L^{-1}|{c}\rangle \|.
\end{equation}
Let
\begin{equation}
    |{c}\rangle = \sum_{k=0}^{m+p}|k\rangle |c_k\rangle ,
\end{equation}
where $|c_k\rangle$ satisfies
\begin{equation}
    \sum_{k=0}^{m+p}\||c_k\rangle \|^2 = \||c\rangle \|^2 \leq 1.
\end{equation}
We first bound $\|L^{-1}|k\rangle|c_k\rangle \|$, and then proceed to bound $\|L^{-1}|{c}\rangle \|$. For $k\in[m+p]_0$, we define
\begin{equation}
    |y_k\rangle \coloneq L^{-1}|k\rangle|c_k\rangle = \sum_{j=0}^{m+p} |j\rangle |y_k^j\rangle. 
\end{equation}
    
We divide this into two cases. First, we consider the case where $k\in [m+p]_0 \setminus [m]_0$. In that case,
\begin{equation}
    |y_k^j\rangle = \begin{cases}
        0, & j \in [k-1]_0 \\
        |c_k\rangle, & j \in[m+p]_0 \setminus [k-1]_0 
    \end{cases}.
\end{equation}
In such a case, it is easy to see that
\begin{align}
\|L^{-1}|k\rangle |c_k\rangle \|^2 = \sum_{l=k}^{m+p} \||c_k\rangle \|^2 &= (m+p-k+1)\||c_k\rangle \|^2 \\
&\leq (m+p+1)\||c_k\rangle \|^2, \ k\in [m+p]_0 \setminus [m]_0.
\end{align}

Second, we consider the case where $k \in[m]_0$. Then, we can see that
\begin{equation}
    |y_k^j\rangle = \begin{cases}
        0, & j \in [k-1]_0 \\
        |c_k\rangle, & j = k \\
        (I+Ah)|y^{j-1}_k\rangle + \sum_{l=k}^{j-2}K_{j-l-1}h^2 |y_k^l\rangle , & j \in [m]_0 \setminus [k]_0 \\
        (I+Ah)|y^{m-1}_k\rangle + \sum_{l=k}^{m-2}K_{m-l-1}h^2 |y_k^l\rangle, &  j \in [m+p]_0 \setminus [m]_0 
    \end{cases}.
\end{equation}

We will prove for this case that 
\begin{equation}
    \||y_k^j\rangle \| \leq \||c_k\rangle \| \ \forall j \in[m]_0.
\end{equation}
We shall prove this by strong induction. For the base case, $j=k$, this is trivial. Now let us assume that for all $k \leq j \leq k'$, $\||y_k^j\rangle \| \leq \||c_k\rangle \|$. Then,
\begin{align}
    \||y_k^{k'+1}\rangle \| &= \left \|(I+Ah)|y_k^{k'}\rangle + \sum_{l=k}^{k'-1}K_{k'-l}h^2 |y_k^l\rangle  \right \| \\
    &\leq \left (\|I+Ah\| + \sum_{l=k}^{k'-1}\|K_{k'-l}\|h^2 \right ) \||c_k\rangle \| \\
    &\leq \||c_k\rangle \|,
\end{align}
where we have used \cref{lem:norm_stability_bound}.
Thus, we similarly obtain
\begin{equation}
    \|L^{-1}|k\rangle |c_k\rangle \|^2 \leq  (m+p+1)\||c_k\rangle \|^2, \ k\in [m]_0.
\end{equation}
Putting these results together, we get

\begin{align}
    \|L^{-1}\| &= \sup_{\||c\rangle \| \leq 1} \|L^{-1}|c\rangle \| \\
    &\leq \sup_{\||c_k\rangle \| \leq 1} \left \| \sum_{k=0}^{m+p}L^{-1} |k\rangle |c_k\rangle \right \| \\
    &\leq \sup_{\||c_k\rangle \| \leq 1} \sum_{k=0}^{m+p} \|L^{-1}|k\rangle|c_k\rangle \| \\
    &= \sup_{\||c_k\rangle \| \leq 1}\sum_{k=0}^{m+p} (m+p+1)\||c_k\rangle \| \\
    &\leq (m+p+1).
\end{align}

Thus, we obtain
\begin{equation}
    \kappa(L) \leq 3(m+p+1).
\end{equation}
\end{proof}

\paragraph{Proof of complexity of history state preparation.} 
We now provide a proof of Theorem~\ref{thm:general-dynamics-algorithm-hist-state}.
\begin{proof}[Proof of Theorem~\ref{thm:general-dynamics-algorithm-hist-state}]
Let $|\Psi'\rangle$ be the output of the quantum algorithm and let $|\Psi\rangle$ be the exact history state. Furthermore, let $|\tilde \Psi\rangle$ be the (approximate) history state output from the $\QLSA$ given no discretization error.

There are two sources of error in the algorithm. The first is the discretization error, as described in~\cref{lem:global-truncation-error} and the second is the $\QLSA$ error, as described in~\cref{thm:QLSA-complexity}. The total error is the sum of each of these errors, as established by the triangle inequality.
\begin{align}
    \||\Psi'\rangle - |\Psi\rangle \| &\leq \| |\Psi'\rangle - |\tilde\Psi\rangle \| + \| |\tilde\Psi\rangle - |\Psi\rangle \| \\
    &\leq \varepsilon_{\mathrm{disc}} + \varepsilon_{\mathrm{qlsa}},
\end{align}
where the two error terms in the last line are the discretization and $\QLSA$ error, respectively.

We choose the parameters of the problem as follows
\begin{align}
    p &= 0 \\
    h &= \min \left \{ \frac{\mathfrak g^2 \varepsilon^2}{16T^3(\Lambda + \Xi T)^2}, \frac{\mathfrak g^2}{6\|A\|u_{\mathrm{max}}^2}, \frac{2|\mu|(1-\mathcal M)}{\|A\|^2 + \mathcal K T}, \frac{1}{\alpha + \beta T} \right \} = \frac{\mathfrak g^2 \varepsilon^2}{16T^3(\Lambda + \Xi T)^2} \\
    m &= \left\lceil\frac{T}{h}\right\rceil,
\end{align}
where the first line holds when $\varepsilon \in (0,\varepsilon_{\mathrm{max}})$. This choice of parameters guarantees that
\begin{equation}
    \varepsilon_{\mathrm{disc}} = \frac{\varepsilon}{2}.
\end{equation}
This choice of parameters guarantees that the constraints of~\cref{thm:post-selection},~\cref{lem:global-truncation-error}, and~\cref{thm:hist-state-error} are satisfied. We also set $\varepsilon_{\mathrm{qlsa}} = \varepsilon/2$, so the output of the quantum algorithm is $\varepsilon$-close to the exact history state. 

Substituting these parameters into \cref{thm:QLSA-complexity}, we get the desired bounds in the theorem statement.
\end{proof}
\begin{remark}
Note that the complexity of history state preparation can be reduced with respect to $T$ if we change the problem formulation to be of the form in Remark~\ref{rem:hist-state-prob}. In fact, the complexity would be identical to that of \cref{thm:general-dynamics-algorithm-final-state}, modulo the factor of $g$ needed for post-selection.
\end{remark}

\paragraph{Proof of complexity of final state preparation.}
We now provide a proof of Theorem~\ref{thm:general-dynamics-algorithm-final-state}.
\begin{proof}[Proof of Theorem~\ref{thm:general-dynamics-algorithm-final-state}]
Just as in the proof of~\cref{thm:general-dynamics-algorithm-hist-state}, the error sources in the quantum algorithm are the discretization error, $\varepsilon_{\mathrm{disc}}$, and the $\QLSA$ error, $\varepsilon_{\mathrm{qlsa}}$. Let $|\mathbf u(T)\rangle$ be the exact (normalized) solution at time $T$ and let $|\psi\rangle$ be the output of the quantum algorithm. Then,
\begin{equation}
    \||\psi\rangle - |\mathbf u(T)\rangle\| \leq \varepsilon_{\mathrm{disc}}+ \varepsilon_{\mathrm{qlsa}}.
\end{equation}
so we choose the parameters as follows
\begin{align}
    h &= \min \left \{\frac{q\varepsilon}{2(\Lambda T + \Xi T^2)}, \frac{2|\mu|(1-\mathcal M)}{\|A\|^2 + \mathcal K T}, \frac{1}{\alpha + \beta T} \right \} = \frac{q\varepsilon}{2(\Lambda T + \Xi T^2)} \\
    p &= m = \left\lceil \frac{T}{h} \right \rceil, 
\end{align}
where the first line holds when $\varepsilon \in (0,\varepsilon_{\mathrm{max}})$. This guarantees that
\begin{align}
    \varepsilon_{\mathrm{disc}} = o(\varepsilon),
\end{align} 
and that the constraints of~\cref{thm:post-selection} and~\cref{lem:global-truncation-error} are satisfied. Furthermore, we set
\begin{equation}
    \varepsilon_{\mathrm{qlsa}} = o\left (\frac{\varepsilon}{g} \right ),
\end{equation}
Substituting these values into the $\QLSA$ complexity in \cref{thm:QLSA-complexity}, we get that, to produce the history state, we need
\begin{equation}
    O\left ( \frac{(\Lambda + \Xi T) T^2 }{q \varepsilon} \log(g/\varepsilon)\right )
\end{equation}
calls to the block-encoding of $A$,
\begin{equation}
    O \left (\frac{ (\Lambda + \Xi T)^2 T^4}{q^2 \varepsilon^2} \log(g/\varepsilon) \right )
\end{equation}
calls to the block-encoding of $K$,
\begin{equation}
    O \left ( \frac{(\Lambda + \Xi T) T^2 }{q \varepsilon} \log(g/\varepsilon)  \right )
\end{equation}
calls to $O_{\mathbf u}$ and $O_{\mathbf b}$, and gate complexity larger by a factor of
\begin{equation}
    O\left ( \poly(\log T,\log (\Lambda + \Xi T), \log (1/q),\log(g/\varepsilon),\log\log(1/\varepsilon))\right ).
\end{equation}

The post-selection success probability is given by \cref{thm:post-selection}. As $m=p$, the success probability is lower bounded as $\Omega(1/g^2)$. We can boost this to $\Omega(1)$ using $O(g)$ rounds of amplitude amplification~\cite{Brassard_2002}. Adding this factor of $g$ to the complexities above amplifies the $\QLSA$ error such that the output of the algorithm is $\varepsilon$-close to $|\mathbf u(T)\rangle$ and yields the stated complexity of the algorithm.

This completes the proof.
\end{proof}

\section{Lower bound for simulating dynamics with general kernels}\label{sec:lb}
In the previous section, we developed an efficient algorithm for simulation convolution $\VIDE$s given by Eq.~\eqref{eq:conv-LVIDE} when $\mathcal M < 1$. In this section we will prove hardness results for simulating convolution $\VIDE$s with $\mathcal M\geq 1$. In particular, we will show that simulating $\VIDE$s with $\mathcal M \geq 1$ on quantum computers is intractable.

To prove this lower bound, we make use of state discrimination lower bounds, which characterize the hardness of distinguishing two quantum states with large overlap. We first briefly review state discrimination lower bounds and then proceed to show how non-Markovian dynamics can accelerate state discrimination.

\subsection{State discrimination lower bounds}
Suppose we have access to a black box that only prepares a state $|\psi\rangle$ or a state $|\phi\rangle$ such that $|\langle \psi|\phi \rangle | \geq 1-\varepsilon$. Our task is to determine which state the black box prepares. As unitary operations do not change the overlap of the two states, we need copies of the state, and we would like to distinguish the states with as few copies as possible. The result below is an information-theoretic lower bound on the hardness of this task.

\begin{lemma}[\cite{liu2021efficient}, Lemma 7] \label{lem:helstrom}
    Let $|\psi\rangle$,$|\phi\rangle\in \mathbb C^2$ be quantum states with $|\langle \psi|\phi \rangle | \geq 1-\varepsilon$. We are given an oracle that is promised to prepare either only $|\psi\rangle$ or only $|\phi\rangle$. Any bounded-error protocol for determining which state the black box prepares must use $\Omega(1/\varepsilon)$ calls to the black box.
\end{lemma}

\subsection{Accelerated state discrimination with non-Markovian dynamics}
In the context of dynamical systems, state discrimination was initially used as a tool to probe the power of nonlinear quantum mechanics~\cite{abrams1998nonlinear,aaronson2005np,childs2016optimal}.
Using state discrimination to prove lower bounds on quantum algorithms has been extensively used in the literature, for both linear and nonlinear differential equations~\cite{liu2021efficient,an2022theory,lewis2024limitations,brustle2025quantum,ameri2026quantum}. The main idea behind these proofs is to tap into the non-unitarity of the differential equations and show that an efficient algorithm for simulating them would violate the lower bound established in \cref{lem:helstrom}. Here, we use the same approach to prove lower bounds for simulating non-Markovian dynamics. First, we will show that there is an instance of the $\VIDE$ of Eq.~\eqref{eq:conv-LVIDE} with $\mathcal{M}\geq 1$ that separates two nearby states exponentially quickly.

\begin{lemma} \label{lem:lb-exact}
    There exists an instance of Eq.~(\ref{eq:conv-LVIDE}) with $\mathcal M\geq1$, with two normalized initial conditions with $1-\varepsilon$ overlap, with $\varepsilon \leq 0.03$, such that the two final (normalized) states after $T=O(\log(1/\varepsilon))$ evolution have an overlap no larger than $0.97$.
\end{lemma}
\begin{proof}
    Consider an instance of Eq.~(\ref{eq:conv-LVIDE}) with $\mathbf{b}=0$ and the following
    \begin{equation} \label{eq:matrices}
        A = \begin{bmatrix}
            -1 & 0 \\
            0 & -1
        \end{bmatrix}, \ 
        K(x) = e^{-\gamma x}\begin{bmatrix}
            0 & 0 \\
            0 & 1
        \end{bmatrix}, \
    \end{equation}
    with $\gamma>0$ being the decay rate, as well as the following initial condition
    \begin{equation}
        \mathbf{u}(0) = \begin{bmatrix}
            u_1(0) \\
            u_2(0)
        \end{bmatrix}.
    \end{equation}
    It is easy to verify that $\mathcal M = \gamma^{-1}$ for this system. We consider $\gamma \in(0,1]$, which also corresponds to $\mathcal M \in[1,\infty)$. This $\VIDE$ is simple enough that it is analytically solvable. The solution for the first component is
    \begin{equation}
        u_1(t) =  e^{-t}u_1(0).
    \end{equation}
    The second component has the following equation
    \begin{equation}
        \frac{du_2}{dt} = -u_2 + \int_0^t e^{-\gamma(t-\tau)} u_2 (\tau) \ d\tau.
    \end{equation}
    Taking the derivative gives
    \begin{align}
        \frac{d^2u_2}{dt^2} &= -\frac{du_2}{dt} + u_2 - \gamma\int_0^t e^{-(t-\tau)} u_2(\tau) \ d\tau \\
        &= -\frac{du_2}{dt} + u_2 - \gamma \left [\frac{du_2}{dt} + u_2 \right ]\\
        &=-(\gamma + 1)\frac{du_2}{dt} - (\gamma - 1) u_2,
    \end{align}
    which is a second-order $\ODE$ that we can solve. The solution is
    \begin{equation}
        u_2(t) = c_1e^{\gamma_+t} + c_2 e^{\gamma_-t},
    \end{equation}
    where
    \begin{equation} \label{eq:growth-rates}
        \gamma_\pm = \frac{-(\gamma + 1)\pm \sqrt{(\gamma + 1)^2 - 4(\gamma - 1)}}{2}.
    \end{equation}
    Since $\gamma \in(0,1]$, we have
    \begin{align}
        0 &\leq \gamma_+ < \frac{-1+\sqrt{5}}{2} < 1 \\
        -2 &\leq \gamma_- < \frac{-1-\sqrt{5}}{2} < -1.
    \end{align} 
    Using the initial condition, we can solve for the coefficients to obtain the unique solution
    \begin{equation}
        u_2(t) = \frac{1}{\gamma_+ - \gamma_-}\left ((1+\gamma_+) e^{\gamma_-t}-(1+\gamma_-)e^{\gamma_+t} \right )u_2(0) .
    \end{equation}
    We pick two set of trajectories, $\mathbf{u}^0 $ and $\mathbf{u}^\varepsilon$, with the following initial conditions
    \begin{equation} \label{eq:initial_conditions}
        \mathbf{u}^0(0) = \begin{bmatrix}
            1 \\
            0
        \end{bmatrix}, \ \mathbf{u}^\varepsilon(0) = \begin{bmatrix}
            \sqrt{1-\varepsilon^2} \\
            \varepsilon
        \end{bmatrix}.
    \end{equation}
    The trajectories evolve as follows:
    \begin{align}
        \mathbf{u}^0(t) &= \begin{bmatrix}
            e^{-t} \\
            0
        \end{bmatrix} \\
        \mathbf{u}^\varepsilon (t) &= \begin{bmatrix}
            \sqrt{1-\varepsilon^2} e^{-t} \\
            \frac{\varepsilon}{\gamma_+ - \gamma_-}\left [(1+\gamma_+) e^{\gamma_-t}-(1+\gamma_-)e^{\gamma_+t} \right ]
        \end{bmatrix}.
    \end{align}
    We define the normalized states $|\mathbf{u}^0(t)\rangle$ and $|\mathbf{u}^\varepsilon(t)\rangle$, which evolve as follows:
    \begin{align}
        |\mathbf{u}^0(t)\rangle  &= \begin{bmatrix}
            1 \\
            0
        \end{bmatrix} \\
        |\mathbf{u}^\varepsilon (t)\rangle  &= \frac{1}{\sqrt{\mathcal N(t)}}\begin{bmatrix}
            \sqrt{1-\varepsilon^2} e^{-t} \\
            \frac{\varepsilon}{\gamma_+ - \gamma_-}\left [(1+\gamma_+) e^{\gamma_-t}-(1+\gamma_-)e^{\gamma_+t} \right ]
        \end{bmatrix},
    \end{align}
    where
    \begin{equation} \label{eq:norm-factor-homog}
        \mathcal{N}(t)\coloneq (1-\varepsilon^2)e^{-2t} + \frac{\varepsilon^2}{(\gamma_+ - \gamma_-)^2} \left[(1+\gamma_+) e^{\gamma_-t}-(1+\gamma_-)e^{\gamma_+t} \right ]^2 .
    \end{equation}
    This gives the inner product:
    \begin{equation} \label{eq:inner-prod}
        |\langle \mathbf{u}^0(t)|\mathbf{u}^\varepsilon (t)\rangle |^2 = \frac{(1-\varepsilon^2)e^{-2t}}{\mathcal N(t)}.
    \end{equation}
    Note that
    \begin{equation}
        |\langle \mathbf{u}^0(0)|\mathbf{u}^\varepsilon (0)\rangle |= \sqrt{1-\varepsilon^2} \geq 1-\varepsilon.
    \end{equation}
    Let 
    \begin{equation}
        t^* = \frac{1}{\nu}\log \left (\frac{1}{\varepsilon} \right ),
    \end{equation}
    where $\nu>0$ will be determined shortly. Substituting this in Eq.~\eqref{eq:norm-factor-homog}, we get
    \begin{align}
        \mathcal{N}(t^*) = (1-\varepsilon^2) \varepsilon^{2/\nu} + \frac{\varepsilon^2}{(\gamma_+ - \gamma_-)^2} \left [ (1+\gamma_+)^2 \varepsilon^{-2\gamma_-/\nu} +(1+\gamma_-)\varepsilon^{-2\gamma_+/\nu} - 2(1+\gamma_+)(1+\gamma_-)\varepsilon^{-(\gamma_++\gamma_-)/\nu}\right].
    \end{align}
    Plugging this into Eq.~\eqref{eq:inner-prod} gives
    \begin{align}
       | \langle \mathbf{u}^0(t^*)|\mathbf{u}^{\varepsilon}(t^*)\rangle|^2 = \frac{1-\varepsilon^2}{1-\varepsilon^2 + \frac{(1+\gamma_+)^2}{(\gamma_+-\gamma_-)^2} \varepsilon^{2-2(1+\gamma_-)/\nu} +\frac{(1+\gamma_-)^2}{(\gamma_+-\gamma_-)^2}\varepsilon^{2-2(1+\gamma_+)/\nu} - \frac{2(1+\gamma_+)(1+\gamma_-)}{(\gamma_+-\gamma_-)^2}\varepsilon^{2-(2+\gamma_++\gamma_-)/\nu}}.
    \end{align}
    We set $\nu = 1+\gamma_+$, giving us
    \begin{align}
        | \langle \mathbf{u}^0(t^*)|\mathbf{u}^{\varepsilon}(t^*)\rangle|^2 &= \frac{1-\varepsilon^2}{1-\varepsilon^2 + \frac{(1+\gamma_-)^2}{(\gamma_+-\gamma_-)^2} + \frac{(1+\gamma_+)^2}{(\gamma_+-\gamma_-)^2} \varepsilon^{2-2(1+\gamma_-)/(1+\gamma_+)} - \frac{2(1+\gamma_+)(1+\gamma_-)}{(\gamma_+-\gamma_-)^2}\varepsilon^{2-(2+\gamma_++\gamma_-)/(1+\gamma_+)}   }\\
        &\leq \frac{1}{1+\left (\frac{1+\gamma_-}{\gamma_+-\gamma_-}\right )^2} \\
        &< \frac{65+5\sqrt{5}}{82}. 
    \end{align}
    So
    \begin{equation}
        | \langle \mathbf{u}^0(t^*)|\mathbf{u}^{\varepsilon}(t^*)\rangle| < \sqrt{\frac{65+5\sqrt{5}}{82}}< 0.97.
    \end{equation}
    Thus, it takes $O(\log(1/\varepsilon))$ time to reduce the overlap of the states from $1-\varepsilon$ to $1-\Omega(1)$.
\end{proof}

The result above applies to the exact time evolution of two trajectories. However, any quantum algorithm simulating $\VIDE$s will have some amount of error. Thus, we now proceed to prove the lower bound for bounded-error quantum algorithms simulating $\VIDE$s with $\mathcal M \geq 1$.

\begin{theorem} \label{thm:lb_approx}
    Assume $\mathcal M \geq 1$. Then there is an instance of the quantum $\VIDE$ problem such that any quantum algorithm producing a (normalized) quantum state approximating $|\mathbf{u}(T)\rangle $ with bounded error must require $e^{\Omega(T)}$ copies of the initial condition.
\end{theorem}
\begin{proof}
    Similar to \cref{lem:lb-exact}, consider the $\VIDE$ in Eq.~\eqref{eq:conv-LVIDE} with $\mathbf{b}=0$, matrices given by Eq.~\eqref{eq:matrices}, and a pair of states, $|\mathbf{u}^0\rangle$ and $|\mathbf{u}^\varepsilon \rangle$, initialized as in Eq.~\eqref{eq:initial_conditions}. From \cref{lem:lb-exact}, we know that, in $t=O(\log(1/\varepsilon))$, 
    \begin{equation}
        |\langle \mathbf{u}^0(t)|\mathbf{u}^\varepsilon(t)\rangle | < 0.97,
    \end{equation}
    with $\varepsilon \leq 0.03$. This implies the states have Euclidean distance
    \begin{align}
        \||\mathbf{u}^0(t)\rangle - |\mathbf{u}^\varepsilon(t)\rangle \| &= \sqrt{2(1-|\langle \mathbf{u}^0(t)|\mathbf{u}^\varepsilon(t)\rangle |)} \\
        &> 0.24 .
    \end{align}
    Suppose our quantum algorithm prepares $|\tilde{\mathbf{u}}^0(T)\rangle$ or $|\tilde{\mathbf{u}}^\varepsilon (T)\rangle$ to within distance $\delta$ of $|{\mathbf{u}}^0 (T)\rangle$ and $|{\mathbf{u}}^\varepsilon (T) \rangle$, respectively. In other words,
    \begin{equation}
        \||\mathbf{u}^0(T)\rangle - |\tilde{\mathbf{u}}^0(T)\rangle \| \leq \delta,
    \end{equation}
    and similarly for $|\mathbf{u}^\varepsilon\rangle$. The triangle inequality gives us
    \begin{align}
        \| |\tilde{\mathbf{u}}^0(T) \rangle - |\tilde{\mathbf{u}}^\varepsilon(T) \rangle \| &\geq \| |{\mathbf{u}}^0(T) \rangle - |{\mathbf{u}}^\varepsilon(T) \rangle \| - \| |{\mathbf{u}}^0(T) \rangle - |\tilde{\mathbf{u}}^0(T) \rangle \| - \| |{\mathbf{u}}^\varepsilon(T) \rangle - |\tilde{\mathbf{u}}^\varepsilon(T) \rangle \| \\
        &> 0.24 - 2\delta,
    \end{align}
    which is constant for, say, $\delta < 0.11$.
\end{proof}
\cref{thm:lb_approx} shows that the problem of simulating general $\VIDE$s with $\mathcal M\geq 1$ is intractable on quantum computers. However, would one be able to circumvent this lower bound by having additional knowledge about the $\VIDE$? We show that the answer to this question is yes, and we discuss more in the next section.

\section{Simulating dynamics with structured kernels}\label{sec:structured_kernels}
We have so far developed algorithms that could handle general memory kernels $K(x)$ (Theorem~\ref{thm:general-dynamics-algorithm-hist-state}) in Section~\ref{sec:general_kernels}. However, these algorithms required the short-term memory condition (Definition~\ref{def:memory-strength}) which can be restrictive, as it requires the Markovian part of the dynamics to be dissipative and the kernel to decay sufficiently rapidly. In Section~\ref{sec:lb}, we showed that this short-term memory condition was also necessary in order to obtain efficient simulation algorithms for $\VIDE$s.

In this section, we consider \emph{structured} kernels where the short-term condition may not be satisfied but an efficient simulation algorithm is still possible to obtain. Formally, we will work with the following definition.
\begin{definition}[Structured kernels]\label{def:struct_kernels}
Let $\varepsilon \in [0,1)$, $T > 0$ and $p \in \mathbb{N}$. We call the memory kernel $K(x)$ structured if $K(x) = -k(x) I$ where $k: \mathbb{R} \rightarrow \mathbb{R}$ and there exists a function $p: [0,1) \rightarrow \mathbb{R}_+$ such that $k$ can be approximated by a sum of $m = \ceil{p(\varepsilon)}$ many exponentials in the following sense
$$
\int_{x=0}^T |k(x) - \widetilde{k}(x)| dx \leq \varepsilon T, \,\, \text{ where } \,\, \widetilde{k}(x) = \sum_{j=1}^{m} w_j e^{-\gamma_j x} \text{ with } \,\, w_j, \gamma_j > 0 \,\, \forall j \in [m].
$$
We then call $\widetilde{k}$, a sum of exponentials ($\SOE$) approximation of $k$ with accuracy $\varepsilon$ and order $p(\varepsilon)$.
\end{definition}
Note that in the above definition $p(\varepsilon)$ is effectively a parameter that can be tuned to achieve a specified accuracy $\varepsilon \in (0,1)$. Moreover, note that kernel functions $k(x)$ that can be approximated by a sum of exponentials pointwise also satisfy the above definition. The above definition is useful to work with when the kernel function is singular for some values of $x$ such as the power-law kernel $k(x) = x^{-\beta}$ where $ \beta \in (0,1)$. This type of decomposition has been extensively studied in the literature~\cite{beylkin2005approximation,beylkin2010approximation,li2010fast,yan2017fast,jiang2017fast,mclean2018exponential} for a wide variety of kernels.

The main result in this section is an efficient quantum algorithm for simulating $\VIDE$s with structured memory kernels (Definition~\ref{def:struct_kernels}) even when $\calM > 1$ and the Markovian part of the dynamics is not dissipative, which we state below.
\begin{theorem}\label{thm:struct_kernels_hist_state}
Let $\varepsilon \in (0,1)$, $T > 0$. Consider Problem~\ref{prob:hist-state} with $K(x) = -k(x) I$ and $\mathbf{b}=0$ where the kernel function $k$ satisfies Definition~\ref{def:struct_kernels} with accuracy $\varepsilon_{\soe}$ and order $p:=p(\varepsilon_{\soe})$. Define
$\omega \coloneq \sum_{j=1}^p w_j$, and $\gamma_{\mathrm{max}} \coloneq \max_{i\in[p]}\gamma_i$. Suppose $A$ has a $(\alpha,a,\varepsilon_A)$-block-encoding and is Lyapunov stable such that 
$$
\expnorm_{\lyap}(A) := \inf_{\substack{P\succ 0, \, PA + A^\dagger P\preceq0}} \sqrt{\kappa(P)}
$$ 
Then, there is a quantum algorithm that outputs $|\widehat{\Psi}\ra$ with probability $\Omega(1)$ s.t. $\norm{\ket{\Psi} - |\widehat{\Psi}\ra} \leq \varepsilon$. Define 
$$
\kappa = T (\alpha + p(\varepsilon_{\soe}) \sqrt{2\omega} + \gamma_{\mathrm{max}}) \expnorm_{\lp}(A).
$$
The algorithm uses the following complexity
\begin{align*}
    \text{Queries to } U_A: & \,\, \widetilde{O}(\kappa^{5/2} \expnorm_{\lp}(A) \poly(\log(\kappa/\varepsilon))) \\
    \text{Queries to } O_{\mathbf{u}}: & \,\, \widetilde{O}(\kappa^{3/2} \expnorm_{\lp}(A) \log(\kappa/\varepsilon)) \\
    \text{Additional gate complexity }: & \,\, \widetilde{O}(\kappa^{5/2} \expnorm_{\lp}(A) \log(\kappa/\varepsilon) + \sqrt{\kappa \expnorm_{\lp}(A)} \poly(\log p, 1/\varepsilon_A)),
\end{align*}
where the block-encoding precision $\varepsilon_A$ and accuracy $\varepsilon_{\soe}$ are set as
$$
\varepsilon_A = o(\alpha^{-1} \varepsilon \kappa^{-2} \poly(\log(\kappa/\varepsilon))^{-1}), \quad \varepsilon_{\soe} = O\left(\frac{\varepsilon}{\kappa^{1/2} T^{2}\expnorm_{\lp}(A)^{3/2}}\right).
$$
\end{theorem}
A similar statement can also be made for final state preparation which we will state later. The quantum algorithm in the above theorem will be based on a technique which we call \textit{Markovianization}. Markovianization involves the embedding of the linear $\VIDE$ in a larger $\ODE$ by introducing new variables, thereby making the dynamics effectively Markovian in a larger state space. 
\subsection{Markovianization}
We now introduce the Markovianization approach that will lead to the algorithm for Theorem~\ref{thm:struct_kernels_hist_state}. We first discuss this approach for general kernels and discuss the associated challenges before discussing why this is amenable to structured kernels satisfying Definition~\ref{def:struct_kernels}.

Consider the linear $\VIDE$ of Eq.~\eqref{eq:conv-LVIDE}. Let us define the following variables:
\begin{align}
    \mathbf{y}_0 &\coloneq \mathbf{u} \\
    \mathbf{y}_1 &\coloneq \int_0^t K(t-\tau) \mathbf{u}(\tau) \ d\tau.
\end{align}
Then, the equations for $\mathbf{y}_0$ and $\mathbf{y}_1$ can be written as:
\begin{align}
    \frac{d\mathbf{y}_0}{dt} &= A\mathbf{y}_0 + \mathbf{y}_1 + \mathbf b \\
    \frac{d\mathbf{y}_1}{dt} &= K(0)\mathbf{y}_0  +\int_0^t K^{'}(t-\tau) \mathbf{u}(\tau) \ d\tau,
\end{align}
where we have used the Leibniz integral rule to obtain the equation for $\mathbf{y}_1$ and $K'$. Defining
\begin{equation}
\mathbf{y}_2 \coloneq \int_0^t K^{'}(t-\tau) \mathbf{u}(\tau) \ d\tau,
\end{equation}
and repeating this process, we obtain an infinite dimensional $\ODE$
\begin{align}
\frac{d}{dt}\begin{bmatrix}
    \mathbf{y}_0 \\
    \mathbf{y}_1 \\
    \vdots \\
    \mathbf{y}_p\\
    \vdots
\end{bmatrix} = \begin{bmatrix}
    A & I \\
    K(0) & & I \\
    \vdots & & &\ddots \\
    K^{(p-1)}(0) & & &  & I \\
    \vdots & & & &&\ddots
\end{bmatrix}
\begin{bmatrix}
    \mathbf{y}_0 \\
    \mathbf{y}_1 \\
    \vdots \\
    \mathbf{y}_p \\
    \vdots
\end{bmatrix} + 
\begin{bmatrix}
    \mathbf{b} \\
    0 \\
    \vdots \\
    0 \\
    \vdots
\end{bmatrix}, \qquad
\mathbf{y}(0) = \begin{bmatrix}
    \mathbf{u}_0 \\
    0 \\
    \vdots \\
    0 \\
    \vdots
\end{bmatrix},
\end{align}
where $K^{(j)}$ denotes the $j$-th derivative of $K$. We must then truncate this $\ODE$ so it is finite dimensional, and determine whether the dynamics of the truncated $\ODE$ system is sufficiently close to that of the $\VIDE$. This is the Markovianization process. Developing a quantum algorithm for Markovianization of general kernels is challenging for several reasons. First, the resultant matrix has a block-companion structure~\cite{barnett1981congenial}, and its stability properties are very difficult to analyze without additional information about the kernel~\cite{wimmer2009block,frederix2013algorithm,de2020stability}. Second, the structure of the block-companion matrix makes the task of bounding the truncation error very challenging, as each variable is coupled to $\mathbf{y}_0$. Lastly, to block-encode the resultant matrix, we would need access to the block-encodings of higher-order time derivatives of the memory kernel, which we may not necessarily have access to. 

We now discuss how these challenges can be circumvented for structured kernels (Definition~\ref{def:struct_kernels}), which we discuss next. We first describe Markovianization for the limiting case: when the kernel is a single exponential.

\subsubsection{Warm-up: The exponential kernel}
Consider the $\VIDE$ of Problem~\ref{prob:hist-state} with an exponentially decaying kernel $K(x) = - k(x) I$ and $\mathbf{b}=0$ with
\begin{equation}
    k(x) := e^{-\gamma x},
\end{equation}
where $\gamma \in \mathbb{R}_+$ is the kernel decay rate. This results in the following $\VIDE$
\begin{equation} \label{eq:vide-exponential-memory}
    \frac{d\mathbf{u}}{dt} = A\mathbf u - \int_0^t e^{-\gamma(t-\tau)}\mathbf{u}(\tau) \ d\tau + \mathbf{b}, \ \mathbf{u}(0) = \mathbf{u}_0,
\end{equation}
which we hope to simulate. We employ the Markovianization technique, and show that, in this special case, the system naturally truncates at a finite level. Let $\mathbf{y}_0 \coloneq \mathbf{u} $ and
\begin{equation}
    \mathbf{y}_1 \coloneq \int_0^t e^{-\gamma (t-\tau)} \mathbf{u}(\tau) \ d\tau.
\end{equation}
We can obtain the following equation governing the dynamics of $\mathbf{y}_1$ via the Leibniz integral rule:
\begin{equation}
    \frac{d\mathbf{y}_1}{dt} = \mathbf{y}_0- \gamma \mathbf{y}_1.
\end{equation}
This results in a system of $\ODE$s over the variables $\mathbf{y} = (\mathbf{y}_0, \mathbf{y}_1)$:
\begin{equation}\label{eq:expon-A-matrix} 
\frac{d\mathbf{y}}{dt} = \mathcal{A}\mathbf{y} + \mathbf{c}, \,\, \mathbf{y}(0) = \mathbf{y}_0,
\end{equation}
where
\begin{align*}
\mathbf{y} = \begin{bmatrix}
    \mathbf{y}^0 \\
    \mathbf{y}^1
\end{bmatrix}, \,\, 
\mathcal{A} = \begin{bmatrix}
    A & -I \\
    I & -\gamma I
\end{bmatrix}, \,\, 
\mathbf{c} = \begin{bmatrix}
    \mathbf{b} \\
    0
\end{bmatrix} \ \text{ with } \ \mathbf{y}_0 = \begin{bmatrix}
    \mathbf{u}_0 \\
    0
\end{bmatrix}.
\end{align*}
Note that another way to Markovianize Eq.~(\ref{eq:vide-exponential-memory}) would be to take its time derivative to obtain a second-order $\ODE$ -- just as we did in the proof of \cref{lem:lb-exact} -- which can then be reduced to a system of first-order $\ODE$s. In such a case the variables would be $\mathbf{y}_0 \coloneq \mathbf{u}$ and $\mathbf{y}_1 \coloneq \frac{d\mathbf{u}}{dt}$, and the resultant matrix $\calA$ would have a different structure. We find the original Markovianization approach much simpler to analyze and implement. 

We can also extend the describe Markovianization approach above to arbitrary exponentially decaying kernels $K(x) = e^{-\gamma x} B$, where $B$ is some arbitrary matrix. In such a case, $\calA$ would then have the following form
\begin{equation}\label{eq:expon-A-matrix-B}
\mathcal{A} = \begin{bmatrix}
    A & -I \\
    -B & -\gamma I
\end{bmatrix}.    
\end{equation}
It is easy to see, from Eq.~\eqref{eq:expon-A-matrix}, that $\expnorm (\mathcal A) \leq 1$ if $\mu(A)\leq 0$ and $\gamma \geq 0$. We can then use the quantum $\ODE$ solver of \cref{thm:ODE_solver_Krovi_final_state} to simulate the $\ODE$ system. For the more general case of Eq.~\eqref{eq:expon-A-matrix-B}, we could appeal to Gershgorin's circle theorem for block matrices~\cite{varga2011gervsgorin} which states that $\expnorm(\mathcal A) \leq 1$ if
\begin{equation} \label{eq:stable-exponential}
    \min\{|\mu(A)|,\gamma\} > \frac{1+\|B\|}{2}.
\end{equation}
The short-term memory condition for the exponentially decaying kernel is equivalent to:
\begin{equation} \label{eq:short_term_expon}
     \gamma |\mu| > \|B\|.
\end{equation}
Comparing Eq.~\textnormal{(\ref{eq:stable-exponential})} with Eq.~\textnormal{(\ref{eq:short_term_expon})}, we see that the former is a stronger condition. We do not provide the complexity analysis for these problems, as we consider a more general version below and which will be captured by Theorem~\ref{thm:struct_kernels_hist_state}.

\subsubsection{Structured kernels}\label{sec:markov-struct-kernels}
We now discuss the Markovianization technique for kernels satisfying Definition~\ref{def:struct_kernels}. We will see that the challenges mentioned earlier for general kernels are overcome with considered such structured kernels. Let $\widetilde{u}$ solve the $\VIDE$ with $k$ replaced by $k_p$. Instead of simulating the original $\VIDE$
\begin{equation} \label{eq:og-VIDE}
    \frac{d\mathbf u}{dt} = A\mathbf u - \int_0^t k(t-\tau) \mathbf{u}(\tau) \ d\tau, \ \mathbf u(0) = \mathbf u_0,
\end{equation}
we then aim to simulate the following approximated $\VIDE$
\begin{equation}\label{eq:soe-VIDE}
    \frac{d\widetilde{\mathbf u}}{dt} = A\widetilde{\mathbf u} - \sum_{j=1}^p w_j \int_0^t e^{-\gamma_j (t-\tau)}\widetilde{\mathbf u}(\tau) \ d\tau, \ \widetilde{\mathbf u}(0) = \mathbf u_0.
\end{equation}

Define the auxiliary variables of 
\begin{equation} \label{eq:markovianized-soe}
\mathbf{y}_0(t) \coloneq \widetilde{\mathbf{u}}, \quad \mathbf{y}_j(t) \coloneq \sqrt{w_j}\int_0^{t} e^{-\gamma_j (t-\tau)} \widetilde{\mathbf{u}}(\tau) \ d\tau, \qquad \forall j\in [p].
\end{equation}

Differentiating the auxiliary variables then gives the following $\ODE$
\begin{equation}\label{eq:markov-ode}
\frac{d}{dt}\underbrace{\begin{bmatrix}
    \mathbf{y}_0 \\
    \mathbf{y}_1 \\
    \vdots \\
    \mathbf{y}_p
\end{bmatrix}}_{:= \mathbf{y}} = \underbrace{\begin{bmatrix}
    A & -\sqrt{w_1}I  & \dots & -\sqrt{w_p}I \\
    \sqrt{w_1}I & -\gamma_1I &   & \\
    \vdots & &  \ddots  &   \\
    \sqrt{w_p}I & &  & -\gamma_pI
\end{bmatrix}}_{:= \calA} \mathbf{y} + \underbrace{\begin{bmatrix}
    \mathbf b \\
    0 \\
    \vdots \\
    0
\end{bmatrix}}_{:=\mathbf{c}}, \quad \text{with } \mathbf y(0) = \begin{bmatrix}
\mathbf u_0 \\
0 \\
\vdots \\
0
\end{bmatrix},
\end{equation}
The $\ODE$ to be solved is thus
$$
\frac{d\mathbf{y}}{dt} = \mathcal A \mathbf{y} + \mathbf{c}, \ \mathbf y(0) = \mathbf y_0.
$$

\subsection{Algorithm and analysis}
We now work towards giving a proof of Theorem~\ref{thm:struct_kernels_hist_state} using the technique of Markovianization discussed in the earlier section for structured kernels satisfying Definition~\ref{def:struct_kernels}.
\subsubsection{Main algorithm}\label{sec:algo-struct-kernels}
For structured kernels satisfying Definition~\ref{def:struct_kernels}, we employ the Markovianization technique discussed in Section~\ref{sec:markov-struct-kernels} to the $\SOE$ approximation of $k(x)$ (Definition~\ref{def:struct_kernels}). This converts the $\VIDE$ to be solved to an inflated system of $\ODE$s as given in Eq.~\eqref{eq:markov-ode}. We can then use a quantum $\ODE$ solver (Theorem~\ref{thm:ODE_solver_Krovi_history_state}) to solve for the resulting $\ODE$. To ensure that the so obtained solution tracks the original solution, it is however necessary to then decide the accuracy associated with the $\SOE$ approximation which we will do in the next sections as part of the analysis of this algorithm. 

\paragraph{Block-encoding.} To instantiate Theorem~\ref{thm:ODE_solver_Krovi_history_state}, we need to provide a block-encoding of $\calA$ given a block-encoding of the Markovian part of the original $\VIDE$ i.e., to $A$. We discuss this below.
\begin{claim} \label{lem:BE-soe-matrix}
Let $\varepsilon \in (0,1)$. Consider the context of Theorem~\ref{thm:struct_kernels_hist_state} and Definition~\ref{def:struct_kernels}. Define
$p:= \ceil{p(\varepsilon)}$, $\omega \coloneq \sum_{j=1}^p w_j$, and $\gamma_{\mathrm{max}} \coloneq \max_{i\in[p]}\gamma_i$.  Suppose $A$ has a $(\alpha,a,\varepsilon_A)$ block-encoding. Then the matrix $\mathcal A$ in Eq.~\eqref{eq:markov-ode} can be encoded as a $(\alpha',a',\varepsilon_{\mathcal A})$ block-encoding, where
\begin{align*}
        \alpha' &= \alpha + p\sqrt{2\omega} + \gamma_{\mathrm{max}}, \quad a'= a + \log(pN) + 3  ,\quad \varepsilon_{\mathcal A}= (\alpha+3)\varepsilon_A
    \end{align*}
    block-encoding using $1$ call to the block-encoding of $A$ and gate complexity larger by a factor of 
    \begin{equation*}
        O\left (\poly(\log p, \log (1/\varepsilon_A)) \right ).
    \end{equation*}
\end{claim}
\begin{proof}
    We first block-encode the top-left component of $\mathcal{A}$, which is $|0\rangle \langle 0| \otimes A$. The matrix $|0\rangle \langle 0|$ can be prepared via \cref{lem:sparse-block-encoding} as a $(1,\log p + 3,\varepsilon_A)$ block-encoding with $O(\poly(\log(p),\log(1/\varepsilon_A)))$ additional gates. Using \cref{lem:tens-prod-BE}, the tensor product can be prepared as a $(\alpha,a+\log p + 3,(\alpha + 1)\varepsilon_A)$ block-encoding with $1$ call to the block-encoding of $A$ and $O(\poly(\log(p),\log(1/\varepsilon_A)))$ additional gates. 
     
    We now split the rest of the matrix into two parts. First, we block encode
    \begin{equation}
        \begin{bmatrix}
            0 & -\sqrt{w_1} & \dots & -\sqrt{w_p} \\
            \sqrt{w_1} \\
            \vdots \\
            \sqrt{w_p}
        \end{bmatrix} \otimes I
    \end{equation}
    The identity on the right-hand-side is trivially preparable as a $(1,0,0)$ block-encoding. The matrix containing weights is $p$-row sparse and $p$-column sparse, and can be prepared using \cref{lem:sparse-block-encoding} as a 
    \begin{equation}
        \left (p\sqrt{2\omega},\log p + 3,\varepsilon_A \right)
    \end{equation} 
    block-encoding with $O(\poly(\log(p),\log(1/\varepsilon_A)))$ additional gates. Taking the tensor product with the identity via \cref{lem:tens-prod-BE} gives a block-encoding with the same triple.
    
    Then, we block-encode
    \begin{equation}
        \begin{bmatrix}
            0\\
            & \gamma_1 \\
            & & \ddots \\
            &&& \gamma_p
        \end{bmatrix} \otimes I
    \end{equation}
    Using \cref{lem:sparse-block-encoding} and \cref{lem:tens-prod-BE}, this is also preparable as a $(\gamma_{\mathrm{max}},\log p + 3, \varepsilon_A)$ block-encoding using $O(\poly(\log(p),\log(1/\varepsilon_A)))$ additional gates.
    We then add the block-encodings as follows. First we construct a two-qubit state-preparation unitary $P$ such that
    \begin{equation}
        P|0\rangle = \frac{1}{\alpha + p\sqrt{2\omega} + \gamma_{\mathrm{max}}}\left (\sqrt{\alpha}|0\rangle + \sqrt{p}{(2\omega)^{1/4}}|1\rangle + \sqrt{\gamma_{\mathrm{max}}}|2\rangle \right )
    \end{equation}
    Using \cref{lem:lin-comb-BE}, we get a
    \begin{equation}
        \left ( \alpha + p\sqrt{2\omega} + \gamma_{\mathrm{max}}, a + \log(pN) + 3, (\alpha+3)\varepsilon_A  \right )
    \end{equation}
    block-encoding using $1$ call to the block-encoding of $A$ and 
    \begin{equation}
        O\left (\poly(\log p, \log (1/\varepsilon_A)) \right )
    \end{equation}
    additional gates.
\end{proof}

\subsubsection{Properties}
We now discuss properties of the dynamical operator $\calA$ that will be used later as part of the analysis.
\paragraph{Norm. } We first bound the norm of the dynamical operator $\mathcal A$. 
\begin{claim}\label{claim:norm-matrix-soe}
Consider the context of Theorem~\ref{thm:struct_kernels_hist_state} and Definition~\ref{def:struct_kernels}. Define
$\omega \coloneq \sum_{j=1}^p w_j$, and $\gamma_{\mathrm{max}} \coloneq \max_{i\in[p]}\gamma_i$. Then,
$$
\|\mathcal A \| = O\left ( \|A\| +\gamma_{\mathrm{max}} + \sqrt{\omega} \right )
$$
\end{claim}
\begin{proof}
We split $\mathcal{A}$ into $\mathcal H + \mathcal G$, where $\mathcal H = (\mathcal A + \mathcal A^\dagger)/2$ and $\mathcal G = (\mathcal A - \mathcal A^\dagger)/2$. Note that $\mathcal H$ is block diagonal, so its norm is
\begin{equation}
    \|\mathcal H\| = \max\left\{ \|A\|,\gamma_{\mathrm{max}} \right\}.
\end{equation}
Now, note that
\begin{equation}
    \mathcal G = \begin{bmatrix}
        0 & -\sqrt{w_1} & \dots & -\sqrt{w_p} \\
        \sqrt{w_1} \\
        \vdots \\
        \sqrt{w_p}
    \end{bmatrix} \otimes I,
\end{equation}
whose norm is bounded by its Frobenius norm:
\begin{align}
\|\mathcal G\| \leq \|\mathcal G\|_{\mathrm F} = \sqrt{2\sum_{j=1}^p w_j} = \sqrt{2\omega}.
\end{align}
Putting things together, we get
\begin{align}
\|\mathcal A\| &\leq \|\mathcal H\| + \|\mathcal G\| \leq \max\{ \|A\|, \gamma_{\mathrm{max}}\} + \sqrt{2\omega} =O\left ( \|A\|+\gamma_{\mathrm{max}} + \sqrt{\omega} \right).
\end{align}
\end{proof}

\paragraph{Norm of the exponential.}
\begin{claim}\label{claim:expnorm-soe}
Consider the context of Theorem~\ref{thm:struct_kernels_hist_state} and Definition~\ref{def:struct_kernels}. Let $\calA$ be as defined in Eq.~\eqref{eq:markov-ode}. Then, the norm of the exponential of $\calA$ is bounded as follows:
$$
\expnorm(\calA) \leq \expnorm_{\lp}(A).
$$
If $\mu(A) \leq 0$ then $\expnorm(\calA) \leq 1$.
\end{claim}
\begin{proof}
Let $P = P^\dagger \succ 0$ be a matrix satisfying $PA + A^\dagger P \preceq 0$. Let $\calP = I_{p+1} \otimes P = \diag(P,P,\ldots,P)$. Consider
$$
\mathcal{B} = \calP^{1/2} \calA \calP^{-1/2} = \begin{bmatrix}
    P^{1/2} A P^{-1/2} & -\sqrt{w_1}I  & \dots & -\sqrt{w_p}I \\
    \sqrt{w_1}I & -\gamma_1I &   & \\
    \vdots & &  \ddots  &   \\
    \sqrt{w_p}I & &  & -\gamma_pI
\end{bmatrix},
$$
where we have used the fact that every off-diagonal block of $\calA$ is a scalar multiple of $I$ and thus remains unchanged by the transformation. We note that its Hermitian part is
\begin{equation}\label{eq:calB_leq0}
\frac{\mathcal{B} + \mathcal{B}^\dagger}{2} = \begin{bmatrix}
    P^{-1/2}(P A + A^\dagger P) P^{-1/2} & 0  & \dots & 0 \\
    0 & -\gamma_1I &   & \\
    \vdots & &  \ddots  &   \\
    0 & &  & -\gamma_pI
\end{bmatrix} \preceq 0,    
\end{equation}
since $PA + A^\dagger P \preceq 0$ by definition and $\gamma_j > 0 \forall j \in [p]$. We now note that 
$$
\calA = \calP^{-1/2} \mathcal{B} \calP^{1/2} \implies \expnorm(\calA) \leq \norm{\calP^{-1/2}} \cdot \expnorm(\mathcal{B}) \cdot \norm{\calP^{1/2}} \leq \sqrt{\norm{\calP^{-1}}\norm{\calP}} \leq \sqrt{\kappa(P)},
$$
where we used submultiplicativity of norms in the first inequality after the implication, used Fact~\ref{fact:expnorm_bound} along with Eq.~\eqref{eq:calB_leq0} in the second inequality and finally the definition of $\calP$. The result follows after taking an infimum over the right hand side of the above expression over all $P$ such that $P = P^\dagger \succ 0$ and satisfying $PA + A^\dagger P \preceq 0$. If $\mu(A) \leq 0$, it then suffices to consider $P = I$ and then the statement follows. This completes the proof.
\end{proof}

\subsubsection{Correctness}
In this section, we describe the parameter instantiations in Theorem~\ref{thm:struct_kernels_hist_state} to ensure that the outputted state from the algorithm discussed in Section~\ref{sec:algo-struct-kernels} is indeed close to the desired history state. To do this, we first show that replacing the memory kernel with its $\SOE$ approximation does not incurr too much error.
\begin{claim} \label{claim:soe-approx-integral}
Consider the context of Theorem~\ref{thm:struct_kernels_hist_state}. Let $k$ be a kernel function as in Definition~\ref{def:struct_kernels} with accuracy $\varepsilon_{\soe} \in (0,1)$ and order $p:=\ceil{p(\varepsilon_{\soe})}$. For $\varepsilon_{\soe} \leq 1/(\expnorm(\calA) \cdot T^2)$, we have
\begin{enumerate}[$(i)$]
    \item $\sup_{t \in [0,T]} \norm{\mathbf{u}(t)} \leq 2 \expnorm(\calA) \norm{\mathbf{u}_0}$,
    \item $\left \| \int_0^t k(t-\tau) \mathbf{u}(\tau) \ d\tau - \int_0^t \sum_{j=1}^p w_j e^{-\gamma_j (t-\tau)} \mathbf u(\tau) \ d\tau \right \| \leq 2 \varepsilon_{\soe} t \expnorm(\calA) \norm{\mathbf{u}_0},\quad \forall t \in [0,T]$.
\end{enumerate}
\end{claim}
\begin{proof}
Let $k(x)$ denote the true kernel function and let $\widetilde{k}(x) = \sum_{j=1}^p w_j e^{-\gamma_j x}$ be the approximation. Our strategy will be to introduce variables corresponding to the kernel function $\widetilde{k}(x)$ but still exactly solves for the true solution $\mathbf u(t)$. This will then require us to add a corresponding forcing term to ensure this. To determine this, let 
\begin{equation}\label{eq:what_if_variables}
\mathbf{v}_0(t) := \mathbf{u}(t), \quad \mathbf{v}_j(t):= \sqrt{w_j} \int_{0}^t e^{-\gamma_j(t-\tau)} \mathbf{u}(\tau) d\tau, \,\, \forall j \in [p],
\end{equation}
and let $\mathbf{v} = (\mathbf{v}_0,\ldots,\mathbf{v}_p)^T$. We will now describe the governing equations of their dynamics. For $\mathbf{v}_0$, we have
\begin{align}
    \frac{d \mathbf{v}_0(t)}{dt} = \frac{d \mathbf{u}}{dt} &= A \mathbf{u}(t) - \int_{0}^t k(t-\tau)\mathbf{u}(\tau) d\tau \nonumber \\ 
    &= A \mathbf{u}(t) - \int_{0}^t \widetilde{k}(t-\tau)\mathbf{u}(\tau) d\tau + \underbrace{\left[\int_{0}^t \widetilde{k}(t-\tau)\mathbf{u}(\tau) d\tau - \int_{0}^t k(t-\tau)\mathbf{u}(\tau) d\tau \right]}_{:=\mathbf{r}(t)} \nonumber \\
    &= A \mathbf{v}_0(t) - \sqrt{w_j} \mathbf{v}_j(t) + \mathbf{r}(t), \label{eq:interim_v0}
\end{align}
where we have defined $\mathbf{r}(t)$ as indicated in the second line, and substituted the definitions of $\widetilde{k}$ and $\mathbf{v}_j(t), \forall j \in [p]_0$ in the final line. The dynamics of $\mathbf{v}_j(t)$ for $j \in [p]$ are
\begin{equation}\label{eq:interim_vj}
    \frac{d \mathbf{v}_0(t)}{dt} = \sqrt{w_j} \mathbf{v}_0(t) - \gamma_j \mathbf{v}_j(t).
\end{equation}
Combining Eq.~\eqref{eq:interim_v0} and Eq.~\eqref{eq:interim_vj} for all $j \in [p]$, we obtain
$$
\frac{d \mathbf{v}}{dt} = \calA \mathbf{v}(t) + \begin{bmatrix}
    \mathbf{r}(t) \\ 0 \\ \vdots \\ 0
\end{bmatrix}, \quad \mathbf{v}(0) = \begin{bmatrix}
    \mathbf{u}_0 \\ 0 \\ \vdots \\ 0
\end{bmatrix}.
$$
The forcing term one needs to be add to ensure that $\mathbf{v}(t)$ is the exact solution corresponding to the kernel $\widetilde{k}$ is then $(\mathbf{r}(t),0,\ldots,0)^T$, which we denote as $\mathbf{f}(t)$. The analytical solution to the above $\ODE$ is
$$
\mathbf{v}(t) = e^{\calA t} \mathbf{v}(0) - \int_0^t e^{\calA (t-s)} \mathbf{f}(t) ds,
$$
which implies that the norm of $\mathbf{v}(t)$ is
\begin{equation}\label{eq:norm_vt}
\norm{\mathbf{v}(t)} \leq \expnorm(\calA) \norm{\mathbf{u}_0} + \expnorm(\calA) \int_0^t \norm{\mathbf{r}(s)}ds,
\end{equation}
where we have used the definition of $\expnorm(\calA):= \sup_{t \in [0,T]} \norm{e^{\calA t}}$, $\norm{\mathbf{f}(t)} = \norm{\mathbf{r}(t)}$ and $\norm{\mathbf{v}(0)} = \norm{\mathbf{u}(0)}$. Let us define the function
\begin{equation}\label{def:sup_vt}
    \mathbf{z}(t) := \sup_{s \in [0,t]} \norm{\mathbf{v}(s)},
\end{equation}
which we note is non-decreasing. We now bound norm of $\mathbf{r}(t)$ as
\begin{align}
    \norm{\mathbf{r}(t)} \leq \int_0^t \Big|k(t-\tau) - \widetilde{k}(t-\tau)\Big| \cdot \norm{\mathbf{u}(\tau)} d\tau 
    &\leq \mathbf{z}(t) \int_0^t \Big|k(t-\tau) - \widetilde{k}(t-\tau)\Big| \nonumber \\
    &\leq \mathbf{z}(t) \varepsilon_{\soe} t, \label{eq:bound_rt}
\end{align}
where we have used triangle inequality in the first inequality in the first line, $\norm{\mathbf{u}(\tau)} = \norm{\mathbf{v}_0(\tau)} \leq \norm{\mathbf{v}(\tau)} \leq \mathbf{z}(t), \forall \tau \in [0,t]$ from definition of $\mathbf{z}(t)$ (Eq.~\eqref{def:sup_vt}) in the second inequality in the first line, and used the fact that $\widetilde{k}$ has accuracy $\varepsilon_{\soe}$ (Definition~\ref{def:struct_kernels}). Substituting Eq.~\eqref{eq:bound_rt} into Eq.~\eqref{eq:norm_vt}, we have for every fixed $t' \in [0,t]$ that
\begin{align}
    \norm{\mathbf{v}(t')} &\leq \expnorm(\calA) \norm{\mathbf{u}_0} + \expnorm(\calA) \varepsilon_{\soe} \int_{0}^{t'}s \mathbf{z}(s) ds \\
    \implies \norm{\mathbf{v}(t')} &\leq \expnorm(\calA) \norm{\mathbf{u}_0} + \expnorm(\calA) \varepsilon_{\soe} \int_{0}^{t}s \mathbf{z}(s) ds,
\end{align}
where in the second line we used the fact that $\int_0^{t'} s \mathbf{z}(s) ds \leq \int_0^{t} s \mathbf{z}(s) ds$ as $t' \leq t$. As the above equation is true for all $t' \in [0,t]$, taking a supremum over all $t' \in [0,t]$ also gives us
\begin{equation}
    \mathbf{z}(t) \leq \expnorm(\calA) \norm{\mathbf{u}_0} + \expnorm(\calA) \varepsilon_{\soe} \int_{0}^{t}s \mathbf{z}(s) ds.
\end{equation}
Using Lemma~\ref{lem:gronwalls_inequality}, we obtain
\begin{equation}\label{eq:norm_final_z}
    \mathbf{z}(t) \leq \expnorm(\calA) e^{\varepsilon_{\soe} \expnorm(\calA) t^2/2} \norm{\mathbf{u}_0}.
\end{equation}
For any choice of $\varepsilon_{\soe} \leq 1/(\expnorm(\calA \cdot T^2))$ and using the definition of $\mathbf{z}(T)$ along with the fact that $\norm{\mathbf{u}(t)} \leq \norm{\mathbf{v}(t)}$, we then ensure that
$$
\sup_{t \in [0,T]} \norm{\mathbf{u}(t)} \leq 2 \expnorm(\calA) \norm{\mathbf{u}_0},
$$
which proves item $(i)$. Combining Eq.~\eqref{eq:bound_rt} with Eq.~\eqref{eq:norm_final_z} for $\varepsilon_{\soe} \leq 1/(\expnorm(\calA \cdot T^2))$ gives us item $(ii)$. This completes the proof.
\end{proof}
The above result then allows us to bound the error in the solution obtained from the $\SOE$ approximation compared to the exact solution of the $\VIDE$. 
\begin{claim} \label{claim:soe-error}
Let $\varepsilon \in (0,1)$. Consider the context of Theorem~\ref{thm:struct_kernels_hist_state}. Let $k$ be a kernel function as in Definition~\ref{def:struct_kernels} with accuracy $\varepsilon_{\soe} \in (0,1)$ and order $p:=\floor{p(\varepsilon_{\soe})}$. Let $\mathbf u$ be the solution to Eq.~\eqref{eq:og-VIDE} and $\tilde{\mathbf u}$ be the solution to Eq.~\eqref{eq:soe-VIDE}. For $\varepsilon_{\soe} \leq \varepsilon/(\expnorm(\calA)^2 \cdot T^2)$, we have 
$$
\|\mathbf{u}(t) - \widetilde{\mathbf{u}}(t) \| \leq \varepsilon \|\mathbf u_0\|, \quad \forall t \in [0,T].
$$
\end{claim}
\begin{proof}
Let us denote the error vector $\boldsymbol{\eta}(t) := \mathbf{u}(t) - \widetilde{\mathbf{u}}(t)$. Let $k(x)$ denote the true kernel function and let $\widetilde{k}(x) = \sum_{j=1}^p w_j e^{-\gamma_j x}$ be the approximation. The equation for the error vector is then
\begin{align}
\frac{d\boldsymbol{\eta}}{dt} &= A\boldsymbol{\eta} -  \int_0^t k(t-\tau) \mathbf u(\tau) \ d\tau + \int_0^t \widetilde{k}(t-\tau) \widetilde{\mathbf u}(\tau) \ d\tau \nonumber \\
&= A\boldsymbol{\eta} + \int_0^t \left [-k(t-\tau)\mathbf u(\tau) + \widetilde k(t-\tau){\mathbf u}(\tau) - \widetilde k(t-\tau){\mathbf u}(\tau) + \widetilde k(t-\tau) \widetilde{\mathbf u}(\tau) \right] \ d\tau \nonumber \\
&= A\boldsymbol{\eta} - \int_0^t \widetilde k(t-\tau) \boldsymbol{\eta}(\tau) \ d\tau  \underbrace{-\int_0^t \left [k(t-\tau)-\widetilde k(t-\tau) \right] \mathbf u (\tau) \ d\tau}_{:= \mathbf{r} (t)}, \label{eq:error-VIDE}
\end{align}
where $\mathbf{r}(t)$ is defined as indicated above. We now use Eq.~(\ref{eq:error-VIDE}) to bound $\|\boldsymbol{\eta}\|$. Using the same Markovianization technique as before and defining $\mathbf{z} = (\mathbf{z}_0,\ldots,\mathbf{z}_p)^T$ where
$$
\mathbf{z}_0(t) = \boldsymbol{\eta}(t), \quad \mathbf{z}_j(t) = \sqrt{w_j} \int_0^t e^{-\gamma_j(t-\tau)} \boldsymbol{\eta}(\tau) d\tau, \,\, \forall j \in [p],
$$
we convert Eq.~(\ref{eq:error-VIDE}) to a system of $\ODE$s:
\begin{equation} \label{eq:matrix-A}
\frac{d\mathbf{z}}{dt} = \calA \mathbf{z} + \underbrace{\begin{bmatrix}
        \mathbf{r}(t) \\
        0 \\
        \vdots \\
        0
    \end{bmatrix}}_{:= \mathbf{f}(t)},
\end{equation}
with $\mathbf{z}(0) = 0$ and $\mathbf{f}(t)$ defined as indicated above. The analytical solution to the above $\ODE$ is
\begin{equation}
    \mathbf{z}(t) = \int_0^t e^{\mathcal A (t-\tau)}\mathbf{f}(\tau) \ d\tau.
\end{equation}
Noting that $\norm{\boldsymbol{\eta}(t)} \leq \norm{\mathbf{z}(t)}$ and taking norms on both sides of the above equation, we obtain
\begin{align}
\|\boldsymbol{\eta}(t)\| \leq \|\mathbf{z}(t)\| \leq \int_0^t \left \|e^{\mathcal A(t-\tau)} \right\| \|\mathbf{f}(\tau)\| \ d\tau \leq \expnorm(\calA) \int_0^t \| \mathbf{r}(\tau)\| \ d\tau \leq \varepsilon_{\soe} t^2 \expnorm(\calA)^2 \norm{\mathbf{u}_0}
\end{align}
where we have used $\|\mathbf{r}(t)\| \leq 2 \varepsilon_{\soe} t \expnorm(\calA) \norm{\mathbf{u}_0}$ from \cref{claim:soe-approx-integral} for $\varepsilon_{\soe} \leq 1/(\expnorm(\calA) T)$ in the last inequality. Setting $\varepsilon_{\soe} \leq \varepsilon/(\expnorm(\calA)^2 T^2)$ ensures that
$$
\norm{\boldsymbol{\eta}(t)} \leq \varepsilon \norm{\mathbf{u}_0}, \quad \forall t \in [0,T],
$$
which is the desired result. This completes the proof.
\end{proof}

\paragraph{History state preparation.} We are now ready to show that the history state prepared by the quantum $\ODE$ solver (Theorem~\ref{thm:general-dynamics-algorithm-hist-state}) applied to the $\ODE$ (Eq.~\eqref{eq:markov-ode}) obtained after Markovianization is close to the true solution.
\begin{claim}\label{claim:markov-ode-hist-state}
Let $\varepsilon \in (0,1)$. Consider the context of Theorem~\ref{thm:struct_kernels_hist_state}. Let $k$ be a kernel function as in Definition~\ref{def:struct_kernels} with accuracy $\varepsilon_{\soe} \in (0,1)$ and order $p:=\ceil{p(\varepsilon_{\soe})}$. Let $\mathbf u$ be the solution to Eq.~\eqref{eq:og-VIDE}. Let $h = O(1/\norm{\calA})$ and $m = \ceil{T/h}$. Define $\ket{\Psi}$ to the true history state and $|\widehat{\Psi}\ra$ an approximation i.e.,
$$
\ket{\Psi} := \calZ^{-1} \sum_{j=0}^m \norm{\mathbf{u}(jh)} \ket{j}\ket{\mathbf{u}(jh)}, \quad |\widehat{\Psi}\ra := \widehat{\calZ}^{-1} \sum_{j=0}^m \norm{\widehat{\mathbf{u}}(jh)} \ket{j}\ket{\mathbf{u}(jh)},
$$
where $\calZ = \Big(\sum_{j=0}^m \norm{\mathbf{u}(jh)}^2 \Big)^{1/2}$ and $\widehat{\calZ} = \Big(\sum_{j=0}^m \norm{\widehat{\mathbf{u}}(jh)}^2 \Big)^{1/2}$. Then, there is an algorithm that uses Theorem~\ref{thm:ODE_solver_Krovi_history_state} with error $\upsilon$ to output $|\widehat{\Psi}\ra$ with probability $\geq \Omega\Big(1/(T \norm{\calA} \expnorm(\calA)^2)\Big)$ such that
$$
\|\ket{\Psi} - |\widehat{\Psi}\ra \| \leq \varepsilon,
$$
for 
$$
\varepsilon_{\soe} = O\left(\frac{\varepsilon}{T^{5/2}\expnorm(\calA)^2 \norm{\calA}^{1/2}}\right) \text{ and } \upsilon = O\left(\frac{\varepsilon}{T^{1/2} \norm{\calA}^{1/2} \expnorm(\calA)}\right).
$$
\end{claim}
\begin{proof}
Let $\mathbf{u}(t)$ be the true solution to Eq.~\eqref{eq:og-VIDE} and let $\widetilde{\mathbf{u}}(t)$ be the true solution to Eq.~\eqref{eq:soe-VIDE}. Let $\Psi$ and $\widetilde{\Psi}$ be the corresponding history state vectors i.e.,
\begin{equation}\label{eq:Psi_tildePsi_hist}
\Psi := \sum_{j=0}^m \ket{j} \otimes \mathbf{u}(t_j), \quad \widetilde{\Psi} := \sum_{j=0}^m \ket{j} \otimes \widetilde{\mathbf{u}}(t_j),
\end{equation}
where we have denoted $t_j := jh$. Let $\varepsilon_1 \in (0,1)$ be an error parameter to be fixed later and suppose we choose $\varepsilon_{\soe} \leq \varepsilon_1/(\expnorm(\calA)^2\cdot T^2)$. Then using the fact that the states corresponding to the time index are orthogonal, we have that
\begin{equation}\label{eq:error_Psi_tildePsi_hist}
\norm{\Psi - \widetilde{\Psi}}^2 = \sum_{j=0}^m \norm{\mathbf{u}(t_j) - \widetilde{\mathbf{u}}(t_j)}^2 \leq m \varepsilon_1^2 \norm{\mathbf{u}_0}^2,
\end{equation}
where we used Claim~\ref{claim:soe-error} in the last inequality. Let the normalized quantum states corresponding to $\Psi$ and $\widetilde{\Psi}$ be $\ket{\Psi}$ and $|\widetilde{\Psi}\ra$ respectively. Using Fact~\ref{lem:normalized_state_error} along with Eq.~\eqref{eq:error_Psi_tildePsi_hist} then gives us
\begin{equation}\label{eq:error_kets_Psi_tildePsi_hist}
\norm{\ket{\Psi} - |\widetilde{\Psi}\ra} \leq \frac{2 \norm{\Psi - \widetilde{\Psi}}}{\norm{\Psi}} \leq \frac{2 \varepsilon_1 \sqrt{m}\norm{\mathbf{u}_0}}{\norm{\mathbf{u}_0}} \leq 2 \varepsilon_1 \sqrt{m},
\end{equation}
where we have used the fact that $\norm{\Psi} \geq \norm{\mathbf{u}_0}$ in the second inequality. 

To solve for the $\VIDE$ in Eq.~\eqref{eq:soe-VIDE}, we solve for the corresponding $\ODE$ in Eq.~\eqref{eq:markov-ode} obtained after Markovianization using the quantum $\ODE$ solver of Theorem~\ref{thm:ODE_solver_Krovi_history_state}. Note that in Eq.~\eqref{eq:soe-VIDE}, we solve for the vector $\mathbf{y}(t)$ where $\mathbf{y}_0(t) = \widetilde{\mathbf{u}}(t)$. Let the history state vector be denoted as $\Phi$ and corresponding normalized state be $\ket{\Phi}$, which are defined as
$$
\Phi := \sum_{j=0}^m \ket{j} \otimes \mathbf{y}(t_j), \quad \ket{\Phi} := \Phi/\norm{\Phi}.
$$
Let the history state outputted by Theorem~\ref{thm:ODE_solver_Krovi_history_state} be denoted as $\ket{\widehat{\Phi}}$ and the corresponding unnormalized state vector be $\widehat{\Phi}$, defined as
$$
\widehat{\Phi} := \sum_{j=0}^m \ket{j} \otimes \widehat{\mathbf{y}}(t_j), \quad \ket{\widehat{\Phi}} := \widehat{\Phi}/\norm{\widehat{\Phi}},
$$
where $\widehat{y}(t_j)$ is the corresponding approximation to $\mathbf{y}(t_j)$ for all $j \in [m]_0$. Let $\varepsilon_2 \in (0,1)$ be the error parameter that we instantiate Theorem~\ref{thm:ODE_solver_Krovi_history_state} with, which we will fix later. Then, we have that
\begin{equation}\label{eq:error_kets_Phi_hatPhi}
\norm{\ket{\Phi} - |\widehat{\Phi}\ra} \leq \varepsilon_2.    
\end{equation}
Noting that we are after an approximation of $\mathbf{u}$ which corresponds to $\mathbf{y}_0 = \widetilde{u}$ and $\widehat{\mathbf{y}}_0 = \widehat{u}$, we define $\Pi_0$ as the projector that would select for this component. The ideal post-selection probability, which we denote by $p$, is 
\begin{equation}\label{eq:ideal_psucc}
    p = \frac{\norm{\widetilde{\Psi}}^2}{\norm{\Phi}^2} = \frac{\sum_{j=0}^m \norm{\widetilde{\mathbf{u}}(t_j)}^2}{\sum_{j=0}^m \norm{\mathbf{y}(t_j)}^2} \geq \frac{\norm{\mathbf{u}_0}^2}{(1 + m \expnorm(\calA)^2)\norm{\mathbf{u}_0}^2} = \frac{1}{1 + m \expnorm(\calA)^2}
\end{equation}
where we have used $\norm{\mathbf{y}(t_j)} \leq \expnorm(\calA) \norm{\mathbf{u}_0}, \forall j \in [p]$ (Claim~\ref{claim:expnorm-soe})\footnote{As we consider the homogeneous $\ODE$ in Eq.~\eqref{eq:markov-ode} with $\mathbf{b}=0$, we have that the analytical solution $\mathbf{y}(t) = \exp(\calA t) \mathbf{y}(0)$. Taking norms on both sides and then noting $\norm{\mathbf{y}(0)} = \norm{\mathbf{u}_0}$ gives us the claimed inequality.}, $\norm{\mathbf{y}(0)} = \norm{\mathbf{u}_0}$, and $\widetilde{\mathbf{u}}(0) = \mathbf{u}_0$ in the third inequality.

We now note that $|\widetilde{\Psi}\ra = \Pi_0 \ket{\Phi}/\norm{\Pi_0 \ket{\Phi}}$ and the target output history state $|\widehat{\Psi}\ra = \Pi_0 |\widehat{\Phi}\ra/\norm{\Pi_0 |\widehat{\Phi}\ra}$. Using Lemma~\ref{lem:extract_from_hist_state} along with the guarantee of Eq.~\eqref{eq:error_kets_Phi_hatPhi}, we then obtain
\begin{equation}
    \norm{|\widetilde{\Psi}\ra -  |\widehat{\Psi}\ra} \leq 2 \varepsilon_2/\sqrt{p} \leq 2 \sqrt{(1 + m \expnorm(\calA)^2)} \varepsilon_2,
\end{equation}
where we have used Eq.~\eqref{eq:ideal_psucc} in the last inequality. We can show that the output history state $|\widehat{\Psi}\ra$ is close to the true history state $\ket{\Psi}$ by evaluating
$$
\norm{\ket{\Psi} - |\widehat{\Psi}\ra} \leq \norm{\ket{\Psi} - |\widetilde{\Psi}\ra} + \norm{|\widetilde{\Psi}\ra - |\widehat{\Psi}\ra} \leq 2 \varepsilon_1 \sqrt{m} + 2\sqrt{(1 + m \expnorm(\calA)^2)} \varepsilon_2,
$$
where we have used Eqs.~\eqref{eq:error_kets_Phi_hatPhi},~\eqref{eq:error_kets_Psi_tildePsi_hist} in the second inequality. Setting $\varepsilon_1 = \varepsilon/(4 \sqrt{m})$ and $\varepsilon_2 = \varepsilon/(4\sqrt{1+m\expnorm(\calA)^2})$ gives us
$$
\norm{\ket{\Psi} - |\widehat{\Psi}\ra} \leq \varepsilon,
$$
which is the desired result. Note that the choice of $\varepsilon_1$ yields setting $\varepsilon_{\soe} = O\Big(\varepsilon/(T^{5/2}\expnorm(\calA)^2 \norm{\calA}^{1/2})\Big)$ where we used $m = \ceil{T/h}$ with $h = O(1/\norm{\calA})$. We can write the error that Theorem~\ref{thm:ODE_solver_Krovi_history_state} is instantiated with as $\varepsilon_2 = O\Big(\varepsilon/(T^{1/2} \norm{\calA}^{1/2} \expnorm(\calA))\Big)$. The corresponding success probability of preparing $|\widehat{\Psi}\ra$ using Lemma~\ref{lem:extract_from_hist_state}, which we denote by $p_{\mathrm{succ}}$ is then
$$
p_{\mathrm{succ}} \geq (\sqrt{p} - \varepsilon_2)^2 \geq \frac{9}{16(1+m \expnorm(\calA)^2)} = \Omega\Big(\frac{1}{T \norm{\calA} \expnorm(\calA)^2}\Big),
$$
where we have used $p \geq 1/(1+m \expnorm(\calA)^2)$ (Eq.~\eqref{eq:ideal_psucc}) and $\varepsilon_2 = \varepsilon/(4\sqrt{1+m \expnorm(\calA)^2})$ along with $\varepsilon \in (0,1)$, in the last inequality. This completes the proof.
\end{proof}

\paragraph{Final state preparation.} We can similarly show that the final state prepared by the quantum $\ODE$ solver (Theorem~\ref{thm:general-dynamics-algorithm-final-state}) applied to the $\ODE$ (Eq.~\eqref{eq:markov-ode}) obtained after Markovianization is close to the true solution.
\begin{claim}\label{claim:markov-ode-final-state}
Let $\varepsilon \in (0,1)$. Consider the context of Theorem~\ref{thm:struct_kernels_hist_state}. Let $k$ be a kernel function as in Definition~\ref{def:struct_kernels} with accuracy $\varepsilon_{\soe} \in (0,1)$ and order $p:=\ceil{p(\varepsilon_{\soe})}$. Let $\mathbf u$ be the solution to Eq.~\eqref{eq:og-VIDE}. Let $h = O(1/\norm{\calA})$ and $m = \ceil{T/h}$. Define $\ket{\mathbf{u}(T)}$ to the true final state and $|\widehat{\mathbf{u}}(T)\ra$ an approximation. Then, there is an algorithm that uses Theorem~\ref{thm:ODE_solver_Krovi_final_state} with error $\upsilon$ to output $|\widehat{\mathbf{u}}(T)\ra$ with probability $\geq \Omega\Big(1/(\expnorm(\calA)^2 g_0^2) \Big)$ such that
$$
\|\ket{\mathbf{u}(T)} - |\widehat{\mathbf{u}}(T)\ra \| \leq \varepsilon,
$$
for 
$$
\varepsilon_{\soe} = O\left(\frac{\varepsilon}{g_0 T^2 \expnorm(\calA)^2}\right) \text{ and } \upsilon = O\left(\frac{\varepsilon}{g_0 \expnorm(\calA)}\right), \text{ where } g_0 := \norm{\mathbf{u}_0}/\norm{\mathbf{u}(T)}.
$$
\end{claim}
\begin{proof}
We will proceed similarly to the proof of Claim~\ref{claim:markov-ode-hist-state}. Let $\mathbf{u}$ be the solution to Eq.~\eqref{eq:og-VIDE} and let $\widetilde{\mathbf{u}}$ be the solution to Eq.~\ref{eq:soe-VIDE}. Let us also define $\ket{\mathbf{u}(t)}:= \mathbf{u}(t)/\norm{\mathbf{u}(t)}$ and $|\widetilde{\mathbf{u}}(t)\ra = \widetilde{\mathbf{u}}/\norm{\widetilde{\mathbf{u}}}$ as the corresponding normalized states. Let $\varepsilon_1 \in (0,1)$ be an error parameter to be fixed later and let $\varepsilon_{\soe} = \varepsilon_1/(\expnorm(\calA)^2 T^2)$. Using Fact~\ref{lem:normalized_state_error} with Claim~\ref{claim:soe-approx-integral}$(i)$, we have
\begin{equation}
\norm{\ket{\mathbf{u}(T)} - |\widetilde{\mathbf{u}}(T)\ra} \leq \frac{2 \norm{\mathbf{u}(T) - \widetilde{\mathbf{u}}(T)}}{\norm{\mathbf{u}(T)}} \leq \frac{2 \varepsilon_1 \norm{\mathbf{u}_0}}{\norm{\mathbf{u}(T)}} = 2 \varepsilon_1 g_0,
\end{equation}
where we have used the definition of $g_0 = \norm{\mathbf{u}_0}/\norm{\mathbf{u}(T)}$. 

To solve for the $\VIDE$ in Eq.~\eqref{eq:soe-VIDE}, we solve for the corresponding $\ODE$ in Eq.~\eqref{eq:markov-ode} obtained after Markovianization using the quantum $\ODE$ solver of Theorem~\ref{thm:ODE_solver_Krovi_final_state}. Note that in Eq.~\eqref{eq:soe-VIDE}, we solve for the vector $\mathbf{y}(t)$ where $\mathbf{y}_0(t) = \widetilde{\mathbf{u}}(t)$. The true final state of Eq.~\eqref{eq:soe-VIDE} is $\mathbf{y}(T)$ with the corresponding normalized state being $\ket{\mathbf{y}(T)} := \mathbf{y}(T)/\norm{\mathbf{y}(T)}$. Let the final target state outputted by Theorem~\ref{thm:ODE_solver_Krovi_final_state} be $|\widehat{\mathbf{y}}(T)\ra$ (with the corresponding state vector being $\widehat{\mathbf{y}}(T)$). Let $\varepsilon_2 \in (0,1)$ be the error that we ran Theorem~\ref{thm:ODE_solver_Krovi_final_state} with. We then have the guarantee that
\begin{equation}\label{eq:error_kets_y_haty}
\norm{\ket{\mathbf{y}(T)} - |\widehat{\mathbf{y}}(T)\ra} \leq \varepsilon_2.    
\end{equation}
Note that we are after an approximation of $\mathbf{u}(T)$ which corresponds to $\mathbf{y}_0(T) = \widetilde{\mathbf{u}}(T)$ and $\widehat{\mathbf{y}}_0(T) = \widehat{\mathbf{u}}(T)$. As we had in the proof of Theorem~\ref{thm:struct_kernels_hist_state}, we define $\Pi_0$ as the projector that would select for this component. We then have $\ket{\widetilde{\mathbf{u}}(T)} := \Pi_0 \ket{\mathbf{y}(T)}/\norm{\Pi_0 \ket{\mathbf{y}(T)}}$ and $|\widehat{u}(T)\ra := \Pi_0 |\widehat{\mathbf{y}}(T) \ra/\norm{\Pi_0 |\widehat{\mathbf{y}}(T) \ra}$. The ideal post-selection probability, which we denote by $p$, is 
\begin{align}
    p = \frac{\norm{\widetilde{\mathbf{u}}(T)}^2}{\norm{\mathbf{y}(T)}^2} \geq  \frac{(\norm{\mathbf{u}(T)} - \norm{\mathbf{u}(T) - \mathbf{\widetilde{u}}(T)})^2}{\expnorm(\calA)^2 \norm{\mathbf{u}_0}^2} &\geq \frac{1}{\expnorm(\calA)^2}\left(\frac{\norm{\mathbf{u}(T)} - \varepsilon_1 \norm{\mathbf{u}_0}}{\norm{\mathbf{u}_0}} \right)^2 \nonumber \\
    &\geq \frac{1}{\expnorm(\calA)^2}\left(\frac{1}{g_0} - \varepsilon_1 \right)^2, \label{eq:ideal_psucc_final_state}
\end{align}
where we have used the reverse triangle inequality for the numerator in the first inequality, $\norm{\mathbf{y}(T)} \leq \expnorm(\calA) \norm{\mathbf{u}_0}$ (Claim~\ref{claim:expnorm-soe})\footnote{As we consider the homogeneous $\ODE$ in Eq.~\eqref{eq:markov-ode} with $\mathbf{b}=0$, we have that the analytical solution $\mathbf{y}(T) = \exp(\calA T) \mathbf{y}(0)$. Taking norms on both sides and then noting $\norm{\mathbf{y}(0)} = \norm{\mathbf{u}_0}$ gives us the claimed inequality.} for the denominator in the first inequality, used Claim~\ref{claim:soe-approx-integral}$(i)$ in the last inequality in the first line and definition of $g_0$ in the last line.

Using Lemma~\ref{lem:extract_from_hist_state} along with the guarantee of Eq.~\eqref{eq:error_kets_y_haty}, we then obtain
\begin{equation}\label{eq:error_tildeu_hatu}
    \norm{|\widetilde{\mathbf{u}}(T)\ra -  |\widehat{\mathbf{u}}(T)\ra} \leq 2 \varepsilon_2/\sqrt{p} \leq 2 \varepsilon_2 \expnorm(\calA)/(1/g_0 - \varepsilon_1),
\end{equation}
where we have used Eq.~\eqref{eq:ideal_psucc_final_state} in the last inequality. We can show that the output final state $|\widehat{\mathbf{u}}(T)\ra$ is close to the true history state $\ket{\mathbf{u}(T)}$ by evaluating
$$
\norm{\ket{\mathbf{u}(T)} - |\widehat{\mathbf{u}}(T)\ra} \leq \norm{\ket{\mathbf{u}(T)} - |\widetilde{\mathbf{u}}(T)\ra} + \norm{|\widetilde{\mathbf{u}}(T)\ra - |\widehat{\mathbf{u}}(T)\ra} \leq 2 \varepsilon_1 g_0 + 2 \varepsilon_2 \expnorm(\calA)/(1/g_0 - \varepsilon_1),
$$
where we have used Eqs.~\eqref{eq:error_kets_y_haty},~\eqref{eq:error_tildeu_hatu} in the second inequality. Setting $\varepsilon_1 = \varepsilon/(4 g_0)$ and $\varepsilon_2 = \varepsilon(1/g_0 - \varepsilon_1)/(4 \expnorm(\calA)) = 3 \varepsilon/(16 g_0 \expnorm(\calA))$ gives us
$$
\norm{\ket{\mathbf{u}(T)} - |\widehat{\mathbf{u}}(T)\ra} \leq  \varepsilon,
$$
which is the desired result. Note that the choice of $\varepsilon_1$ yields setting $\varepsilon_{\soe} = O\Big(\varepsilon/(g_0 T^2 \expnorm(\calA)^2)\Big)$. The success probability of preparing $|\widehat{\mathbf{u}}(T)\ra$ using Lemma~\ref{lem:extract_from_hist_state}, which we denote by $p_{\mathrm{succ}}$ is then
$$
p_{\mathrm{succ}} \geq (\sqrt{p} - \varepsilon_2)^2 \geq \left[\frac{3}{4\expnorm(\calA)}\left(\frac{1}{g_0} - \varepsilon_1\right)\right]^2 \geq \left[\frac{9}{16 \expnorm(\calA) g_0} \right]^2 = \Omega\Big(\frac{1}{\expnorm(\calA)^2 g_0^2} \Big),
$$
where we have used Eq.~\eqref{eq:ideal_psucc_final_state} and definition of $\varepsilon_2$ along with $\varepsilon \in (0,1)$ in the first inequality, and then the definition of $\varepsilon_1$ in the last inequality. This completes the proof.
\end{proof}

\subsubsection{Complexity and putting everything together}
We have so far commented on the correctness of our quantum algorithms for simulating $\VIDE$s with structured kernels (Definition~\ref{def:struct_kernels}) based on existing quantum $\ODE$ solvers to simulate the Markovianized system. We now characterize the complexity of the algorithms and thereby provide a proof of Theorem~\ref{thm:struct_kernels_hist_state}.

\paragraph{History state preparation.}
\begin{proof}[Proof of Theorem~\ref{thm:struct_kernels_hist_state}]
We use Theorem~\ref{thm:ODE_solver_Krovi_history_state} to solve Eq.~\eqref{eq:markov-ode} and specifically the algorithm as outlined in Claim~\ref{claim:markov-ode-hist-state}. It follows from the same claim that the prepared history state $|\widehat{\Psi}\ra$ satisfies
$$
\norm{\ket{\Psi} - |\widehat{\Psi}\ra} \leq \varepsilon.
$$
This required setting $\varepsilon_{\soe}$ and the error of Theorem~\ref{thm:ODE_solver_Krovi_history_state}, which we denote as $\upsilon$, as 
$$
\varepsilon_{\soe} = O\left(\frac{\varepsilon}{T^{5/2}\expnorm(\calA)^2 \norm{\calA}^{1/2}}\right) \text{ and } \upsilon = O\left(\frac{\varepsilon}{T^{1/2} \norm{\calA}^{1/2} \expnorm(\calA)}\right).
$$
The $\ODE$ solver of Theorem~\ref{thm:ODE_solver_Krovi_history_state} uses the following number of queries to the block-encoding of $\calA$:
\begin{equation}\label{eq:struct-kernels-queriesA-hist}
\widetilde{O}(\kappa^2 \poly(\log(1/\upsilon), \log(\kappa))),
\end{equation}
where
$$
\kappa = T \norm{\calA} \expnorm(\calA) \leq T (\alpha + p(\varepsilon_{\soe}) \sqrt{2\omega} + \gamma_{\mathrm{max}}) \expnorm_{\lp}(A),
$$
where we have used Claim~\ref{claim:norm-matrix-soe} and Claim~\ref{claim:expnorm-soe}. Noting the value of $\upsilon$, we then have from Eq.~\eqref{eq:struct-kernels-queriesA-hist} that the query complexity to block-encoding of $\calA$ is
$$
Q_\calA = \widetilde{O}(\kappa^2 \poly(\log(\kappa/\varepsilon)).
$$
Each call to the block-encoding of $\calA$ requires $1$ call to $A$ and $G_1 = \poly(\log p, 1/\varepsilon_A)$ additional gates. The resulting number of queries to the block-encoding of $A$ is then just $Q_{\calA}$. Each call to the initial state preparation of $\mathbf{y}(0)$ costs one preparation of $\mathbf{u}_0$ i.e., using $1$ call to $O_{\mathbf{u}}$. The number of calls from Theorem~\ref{thm:ODE_solver_Krovi_history_state} is
$$
Q_{\mathbf{u}} = \widetilde{O}(\kappa \log(\kappa/\varepsilon)).
$$
The additional gate complexity from using Theorem~\ref{thm:ODE_solver_Krovi_history_state} is
$$
G_2 = \widetilde{O}(\kappa^2 \log(\kappa/\varepsilon)).
$$
The total additional gate complexity is then $G_1 + G_2$. We now need to set the precision $\varepsilon_{\calA}$ of the block-encoding of $\calA$ as given in Theorem~\ref{thm:ODE_solver_Krovi_history_state} with the error $\upsilon$. Noting the relation of $\varepsilon_{\calA} = (\alpha + 3)\varepsilon_A$ (Claim~\ref{lem:BE-soe-matrix}), we then have to set $\varepsilon_A$ as
$$
\varepsilon_A = o(\alpha^{-1} \varepsilon \kappa^{-2} \poly(\log(\kappa/\varepsilon))^{-1})
$$
Finally, we have only so far discussed the cost of preparing the history state $\ket{\widehat{\Phi}}$ (see Claim~\ref{claim:markov-ode-hist-state} for notation) but need to extract the desired history state $\ket{\widehat{\Psi}}$ from it. The corresponding post-selection probability denoted by $p_{\mathrm{succ}}$ from Claim~\ref{claim:markov-ode-hist-state} is
$$
p_{\mathrm{succ}} = \Omega(1/(\kappa \expnorm_{\lp}(A))).
$$
Using amplitude amplification, we can boost the success probability to $\Omega(1)$ by using $O(\sqrt{\kappa \expnorm_{\lp}(\calA)})$ many calls to the circuit of Theorem~\ref{thm:ODE_solver_Krovi_history_state}. The resulting complexity is then $O(\sqrt{\kappa \expnorm_{\lp}(\calA)} Q_{\calA})$ calls to the block-encoding of $A$,  $O(\sqrt{\kappa \expnorm_{\lp}(\calA)}Q_{\mathbf{u}}$ calls to $O_{\mathbf{u}}$ and $O(\sqrt{\kappa \expnorm_{\lp}(\calA)}(G_1 + G_2))$ additional gate complexity. This completes the proof.
\end{proof}

\paragraph{Final state preparation.}
We can also prove a similar result for final state preparation, which we state below formally.
\begin{theorem}\label{thm:struct_kernels_final_state}
Let $\varepsilon \in (0,1)$, $T > 0$. Consider the context of Theorem~\ref{thm:struct_kernels_hist_state}. Let $\norm{\mathbf{u}(T)} > 0$ and define $g_0 := \norm{\mathbf{u}_0}/\norm{\mathbf{u}(T)}$. Then, there is a quantum algorithm that outputs $|\widehat{\mathbf{u}}(T)\ra$ with probability $\Omega(1)$ s.t. $\norm{\ket{\mathbf{u}(T)} - |\widehat{\mathbf{u}}(T)\ra} \leq \varepsilon$. Define 
$$
\kappa = T (\alpha + p(\varepsilon_{\soe}) \sqrt{2\omega} + \gamma_{\mathrm{max}}) \expnorm_{\lp}(A).
$$
The algorithm uses the following complexity
\begin{align*}
    \text{Queries to } U_A: & \,\, \widetilde{O}(g_0^2  \expnorm_{\lp}(A)^2 \kappa \, \poly(\log(g_0 \expnorm_{\lp}(A)/\varepsilon)) \\
    \text{Queries to } O_{\mathbf{u}}: & \,\, \widetilde{O}(g_0^2 \expnorm_{\lp}(A)^2 \kappa \log(g_0 \expnorm_{\lp}(A)/\varepsilon)) \\
    \text{Additional gate complexity }: & \,\, \widetilde{O}(g_0^2 \expnorm_{\lp}(A)^2 \kappa \log(g_0 \expnorm_{\lp}(A)/\varepsilon) + g_0 \expnorm_{\lp}(A) \poly(\log p, 1/\varepsilon_A)),
\end{align*}
where the block-encoding precision $\varepsilon_A$ and accuracy $\varepsilon_{\soe}$ are set as
$$
\varepsilon_A = o( \varepsilon \alpha^{-1} g_0^{-1} \expnorm_{\lp}(A)^{-1} \kappa^{-1} \poly(\log(g_0 \expnorm_{\lp}(A)/\varepsilon))^{-1}), \quad \varepsilon_{\soe} = O\left(\frac{\varepsilon}{g_0 T^2 \expnorm(\calA)^2}\right).
$$
\end{theorem}
\begin{proof}
We use Theorem~\ref{thm:ODE_solver_Krovi_final_state} to solve Eq.~\eqref{eq:markov-ode} and specifically the algorithm as outlined in Claim~\ref{claim:markov-ode-final-state}. It follows from the same claim that the prepared final state $|\widehat{\mathbf{u}}(T)\ra$ satisfies
$$
\norm{\ket{\mathbf{u}(T)} - |\widehat{\mathbf{u}}(T)\ra} \leq \varepsilon.
$$
This required setting $\varepsilon_{\soe}$ and the error of Theorem~\ref{thm:ODE_solver_Krovi_final_state}, which we denote as $\upsilon$, as 
$$
\varepsilon_{\soe} = O\left(\frac{\varepsilon}{g_0 T^2 \expnorm(\calA)^2}\right) \text{ and } \upsilon = O\left(\frac{\varepsilon}{g_0 \expnorm(\calA)}\right), \text{ where } g_0 := \norm{\mathbf{u}_0}/\norm{\mathbf{u}(T)}.
$$
The $\ODE$ solver of Theorem~\ref{thm:ODE_solver_Krovi_final_state} uses the following number of queries to the block-encoding of $\calA$:
\begin{equation}\label{eq:struct-kernels-queriesA-final}
\widetilde{O}(g \kappa \poly(\log(1/\upsilon), \log(\kappa))),
\end{equation}
where 
$$
g := \frac{\max_{t \in [0,T]} \norm{\mathbf{y}(T)}}{\norm{\mathbf{y}(T)}}, \quad \kappa = T \norm{\calA} \expnorm(\calA) \leq T (\alpha + p(\varepsilon_{\soe}) \sqrt{2\omega} + \gamma_{\mathrm{max}}) \expnorm_{\lp}(A),
$$
where we have used Claim~\ref{claim:norm-matrix-soe} and Claim~\ref{claim:expnorm-soe}. We can simplify $g$ as follows
$$
g := \frac{\max_{t \in [0,T]} \norm{\mathbf{y}(t)}}{\norm{\mathbf{y}(T)}} \leq \frac{\expnorm(\calA) \norm{\mathbf{u}_0}}{3 \norm{\mathbf{u}_0}/(4 g_0)} \leq 2 \expnorm_{\lp}(A) g_0,
$$
where we have used the fact that $\norm{\mathbf{y}(T)} \leq \expnorm(\calA) \norm{\mathbf{u}_0}$ (Claim~\ref{claim:expnorm-soe})\footnote{As we consider the homogeneous $\ODE$ in Eq.~\eqref{eq:markov-ode} with $\mathbf{b}=0$, we have that the analytical solution $\mathbf{y}(T) = \exp(\calA T) \mathbf{y}(0)$. Taking norms on both sides and then noting $\norm{\mathbf{y}(0)} = \norm{\mathbf{u}_0}$ gives us the claimed inequality.} and $\norm{\mathbf{y}(T)} \geq \norm{\mathbf{\widetilde{u}}(T)} \geq \norm{\mathbf{u}(T)} - \norm{\mathbf{u}(T) - \mathbf{\widetilde{u}}(T)} \geq 3 \norm{\mathbf{u}_0}/(4 g_0)$ (where the latter inequality follows from Eq.~\eqref{eq:ideal_psucc_final_state} in the proof of Claim~\ref{claim:markov-ode-final-state} and the definition of errors therein).
Noting the above simplifications and the value of $\upsilon$, we then have from Eq.~\eqref{eq:struct-kernels-queriesA-final} that the query complexity to block-encoding of $\calA$ is
$$
Q_\calA = \widetilde{O}(g_0  \expnorm_{\lp}(A) \kappa \, \poly(\log(g_0 \expnorm_{\lp}(A)/\varepsilon)).
$$
Each call to the block-encoding of $\calA$ requires $1$ call to $A$ and $G_1 = \poly(\log p, 1/\varepsilon_A)$ additional gates. The resulting number of queries to the block-encoding of $A$ is then just $Q_{\calA}$. Each call to the initial state preparation of $\mathbf{y}(0)$ costs one preparation of $\mathbf{u}_0$ i.e., using $1$ call to $O_{\mathbf{u}}$. The number of calls from Theorem~\ref{thm:ODE_solver_Krovi_final_state} is
$$
Q_{\mathbf{u}} = \widetilde{O}(g_0 \expnorm_{\lp}(A) \kappa \log(g_0 \expnorm_{\lp}(A)/\varepsilon)).
$$
The additional gate complexity from using Theorem~\ref{thm:ODE_solver_Krovi_final_state} is
$$
G_2 = \widetilde{O}(g_0 \expnorm_{\lp}(A) \kappa \log(g_0 \expnorm_{\lp}(A)/\varepsilon)).
$$
The total additional gate complexity is then $G_1 + G_2$. We now need to set the precision $\varepsilon_{\calA}$ of the block-encoding of $\calA$ as given in Theorem~\ref{thm:ODE_solver_Krovi_final_state} with the error $\upsilon$. Noting the relation of $\varepsilon_{\calA} = (\alpha + 3)\varepsilon_A$ (Claim~\ref{lem:BE-soe-matrix}), we then have to set $\varepsilon_A$ as
$$
\varepsilon_A = o( \varepsilon \alpha^{-1} g_0^{-1} \expnorm_{\lp}(A)^{-1} \kappa^{-1} \poly(\log(g_0 \expnorm_{\lp}(A)/\varepsilon))^{-1})
$$
Finally, we have only so far discussed the cost of preparing the history state $\ket{\widehat{\mathbf{y}}(T)}$ (see Claim~\ref{claim:markov-ode-final-state} for notation) but need to extract the desired state $\ket{\widehat{\mathbf{u}}(T)}$ from it. The corresponding post-selection probability denoted by $p_{\mathrm{succ}}$ from Claim~\ref{claim:markov-ode-final-state} is
$$
p_{\mathrm{succ}} = \Omega(1/(g_0^2 \expnorm_{\lp}(A)^2)).
$$
Using amplitude amplification, we can boost the success probability to $\Omega(1)$ by using $O(g_0 \expnorm_{\lp}(A))$ many calls to the circuit of Theorem~\ref{thm:ODE_solver_Krovi_history_state}. The resulting complexity is then $O(g_0 \expnorm_{\lp}(A) Q_{\calA})$ calls to the block-encoding of $A$,  $O(g_0 \expnorm_{\lp}(A) Q_{\mathbf{u}}$ calls to $O_{\mathbf{u}}$ and $O(g_0 \expnorm_{\lp}(A)(G_1 + G_2))$ additional gate complexity. This completes the proof.
\end{proof}

\subsubsection{Application to power-law kernel}
Let us consider the $\VIDE$ of Problem~\ref{prob:hist-state} with $K(x) = -k(x) I$ where
\begin{equation} \label{eq:power-law-kernel}
    k(x) = \frac{1}{x^\beta}, \ 0<\beta<1.
\end{equation}
We consider the homogeneous case, so $\mathbf b = 0$. Note that since $0<\beta<1$, the kernel is weakly singular at the origin. Furthermore, as the kernel does not decay to zero sufficiently rapidly, $\mathcal M = \infty$ regardless of $\mu(A)$. Despite these obstacles, we will show that it is indeed possible to develop an efficient quantum algorithm for a $\VIDE$ with such a kernel. Note that this $\VIDE$ is stable due to Claim~\ref{claim:expnorm-soe} even though $\mathcal M \nless 1$ (which prevents application of \cref{lem:stability}). 

\paragraph{Sum of exponentials approximation.} 
The power-law kernel of Eq.~\eqref{eq:power-law-kernel} can be approximated as a sum of exponentials ($\SOE$) as described below.

\begin{theorem}[\cite{li2010fast}, Theorem 2] \label{thm:Li}
Let $\delta > 0$ and $p \in \mathbb{R}$. For every $\varepsilon \in (0,1)$, there exists a set of $p$ many positive weights $\{w_j\}_{j \in [p]}$ and positive decay rates $\{\gamma_j\}_{j \in [p]}$ such that
$$
\left | \frac{1}{x^\beta} - \sum_{j=1}^p w_j e^{-\gamma_j x} \right | \leq \varepsilon,
$$
on $x\in[\delta,\infty)$ with $p = O \left ( \left (\log(1/\varepsilon) + \log(1/\delta) \right)^2 \right )$.
\end{theorem}
Note the above approximation is independent of the final simulation time $T$. We also have the following result regarding the weights and decay rates of the $\SOE$ approximation, which we will require later for characterizing the complexity of our quantum algorithm.

\begin{lemma}\label{lem:weight-decay-rate-sums}
Let $\{w_j\}$ and $\{\gamma_j\}$ be the weights and decay rates of the $\SOE$ approximation of $k(x) = x^{-\beta}, \beta \in (0,1)$ in~\cref{thm:Li}, respectively. Define
$\omega \coloneq \sum_{j=1}^p w_j$, and $\gamma_{\mathrm{max}} \coloneq \max_{i\in[p]}\gamma_i$. Then,
$$
\omega = O \left ( \delta^{-\beta}   (\log(1/\varepsilon) +  \log(1/\delta))^\beta\right ), \,\, \text{ and } \,\,
\gamma_{\mathrm{max}} = O\left (\delta^{-1} (\log(1/\varepsilon) + \log(1/\delta)) \right).
$$
\end{lemma}
\begin{proof}
We investigate the quadrature used in~\cref{thm:Li}. The function $x^{-\beta}$ is converted to an integral:
\begin{equation}
    \frac{1}{x^\beta} = \frac{1}{\beta \Gamma(\beta)} \int_0^\infty e^{-s^{\beta^{-1}}x} \ ds,
\end{equation}
where the Laplace transform is used in an intermediary step. To discretize the integral, its domain is initially truncated from $[0,\infty)$ to $[0,L]$, where in~\cite[Eq.~(17)]{li2010fast}, we have
\begin{equation}
    L = \delta^{-\beta} \left ( \log(3/\varepsilon) + \beta \log(1/\delta) \right )^{\beta}.
\end{equation}
The interval $[0,L]$ is then divided into dyadic sub-intervals
\begin{equation}
    \bigcup_{j=j_{\mathrm{min}}}^{j_{\mathrm{max}}} \left [2^j, 2^{j+1} \right],
\end{equation}
and the integral on each sub-interval is discretized using a Gauss-Legendre quadrature.

A well-known property of Gauss-Legendre quadrature is that the sum of the weights is the length of the integration interval~\cite{davis2007methods}. Summing the weights over all the dyadic intervals yields the length of the interval $\left [2^{j_{\mathrm{min}}},2^{j_{\mathrm{max}}+1}\right]$, so
\begin{align}
\sum_{j=1}^p w_p \leq\frac{1}{\beta \Gamma (\beta)} \left ( 2^{j_{\mathrm{max}}+1} - 2^{j_{\mathrm{min}}} \right) \leq \frac{L}{\beta \Gamma (\beta)} = O \left ( \delta^{-\beta}   (\log(1/\varepsilon) +  \log(1/\delta))^\beta\right ),
\end{align}
where, in the last step, we have used the fact that $1\leq\frac{1}{\beta \Gamma(\beta)} <1.13 $ for $\beta\in(0,1)$. The maximum decay rate is bounded as
\begin{equation}
    \gamma_{\mathrm{max}} \leq L^{\beta^{-1}} = O\left (\delta^{-1} (\log(1/\varepsilon) + \log(1/\delta)) \right),
\end{equation}
which completes the proof.
\end{proof}

We now show that the power-law kernel effectively satisfies Definition~\ref{def:struct_kernels}.
\begin{claim}\label{claim:power-law-is-structured}
Let $T > 0$, $\delta \in (0,T)$, and $\varepsilon \in (0,1)$. Let $k(x) = x^{-\beta}$ with $\beta \in (0,1)$. Let $\widetilde{k}(x) = \sum_{j=1}^p w_j e^{-\gamma_j x}$ be the approximation of $k(x)$ from Theorem~\ref{thm:Li}. Define
$$
\omega := \sum_{j=1}^p w_j, \quad \Delta := \frac{\delta^{1-\beta}}{1-\beta} + \delta \omega.
$$
Then,
$$
\int_{0}^T |k(x) - \widetilde{k}(x)| \ dx \leq \Big(\varepsilon + \frac{\Delta}{T}\Big) T.
$$
\end{claim}
\begin{proof}
We evaluate the relevant integral by splitting into two parts over $[0,\delta)$ and $[\delta,T]$ as follows
\begin{align*}
\int_{0}^T |k(x) - \widetilde{k}(x)| \ dx
&= \int_{\delta}^{T} \left | \frac{1}{x^\beta} - \int_{\delta}^{t} \sum_{j=1}^p w_j e^{-\gamma_jx} \right | \ dx +  \int_{0}^{\delta} \left | \frac{1}{x^\beta} - \int_{0}^{\delta} \sum_{j=1}^p w_j e^{-\gamma_jx}  \right | \ dx \\
&\leq  \int_{\delta}^{T} \left |\frac{1}{x^\beta}  - \sum_{j=1}^p w_j e^{-\gamma_jx} \right | \ dx + \int_{0}^\delta \frac{1}{x^\beta} dx + \int_{0}^\delta  \sum_{j=1}^p w_j e^{-\gamma_jx} \ dx \\
&\leq  \varepsilon T + \frac{\delta^{1-\beta}}{1-\beta} + \sum_{j=1}^p \frac{w_j}{\gamma_j} \left (1- e^{-\gamma_j \delta} \right ) \\
&\leq  \varepsilon T + \frac{\delta^{1-\beta}}{1-\beta} + \sum_{j=1}^p w_j \delta \\
&=  \Big(\varepsilon + \frac{\delta^{1-\beta}}{T(1-\beta)} + \frac{\delta \omega}{T} \Big) T,
\end{align*}
where we have noted that $w_j > 0, \forall j \in [p]$ in the second line, applied the promise of $\widetilde{k}(x)$ for $x \in [\delta,\infty)$ in the third line, used $1 - e^{-y} \leq y, \forall y \geq 0$ to show that $(1-e^{-\gamma_j \delta}) \leq \gamma_j \delta$ after noting that $\gamma_j > 0, \forall j \in [p]$. This gives us the desired result and completes the proof.
\end{proof}

\paragraph{History state preparation.}
We can now apply Theorem~\ref{thm:struct_kernels_hist_state} to prepare the history state corresponding to a $\VIDE$ with the power-law kernel. We state the corresponding result below.
\begin{restatable}{theorem}{powerhistory} \label{thm:power-law-alg-hist-state}
Let $\varepsilon \in (0,1)$. There exists a quantum algorithm that solves Problem~\ref{prob:hist-state} with $\mathbf b=0$, $K(x) = -x^{-\beta} I$ with $\beta \in (0,1)$. Define
\begin{align*}
\kappa &\coloneq O \left ( \alpha T + \frac{T}{1-\beta}\left ( \frac{\expnorm_{\lp}(A) \|\mathbf u_0\| \alpha^{1/2} T^{5/2}}{ \varepsilon}\right)^{\frac{1}{1-\beta}}  \right), \\
r &\coloneq \tilde O\left ( \frac{1}{(1-\beta)^{1+\beta/2}} \left (\frac{\expnorm_{\lp}(A) \|\mathbf u_0\| \alpha^{1/2}T^{5/2}}{\varepsilon}\right)^{\frac{\beta}{(1-\beta)}}T \right ).
\end{align*}
The algorithm uses the following complexity:
\begin{align*}
\text{Queries to } U_{A} &: \,\, \tilde O\left(\kappa^{5/2} r \expnorm_{\lp}(A) \cdot\poly(
        \log(\kappa/\varepsilon))  \right), \\
\text{Queries to } O_{\mathbf u}, O_{\mathbf b} &: \,\, \tilde O\left(\kappa^{3/2} r \expnorm_{\lp}(A)
        \log(\kappa/\varepsilon)\right), \\
\text{Additional gate complexity } &: \,\, \tilde O\left(\kappa^{5/2} \expnorm_{\lp}(A) r \cdot \poly\left(
        \log(\kappa/\varepsilon),\log \left (\frac{1}{1-\beta}\right ) \right)\right).
\end{align*}
The algorithm is successful provided the block-encoding precision is 
\begin{equation*}
    \varepsilon_A = o\left(\alpha^{-1} \varepsilon r^{-1}\kappa^{-2}\cdot \poly (
    \log(\kappa/\varepsilon) )^{-1} \right).
\end{equation*}
\end{restatable}
\begin{proof}
We will use the algorithm of \cref{thm:struct_kernels_hist_state}. Let $\varepsilon_1 \in (0,1)$ be an error parameter to be fixed later. From Claim~\ref{claim:power-law-is-structured}, we have that there exists $p \in \mathbb{N}$ such that there exists an approximation $\widetilde{k}$ to $k$ (Theorem~\ref{thm:Li}) such that
$$
\int_{0}^T |k(x) - \widetilde{k}(x)| \ dx \leq \underbrace{\Big(\varepsilon_1 + \frac{\Delta}{T}\Big)}_{:=\varepsilon_{\soe}} T, \text{ where } \Delta := \frac{\delta^{1-\beta}}{1-\beta} + \delta \omega,
$$
for some $\delta \in (0,T)$ and $\varepsilon_{\soe}$ is defined as indicated above. To simplify these expressions, let us first comment on the order of $\widetilde{k}$. Define 
$$
L:= O(\log(1/\varepsilon_1) + \log(1/\delta)).
$$
We have from Lemma~\ref{lem:weight-decay-rate-sums} that
$$
p(\varepsilon_1) = O(L^2), \quad \omega = O(\delta^{-\beta} L^\beta), \quad \gamma_{\max} = O(\delta^{-1} L).
$$
Moreover, we can write
$$
\varepsilon_{\soe} \leq \varepsilon_1 + \frac{\delta^{1-\beta}}{T} \left(\frac{1}{1-\beta} + L^\beta\right)
$$

The algorithm of Theorem~\ref{thm:struct_kernels_hist_state} (which in turn uses Theorem~\ref{thm:ODE_solver_Krovi_history_state}) then ensures that 
$$
\norm{\ket{\Psi} - |\widehat{\Psi}\ra} \leq \varepsilon.
$$
This requires setting $\varepsilon_{\soe}$ as 
$$
\varepsilon_{\soe} = O\left(\frac{\varepsilon}{T^{5/2}\expnorm(\calA)^2 \norm{\calA}^{1/2}}\right) \text{ and } \upsilon = O\left(\frac{\varepsilon}{T^{1/2} \norm{\calA}^{1/2} \expnorm(\calA)}\right).
$$
The above choice of $\varepsilon_{\soe}$ can be obtained by setting $\varepsilon_1 = \varepsilon_{\soe}/2$ and $\Delta = \varepsilon_{\soe}T/2$. 
We then have that
$$
\Delta = O\left( \delta^{1-\beta} \left[\frac{1}{1-\beta} + L^\beta \right]\right).
$$
One may then take 
$$
\delta = O\Big( (\varepsilon T)^{1/(1-\beta)} \Big), \quad p = \poly(\log(1/(\varepsilon T)))
$$
Plugging in values from above into the guarantees of Theorem~\ref{thm:struct_kernels_hist_state} then gives us the desired result. This completes the proof.
\end{proof}

\paragraph{Final state preparation.}
We also have the following result for final state preparation.
\begin{restatable}{corollary}{powerfinal}\label{thm:power-law-alg-final-state}
Let $\varepsilon \in (0,1)$. Define $g_0 := \norm{\mathbf{u}_0}/\norm{\mathbf{u}(T)}$ and let $\norm{\mathbf{u}(T)} > 0$. There exists a quantum algorithm that solves Problem~\ref{prob:final-state} with $\mathbf b=0$, $K(x) = -x^{-\beta} I$ for $0<\beta<1$ with the following complexity:
\begin{align*}
\text{Queries to } U_{A} &: \,\, \tilde O\Big(g_0^2 \expnorm_{\lp}(A) r \kappa \cdot \poly(
        \log(g_0 r \expnorm_{\lp}(A)/\varepsilon),\log(\kappa)) \Big), \\
\text{Queries to } O_{\mathbf u} &: \,\, \tilde O\Big( g_0^2 r^2 \expnorm_{\lp}(A)^2 \kappa \cdot\log(g_0 r \expnorm_{\lp}(A)/\varepsilon)\Big), \\
\text{Additional gate complexity } &: \,\, \widetilde O \Big ( g_0^2 r^2 \expnorm_{\lp}(A)^2 \kappa \cdot \poly(
\log(g_0 r \expnorm_{\lp}(A)/\varepsilon)) \Big ),
\end{align*}
where
\begin{align*}
\kappa &\coloneq O \left (\alpha T + \frac{T}{1-\beta} \left (\frac{g_0 T}{\varepsilon}\right )^{\frac{1}{1-\beta}} \right) \\
r &\coloneq \widetilde O\left ( \frac{1}{(1- \beta)^{1+\beta/2}} \left (\frac{g_0 T}{\varepsilon} \right )^{\frac{1}{2(1-\beta)}} T \right ).
\end{align*}
The algorithm is successful provided the block-encoding precision is 
$$
\varepsilon_A = o\left(\varepsilon (g_0 r \expnorm_{\lp}(A) \kappa)^{-1}\cdot \poly(
    \log(g_0 r \expnorm_{\lp}(A)/\varepsilon),\log(\kappa))^{-1}  \right).
$$
\end{restatable}

%%%%%%%%%%%%%%%%%%%%%%%%%%%%%%%%%%%%%%%%%%%%%%%%%%%%%%%%%%%%%%%%%%%%%%%%%%%%%%%%%
\section{Application to the Mori-Zwanzig formalism} \label{sec:applications}
In this section, we will describe an application of the general-kernel quantum algorithm that we developed in Section~\ref{sec:general_kernels}. Suppose we have a differential equation of the form
\begin{equation}
    \frac{d\mathbf{g}}{dt} = \mathbf{f}(\mathbf{g},t).
\end{equation}
Suppose $\mathbf{g}$ can be decomposed as $\mathbf{g}_M \oplus \mathbf{g}_{\bar M}$, where $\mathbf{g}_M$ contains the set of variables that we are interested in -- typically referred to as the \textit{resolved variables} -- and $\mathbf{g}_{\bar M}$ contains variables that we are not interested in -- usually referred to as the \textit{unresolved variables} -- but whose presence affects the dynamics of $\mathbf{g}_M$ nontrivially.

The goal is to only probe the dynamics of $\mathbf{g}_M$ without having to worry about the dynamics of $\mathbf{g}_{\bar M}$. This approach is advantageous for two reasons. First, in many high-dimensional systems, there are only a handful of variables that we may actually be interested in studying. An example is high-dimensional microscopic dynamics in nonequilibrium statistical mechanics~\cite{zwanzig2001nonequilibrium}. A more relevant example is system-bath coupling in quantum mechanics, where we are interested in the system dynamics but cannot ignore bath effects~\cite{breuer2002theory}. Second, by reducing the dimension of the system, we can investigate the dynamics of the resolved variables more clearly, either through analytics or simulations. For general dynamical systems, the Mori-Zwanzig formalism was constructed as a tool to coarse-grain the dynamics of $\mathbf{g}$ by only keeping the resolved variables~\cite{mori1965transport,zwanzig1973nonlinear,gouasmi2017priori}. In quantum mechanics, a similar approach, know as the Nakajima-Zwanzig formalism, captures this idea~\cite{nakajima1958quantum,zwanzig1960ensemble}, from which the Lindblad master equation~\cite{breuer2002theory} can be derived as a Markovian approximation~\cite{lidar2019lecture,gonzalez2024tutorial}. Similar ``coarse-graining" formalisms have found applications in other areas as well, such as dynamic mode decomposition~\cite{lin2021data}.

In this section we will derive the Mori-Zwanzig framework for linear $\ODE$s, and show that coarse-graining turns the original linear $\ODE$ into a linear convolution $\VIDE$. We will then use the algorithm of Section~\ref{sec:general_kernels} to simulate the system. We also briefly analyze the complexity of such our algorithms.
\subsection{Overview of the Mori-Zwanzig formalism}
Suppose we have a homogeneous linear $\ODE$ of the form
\begin{equation}
    \frac{d\mathbf g}{dt} = L\mathbf g + \mathbf b,
\end{equation}
where
\begin{equation}
    \mathbf g = \begin{bmatrix}
        \mathbf g_{ M}(t) \\
        \mathbf g_{\bar{M}}(t)
    \end{bmatrix}, \ L = \begin{bmatrix}
        L_{MM} & L_{M \bar{M}} \\
        L_{\bar{M} M} & L_{\bar{M} \bar{M}} 
    \end{bmatrix}, \ \mathbf{b} = \begin{bmatrix}
        \mathbf{b}_M \\
        \mathbf b_{\bar M}
    \end{bmatrix}.
\end{equation}
Suppose we are only interested in the dynamics of $\mathbf g_{M}$. In other words, we would like to get a closed-form expression for the time evolution of $\mathbf g_{M}$. Consider the equation for $\mathbf g_{\bar M}$:

\begin{equation}
    \frac{d\mathbf g_{\bar M}}{dt} = L_{\bar M \bar M} \mathbf g_{\bar M} + L_{\bar M M}\mathbf{g}_{M} + \mathbf{b}_{\bar M}.
\end{equation}
Treating the last two terms as inhomogeneities, we can write
\begin{equation}
    \mathbf g_{\bar M} = e^{L_{\bar M \bar M}t}\mathbf{g}_{\bar M}(0) + \int_0^t e^{L_{\bar M \bar M}(t-\tau)} L_{\bar M M} \mathbf{b}_{\bar M} \ d\tau + \int_0^t e^{L_{\bar M \bar M}(t-\tau)}L_{\bar M M} \mathbf{g}_{M}(\tau) \ d\tau.
\end{equation}

We substitute this in the equation for $\mathbf{g}_{M}$ to get:
\begin{equation}
    \begin{aligned}
        \frac{d\mathbf g_{M}}{dt} &= L_{MM} \mathbf g_{M}+\mathbf{b}_M \\
        &+ L_{M \bar{M}} \left [e^{L_{\bar M \bar M}t}\mathbf{g}_{\bar M}(0) + \int_0^t e^{L_{\bar M \bar M}(t-\tau)} L_{\bar M M} \mathbf{b}_{\bar M} \ d\tau + \int_0^t e^{L_{\bar M \bar M}(t-\tau)}L_{\bar M M} \mathbf{g}_{M}(\tau) \ d\tau \right ].
    \end{aligned}
\end{equation}

Let
\begin{align}
    A &\coloneq L_{MM} \\
    K(x) &\coloneq L_{M \bar{M}}e^{L_{\bar{M} \bar{M}}x}L_{\bar{M}M} \label{eq:lcg-kernel} \\
    \mathbf c(t) &\coloneq \mathbf{b}_M + L_{M \bar{ M}}e^{L_{\bar{M} \bar{M}}t} \mathbf g_{\bar{M}}(0) + L_{M\bar M}\int_0^t e^{L_{\bar M \bar M}(t-\tau)}L_{\bar M M} \mathbf b_{\bar M} \ d\tau  \\
    \mathbf u &\coloneq \mathbf g_M.
\end{align}
This gives us the linear $\VIDE$:
\begin{equation} \label{eq:MZ-vide}
    \frac{d\mathbf u}{dt} = A\mathbf{u}+\int_0^t K(t-\tau) \mathbf u(\tau) \ d\tau + \mathbf{c}(t).
\end{equation}

So, we have effectively removed the dynamics of $\mathbf g_{\bar M}$. The intuition behind this approach is demonstrated in Figure~\ref{fig:MZ}.

\begin{figure}[h!]
    \centering    
    \includegraphics[width=0.65\linewidth]{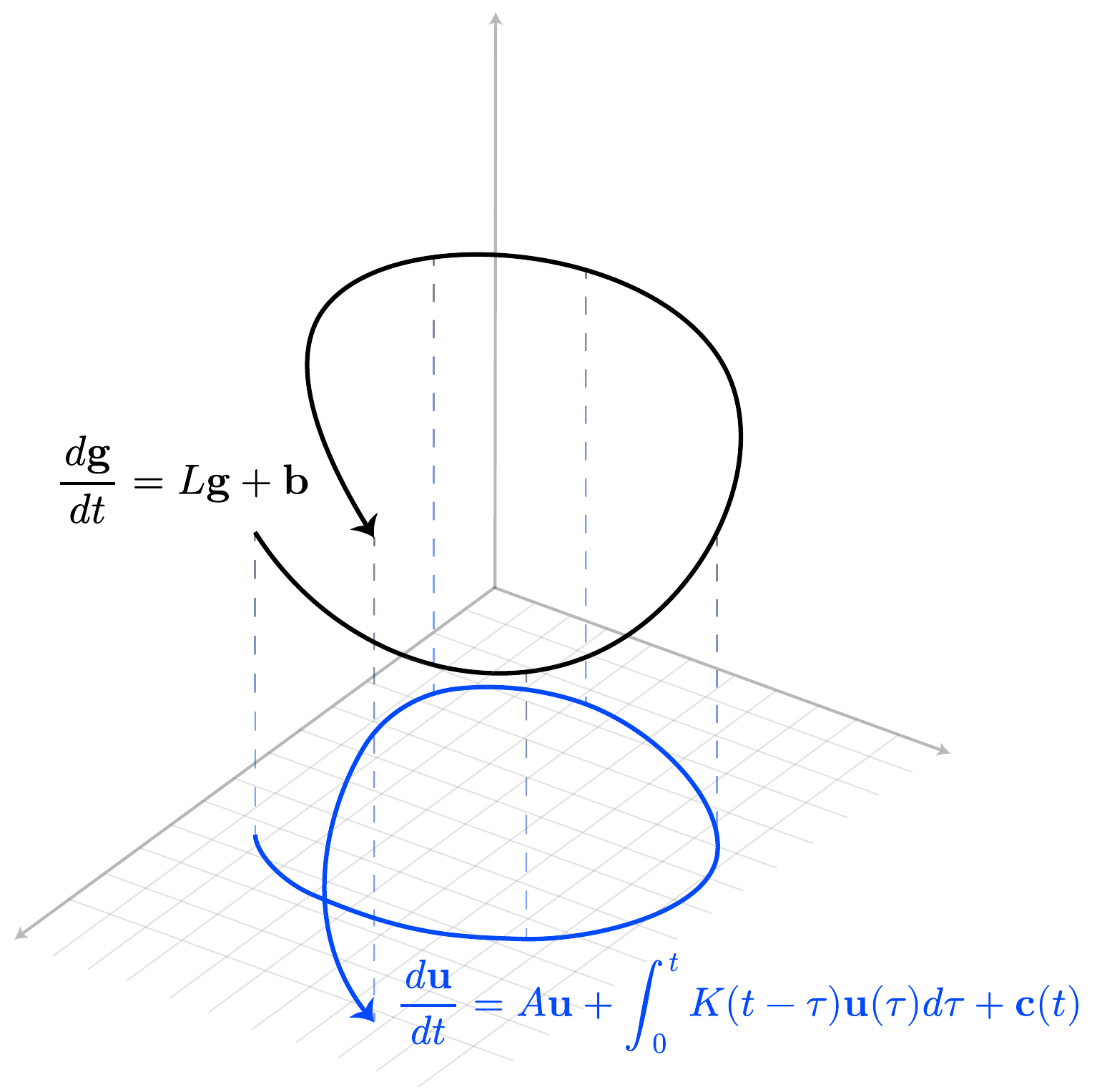}
    \caption{Illustration of the linear coarse-graining formalism for a three-dimensional dynamical system. While the full dynamics is Markovian and described by a linear $\ODE$ system, the projected dynamics is non-Markovian and described by a linear $\VIDE$.}
    \label{fig:MZ}
\end{figure}

In this paper, we consider a subset of the linear coarse-graining problem, where $\mathbf{g}_{\bar M}(0) = \mathbf{b}_{\bar M} = 0$. This simplifies Eq.~\eqref{eq:MZ-vide} by making the inhomogeneous term time-independent, resulting in
\begin{equation} \label{eq:MZ-vide-subset}
    \frac{d\mathbf u}{dt} = A\mathbf{u}+\int_0^t K(t-\tau) \mathbf u(\tau) \ d\tau + \mathbf{b}_M, \ \mathbf{u}(0) = \mathbf{g}_M(0).
\end{equation}
The main challenge in the linear (and nonlinear) Mori-Zwanzig formalism is that we do not have direct access to the dynamics of the unresolved variables. In other words, we may not know of the specific structure of $L_{M\bar M}$, $L_{\bar M M}$, and $L_{\bar M \bar M}$ and, hence, we do not have direct access to $K(x)$ in Eq.~\eqref{eq:lcg-kernel}. Typically, $K(x)$ is replaced by some other heuristic kernel $\tilde K(x)$ that approximates $K(x)$ sufficiently well:
\begin{equation}
    \|K(x) - \tilde K(x) \| \leq \delta \|K(x)\|, \ \forall x\geq 0,
\end{equation}
where $\delta$ is the relative error, which we assume to have control over. This gives us an approximate $\VIDE$
\begin{equation} \label{eq:MZ-vide-subset-approx}
    \frac{d\tilde{\mathbf u}}{dt} = A\tilde{\mathbf u} + \int_0^t \tilde K(t-\tau)  \tilde{\mathbf u}(\tau) \ d\tau + \mathbf{b}_M, \ \tilde{\mathbf{u}}(0) = \mathbf{g}_M(0).
\end{equation}

Then, instead of simulating the $\VIDE$ of Eq.~\eqref{eq:MZ-vide-subset}, we simulate Eq.~\eqref{eq:MZ-vide-subset-approx}. 

\subsection{Quantum algorithm for the linear Mori-Zwanzig formalism}
We use the quantum algorithm in Section~\ref{sec:general_kernels} to simulate Eq.~\eqref{eq:MZ-vide-subset-approx}. Our algorithm requires the short-term memory condition to be satisfied. We can find this regime for Eq.~\eqref{eq:MZ-vide-subset} from the following Claim.
\begin{claim}[Short-term memory condition for linear coarse-graining]
    Let
    \begin{align}
        \mu_{M} &\coloneq \mu(L_{MM}) \\
        \mu_{\bar{M}} &\coloneq \mu(L_{\bar{M} \bar{M}}).
    \end{align}
    Let $\mathcal M$ be the memory strength associated with Eq.~\eqref{eq:MZ-vide-subset}. The linear coarse-graining formalism of Eq.~\eqref{eq:MZ-vide-subset} satisfies the short-term memory condition $\mathcal M < 1$ if $\mu_{M}<0$, $\mu_{\bar{M}} < 0$, and
    \begin{equation}
        \frac{\|L_{M \bar{M}}\| \|L_{\bar{M}M}\|}{|\mu_{M} \mu_{\bar{M}}|} < 1.
    \end{equation}
\end{claim}
\begin{proof}
The memory strength is bounded as follows:
\begin{align*}
\mathcal M = \frac{1}{|\mu(A)|}\int_0^\infty \|K(x)\| \ dx 
&= \frac{1}{|\mu_{M}|}\int_{0}^\infty \left \|L_{M\bar M} e^{L_{\bar M \bar M}x}L_{\bar M M} \right \| \ dx \\
    &\leq \frac{\|L_{M\bar M}\| \|L_{\bar M M}\|}{|\mu_{M}|}\int_{0}^\infty \left \| e^{L_{\bar M \bar M}x} \right \| \ dx \\
    &\leq \frac{\|L_{M\bar M}\| \|L_{\bar M M}\|}{|\mu_{M}|}\int_{0}^\infty  e^{-|\mu_{\bar M}| x} \ dx \\
    &= \frac{\|L_{M \bar{M}}\| \|L_{\bar{M}M}\|}{|\mu_{M} \mu_{\bar{M}}|}.
\end{align*}
This completes the proof.
\end{proof}
Since we have control over the error in approximating the kernel, we can guarantee that the approximated system also satisfies the short-term memory condition:
\begin{claim}
Suppose the $\VIDE$ in Eq.~\eqref{eq:MZ-vide-subset} satisfies $\mathcal{M}<1$. Let $\tilde{\mathcal M}$ be the memory strength of Eq.~\eqref{eq:MZ-vide-subset-approx}. If $\delta < (1 - \calM)/\calM$, then $\tilde{\mathcal M} < 1$.
\end{claim}
\begin{proof}
Note that $\mu(A) = \mu_M$. Then,
\begin{align}
\tilde{\mathcal{M}} = \frac{1}{|\mu_M|}\int_0^\infty \|\tilde K(x)\| \ dx 
&\leq \frac{1}{|\mu_M|} \left [ \int_0^\infty \|K(x)-\tilde K(x)\| \ dx + \int_0^\infty \|K(x)\| \ dx \right ] \\
&\leq \frac{1+\delta}{|\mu_M|}\int_0^\infty \|K(x)\| \ dx \\
&\leq (1+\delta)\mathcal M.
\end{align}
For any choice of $\delta < \frac{1-\mathcal M}{\mathcal M}$, we then ensure that $\tilde{\mathcal M} < 1$.
\end{proof}
We also need to bound the error from approximating Eq.~\eqref{eq:MZ-vide-subset} with Eq.~\eqref{eq:MZ-vide-subset-approx}, which we do as follows.
\begin{claim} \label{lem:mz-error} 
Consider Eq.~\eqref{eq:MZ-vide-subset} and Eq.~\eqref{eq:MZ-vide-subset-approx}. Define the error $\boldsymbol{\eta} \coloneq \mathbf u - \tilde{\mathbf u}$. Then, 
$$
\|\boldsymbol{\eta}\| \leq \frac{u_{\mathrm{max}} \mathcal M}{1-(1+\delta)\mathcal M} \delta,
$$
where $u_{\mathrm{max}}$ is defined as in~\cref{lem:stability}.
\end{claim}
\begin{proof}
The equation for the error is
\begin{align}
    \frac{d\boldsymbol{\eta}}{dt} &= A\boldsymbol{\eta} + \int_0^t \left [K(t-\tau) \mathbf{u}(\tau) - \tilde K(t-\tau) \tilde{\mathbf{u}}(\tau) \right ] \ d\tau \\
    &= A\boldsymbol{\eta} + \int_0^t \tilde K(t-\tau) \boldsymbol{\eta}(\tau) \ d\tau + \tilde{\mathbf{b}}(t),
\end{align}
where
\begin{equation}
    \tilde{\mathbf{b}}(t) = \int_0^t [K(t-\tau) - \tilde K(t-\tau)] {\mathbf{u}}(\tau ) \ d\tau.
\end{equation}
Note that
\begin{align}
\|\tilde{\mathbf{b}}(t)\| \leq \int_0^t \|K(t-\tau) - \tilde K(t-\tau)\| \|\mathbf{u}(\tau)\| \ d\tau \leq \delta {u}_{\mathrm{max}} \int_0^\infty \|K(x)\| dx = \delta {u}_{\mathrm{max}} |\mu_M|\mathcal M.
\end{align}
Then, using~\cref{lem:stability}, we get
\begin{equation}
    \|\boldsymbol{\eta}\| \leq \frac{{u}_{\mathrm{max}} \mathcal M}{1-(1+\delta)\mathcal M} \delta.
\end{equation}
\end{proof}

We are now ready to use the algorithm in Section~\ref{sec:general_kernels}. 
\begin{corollary} [History state preparation, LCG formalism] \label{thm:lcg-history}
There exists a quantum algorithm that solves Problem~\ref{prob:hist-state} in the context of the linear Mori-Zwanzig formalism with
\begin{align*}
\text{Queries to } U_A: & \,\, O\left (\frac{(\Lambda + \Xi T)^2 T^4}{\mathfrak g^2} \frac{\log(1/\varepsilon)}{\varepsilon^2} \right ) \\
\text{Queries to } \{U_{K_j}\}_j: & \,\, O\left (\frac{(\Lambda + \Xi T)^4 T^8}{\mathfrak g^4} \frac{\log(1/\varepsilon)}{\varepsilon^4} \right) \\
\text{Queries to } O_{\mathbf u}, O_{\mathbf b}: & \,\, O \left (\frac{(\Lambda + \Xi T)^2 T^4}{\mathfrak g^2} \frac{\log(1/\varepsilon)}{\varepsilon^2}  \right ),
\end{align*}
and additional gate complexity that is larger by a factor of
$$
O\left ( \poly(\log T, \log(\Lambda + \Xi T),\log(1/\mathfrak g),\log(1/\varepsilon),\log\log(1/\varepsilon))\right )
$$
with $\Lambda$ and $\Xi$ being defined in Lemmas~\ref{lem:bounded-derivative} and~\ref{lem:bounded-double-derivative}, respectively. Here,
$$
\mathfrak g(T) \coloneq \left [\frac{1}{T}\int_0^T \|\mathbf u(\tau)\|^2 \ d\tau \right ]^{1/2}.
$$
The algorithm is successful provided the block-encoding precision is
$$
\delta = o \left (\frac{\mathfrak g^2 \varepsilon^3}{T^4(\Lambda+\Xi T)^2\log(1/\varepsilon)}  \right ).
$$
\end{corollary}
\begin{proof}
    We set the error of the history state preparation algorithm in \cref{thm:general-dynamics-algorithm-hist-state} to be $\varepsilon/2$ and set
    \begin{equation}
        \delta = \frac{1-\mathcal M}{\mathcal M}\frac{\varepsilon}{2u_{\mathrm{max}}+\varepsilon}
    \end{equation}
    to ensure, from~\cref{lem:mz-error}, that $\|\boldsymbol{\eta}\|\leq \varepsilon/2$. This ensures that the output of the quantum algorithm is $\varepsilon$-close to the history state. The rest of the analysis follows from \cref{thm:general-dynamics-algorithm-hist-state}.
\end{proof}

\begin{corollary} [Final state preparation, LCG formalism]
There exists a quantum algorithm that solves Problem~\ref{prob:final-state} in the context of the linear Mori-Zwanzig formalism with the following complexity
\begin{align*}
\text{Queries to } U_A: & \,\, O\left ( g\frac{(\Lambda + \Xi T) T^2 }{q \varepsilon} \log(1/\varepsilon)\right ) \\
\text{Queries to } \{U_{K_j}\}_j: & \,\, O \left (g\frac{ (\Lambda + \Xi T)^2 T^4}{q^2 \varepsilon^2} \log(1/\varepsilon) \right ) \\
\text{Queries to } O_{\mathbf u}, O_{\mathbf b}: & \,\, O \left ( g\frac{(\Lambda  + \Xi T)T^2}{q\varepsilon} \log(1/\varepsilon)  \right ),
\end{align*}
and gate complexity that is larger by a factor of
$$
O\left ( \poly(\log T,\log(\Lambda + \Xi T),\log(1/q),\log(1/\varepsilon),\log\log(1/\varepsilon))\right ),
$$
with $\Lambda$ and $\Xi$ being defined in Lemmas~\ref{lem:bounded-derivative} and~\ref{lem:bounded-double-derivative}, respectively. Here,
$$
q(T) \coloneq \|\mathbf u(T)\|, \text{ and } g(T) \coloneq \frac{\max_{t\in[0,T]}\|\mathbf u(t)\|}{\|\mathbf u(T)\|}.
$$
The algorithm is successful provided the block-encoding precision is
$$
\delta = o \left (\frac{q \varepsilon^2}{T^2(\Lambda+\Xi T)\log(1/\varepsilon)}  \right ).
$$
\end{corollary}
\begin{proof}
    Similar to \cref{thm:lcg-history}, we set the error of the final state preparation algorithm in \cref{thm:general-dynamics-algorithm-final-state} to be $\varepsilon/2$ and set
    \begin{equation}
        \delta = \frac{1-\mathcal M}{\mathcal M}\frac{\varepsilon}{2u_{\mathrm{max}}+\varepsilon}
    \end{equation}
    to ensure, from~\cref{lem:mz-error}, that $\|\boldsymbol{\eta}\|\leq \varepsilon/2$. This ensures that the output of the quantum algorithm is $\varepsilon$-close to the final state. The rest of the analysis follows from \cref{thm:general-dynamics-algorithm-final-state}.
\end{proof}

%%%%%%%%%%%%%%%%%%%%%%%%%%%%%%%%%%%%%%%%%%%%%%%%%%%%%%%%%%%%%%%

\bibliographystyle{alphaurl}
\bibliography{references}
%%%%%%%%%%%%%%%%%%%%%%%%%%%%%%%%%%%%%%%%%%%%%%%%%%%%%%%%%%%%%%%

\end{document}